\documentclass[11pt]{article}
\usepackage[margin=1in]{geometry}
\usepackage[english]{babel}
\usepackage{xr}
\usepackage[font={footnotesize}, labelfont={footnotesize, bf}]{caption}

\usepackage{setspace}

\usepackage{tikz}
\usetikzlibrary{arrows}

\usepackage{mathrsfs,hyperref, algorithm}
\usepackage[]{amsmath}
\usepackage{amsthm}
\usepackage{fix-cm}
\usepackage[]{amssymb}
\usepackage[]{latexsym}
\usepackage[right]{eurosym}
\usepackage[T1]{fontenc}
\usepackage[]{graphicx}
\usepackage[]{epsfig}
\usepackage{fancyhdr}
\usepackage{pstricks}
\usepackage{multirow}
\usepackage{natbib}
\usepackage{subcaption}

\numberwithin{equation}{section}

\makeatletter

\DeclareMathSymbol{\shortminus}{\mathbin}{AMSa}{"39}

\newcommand{\bbr}{\mathbb{R}}

\newcommand{\bbe}{\mathbb{E}}
\newcommand{\bbn}{\mathbb{N}}

\newcommand{\bbp}{\mathbb{P}}

\newcommand{\bbt}{\mathbb{T}}

\newcommand{\fcal}{\mathcal{F}}

\newcommand{\ncal}{\mathcal{N}}

\newcommand{\acal}{\mathcal{A}}

\newcommand{\lossn}{\mathbf{L{}}^{\!n}\hspace{-1pt}}
\newcommand{\loss}{\mathbf{L}}

\newcommand{\eps}{\varepsilon}

\newcommand{\f}{\frac}

\newcounter{modcount}
\newcommand{\modulo}[2]{%
\setcounter{modcount}{#1}\relax
\ifnum\value{modcount}<#2\relax
\else\relax
\addtocounter{modcount}{-#2}\relax
\modulo{\value{modcount}}{#2}\relax
\fi}
\newcommand{\tablepictures}[4][c]{\begin{tabular}[#1]{@{}c@{}}#2\vspace{0.5cm}\\(\alph{#4}) #3\end{tabular}}
\newcounter{gridsearch}
\newcommand{\tabpic}[2]{
    \stepcounter{gridsearch}
    \modulo{\thegridsearch}{2}
    \ifnum\value{modcount}=0
        \tablepictures[t]{#1}{#2}{gridsearch}\\[2.0cm]
    \else
        \tablepictures[t]{#1}{#2}{gridsearch}&~&
    \fi
}

\makeatother
\newtheorem{lemma}{Lemma}[section]
\newtheorem{proposition}[lemma]{Proposition}
\newtheorem{theorem}[lemma]{Theorem}

\newtheorem{example1}[lemma]{Example}
\newtheorem{rem1}[lemma]{Remark}
\newtheorem{assumption}[lemma]{Assumption}
\newtheorem{alg1}[lemma]{Algorithm}
\newtheorem{me1}[lemma]{Mechanism}

\newenvironment{remark}{\begin{rem1}\rm}{\end{rem1}}
\newenvironment{example}{\begin{example1}\rm}{\end{example1}}

\usepackage{color}

\newcommand{\T}{\top}
\newcommand{\tr}{\top}

\begin{document}

\title{A Dynamic Default Contagion Model:\\ From Eisenberg--Noe to the Mean Field}
\author{Zachary Feinstein\thanks{Stevens Institute of Technology, School of Business, Hoboken, NJ 07030, USA. \tt{zfeinste@stevens.edu}} \and Andreas Sojmark\thanks{Imperial College London, Department of Mathematics, London, SW7 2AZ, UK. \tt{a.sojmark@imperial.ac.uk}}}
\date{\today}
\maketitle
\abstract{
	In this work we introduce a model of default contagion that combines the approaches of Eisenberg--Noe interbank networks and dynamic mean field interactions.  The proposed contagion mechanism provides an endogenous rule for early defaults in a network of financial institutions.  The main result is to demonstrate a mean field interaction that can be found as the limit of the finite bank system generated from a finite Eisenberg--Noe style network. In this way, we connect two previously disparate frameworks for systemic risk, and in turn we provide a bridge for exploiting recent advances in mean field analysis when modelling systemic risk. The mean field limit is shown to be well-posed and is identified as a certain conditional McKean--Vlasov type problem that respects the original network topology under suitable assumptions.
	~
}

\section{Introduction}\label{sec:intro}
More than a decade after the collapse of Lehman Brothers and the threat of contagious defaults throughout the global financial system in 2008, systemic risk is still of vital importance to study.  Systemic risk is the risk of financial contagion, i.e., when the failure of one institution spreads to others due to interlinkages in balance sheets both direct (e.g., via obligations) or indirect (e.g., via overlapping portfolios).  The 2008 financial crisis demonstrated the magnitude of the costs that systemic crises produce; this necessitates the design of models to consider stress testing of financial institutions to improve regulation and mitigate the worst effects of a crisis.

In this work, we aim to bridge the divide between two, currently unrelated, modeling techniques for financial contagion.  That is, we will connect the Eisenberg--Noe network approach popularized by \cite{EN01} to the more recent mean field approaches of systemic risk.  As the goal of this paper is primarily to highlight the overlapping notions between those works, and demonstrate that the network models in fact converge to the mean field limit, we will focus on simple, but realistic, financial settings that illustrate this point.  

Briefly, there are two main contagion channels for systemic risk: \eqref{item:default} \emph{default contagion} and \eqref{item:liquidity} \emph{liquidity contagion}. 
\begin{enumerate}
	\item\label{item:default} Default contagion occurs if the failure of one bank or institution to repay its debts in full causes other banks to default triggering a chain reaction of failing banks.  This occurs through, e.g., a network of interbank obligations as studied in the seminal works of~\cite{EN01,RV13} in a static, network-based setting.  More specifically, in those works, the default of a bank causes direct impacts to the balance sheets of other banks in the financial system. This loss of capital (potentially) causes other banks to default, thus spreading the original shock further throughout the financial system.  Such an event is denoted ``default contagion'' as the contagion is via default events.  
	\item\label{item:liquidity} Liquidity contagion occurs if the illiquidity of one bank or institution (as measured by, e.g., the leverage ratio) causes other banks to also become illiquid.  This occurs through, e.g., a fire sale of assets; the liquidation of assets causes the prices fall, this harms the leverage ratio of all other institutions via mark-to-market accounting thus causing further liquidations.  This has been studied in~\cite{CFS05,BW19,feinstein2019leverage}.
\end{enumerate}
In this work we will focus solely on default contagion.  Utilizing just this notion of contagion allows us to focus on the two main streams of literature mentioned above---namely balance sheet constructed network models and dynamic mean field models---in order to compare them and, ultimately, show that these at first sight divergent areas are, in fact, studying the same phenomena.

The specific modeling assumptions undertaken in this work are chosen in order to capture realistic financial networks that incorporate dynamic defaults that come as a \emph{random} shock to the system.  This dynamic and stochastic default contagion is due to all banks holding assets that evolve stochastically over time (in this work, often assumed to be a generalised geometric Brownian motion) and having interbank contractual obligations with fixed repayment schedules.  Key to this construction is the realistic determination of default considered herein; as in real financial systems, we impose that a bank enters default once it has negative capital (on its balance sheet due to appropriate accounting techniques).  We are primarily interested in how these default shocks propagate through the system over time as a default contagion event. Moreover, we are interested in how these events depend on the common exposures of the banks, which we model through a common noise component of the external assets. Passing to the mean field limit, the direct effects of the idiosyncratic noise are averaged out, while the propagation of feedback effects from default contagion remains conditional on the common noise.

The remainder of this paper will be organized as follows. Section~\ref{sec:review} provides a more detailed, but succinct, overview of both the network-based models and mean field models for systemic risk. It seems that such a combined literature review has not previously been attempted, so we hope it can spark more interaction between these two areas of research in systemic risk; however, those comfortable with this background can safely begin this work in Section~\ref{sec:finite}. From here, Section~\ref{sec:finite} provides an extension of the dynamic network model of~\cite{BBF18} to include early defaults based on the realised capital of the institutions.  This section is studied with a finite network of banks to provide understanding of the system dynamics and is novel on its own. To illustrate the workings of our model, simple numerical examples are provided. Focusing on a particular core-periphery setting, this model is extended in Section~\ref{sec:mf} to consider the limit as the number of institutions becomes large. In so doing, we introduce a mean field model of systemic risk directly built from the balance sheet approach, and we provide a simple numerical simulation based on a concrete core-periphery example from Section~\ref{sec:finite}. Finally, Section~\ref{sec:well-posedness} is dedicated to a more general mean field analysis of which the model in Section~\ref{sec:mf} is a special case. This model and the associated mathematical results are the main contributions of this work. By presenting both a finite network and mean field limit for the same systems, we are able to take advantage of the positive aspects of both modelling frameworks, including the balance sheet structure from the (finite) Eisenberg--Noe setting combined with the lower parameter space and more concise mathematical results of the mean field approach.

\section{Eisenberg--Noe and Mean Field Approaches to Systemic Risk}\label{sec:review}
\subsection{Network Approaches}\label{sec:intro-network}

Network approaches to systemic risk consider a directed and weighted graph of interbank obligations to determine the resultant clearing payments made between financial institutions.  These networks determine the defaulting set of banks by finding exactly those institutions who do not pay off their obligations in full when the network clears.  A fundamental concept for such models is the stylized banking balance sheet.  By and large, these models provide a static snapshot of the health of the financial system, though dynamics have begun to be included in select network models; more details on these models will be provided herein.

\paragraph{The Eisenberg--Noe model.} To fix the concepts, we will present first the Eisenberg--Noe clearing payment system from the seminal work~\cite{EN01}.
In short, consider $n$ banks that make up the financial system, each with some initial endowment $x_i \geq 0$ for firm $i$.  Each firm additionally has liabilities to other institutions denoted by the total liabilities $\bar p_i \geq 0$ for firm $i$ and the relative obligations between two firms $\pi_{ij}$ is given by the proportion of the total liabilities of firm $i$ owed to $j$.  The realised payments $p \in \bbr^n$ that each firm makes is obtained by the fixed point equation
\begin{equation}\label{eq:EN-fixedpt}
p = \bar p \wedge \left(x + \Pi^\T p\right)
\end{equation}
where $a \wedge b := (\min(a_1,b_1),\dots,\min(a_n,b_n))$ denotes the lattice minimum in the usual way.
One of the key strengths of this model is in its relation to a simple balance sheet that can be calibrated to data, as undertaken in, e.g.,~\cite{UW04,MZM12,ACP14,GV16} using, e.g., data from the European Banking Authority. 

\paragraph{Analysis of \eqref{eq:EN-fixedpt}.}
The Eisenberg--Noe model for clearing payments inherently codifies three key financial constructs: 
\begin{enumerate}
\item \emph{priority of debt over equity}: a firm must first pay off its debts in full before it accumulates any equity; 
\item \emph{limited liabilities}: no firm pays more than their contractual obligations; and 
\item \emph{pro-rata repayment}: there is no seniority structure of debt.
\end{enumerate}
This last assumption, on pro-rata repayment, has been weakened to allow for varying seniority structures and prioritized repayment in, e.g.,~\cite{E07,feinstein2017currency}.  A fourth assumption is also considered in the Eisenberg--Noe model, though it is relaxed in many subsequent works.  That is,
\begin{enumerate}\setcounter{enumi}{3}
\item \emph{full recovery in default}: a firm has no costs associated with defaulting on its obligations.
\end{enumerate} 
This was studied in~\cite{RV13,GY14,AW_15} as a strict extension of the Eisenberg--Noe model.  The model of~\cite{GK10} can also be viewed as an extension~\eqref{eq:EN-fixedpt} with 0 recovery in default, i.e., no payment is made in case of default.
Given that these financial constructs are all rules that define the value of assets and liabilities, the Eisenberg--Noe model and its extensions are often described as a balance sheet description of the financial system.

Mathematically, as the clearing payment problem inherent to the Eisenberg--Noe model is a fixed point equation~\eqref{eq:EN-fixedpt}, the question of existence and uniqueness is of the paramount importance.  Under the first three financial constructs, and thus allowing for bankruptcy costs, there exists a lattice of clearing payments via the application of Tarski's fixed point theorem.  In particular, this implies that there exists a greatest clearing payment vector; such a payment scheme would always be chosen as all institutions in the system have the greatest possible equity under this scheme, i.e., it is the Nash equilibrium of all clearing payments.  In addition, if we introduce the fourth financial construct (full recovery of assets) then, under very simple assumptions (e.g., all banks in the system hold some initial endowment), the clearing payment is unique.  In addition, in the Eisenberg--Noe setting, the sensitivity of the clearing payments to changes in the system parameters has been studied in~\cite{LS10} for the endowments and~\cite{feinstein2017sensitivity} for consideration of the relative liabilities.

The \emph{fictitious default algorithm}, first presented in~\cite{EN01}, is used to efficiently find the greatest clearing solution.  Briefly, this algorithm initially assumes no banks are in default and determines the clearing payments under such an assumption. If any banks default in that scenario then in the greatest clearing solution they must also be in default. Fixing only those banks as defaulting, a new clearing payment is computed under such a setting.  This process of checking for defaults and determining new clearing payments under fixed set of insolvent banks is repeated until there are no new defaults.  As the Eisenberg--Noe model consists of finite number $n$ of banks, this process is guaranteed to converge in at most $n$ iterations.  Though this algorithm is efficient, and the underlying problem is mathematically well-structured, analytical results are typically not feasible to provide.

\paragraph{Random graph approach.}
In order to obtain some analytical results in this network framework, prior works consider passing to the $n \to \infty$ limit of banks in the system.  These network asymptotics are considered in settings in which the interbank liabilities are described by a random graph.  The simplest random graph model is the Erd\"os--R\'enyi network in which a fixed valued connection is made between any two banks based on a fixed probability.  In many of these asymptotic studies the model of~\cite{GK10} is utilised, i.e., there is no recovery in case of default.  Such an all-or-nothing payment setting leads to tractable formulae for the probability of defaults in the graph.  We refer to~\cite{hurd2016}, and references therein, for a detailed survey of this random graph approach.

The analytical results from this random graph approach allow for further considerations of system stability as well.  For the finite network setting, systemic risk measures to determine acceptable capital requirements for the entire financial system have been proposed in, e.g.,~\cite{chen2013axiomatic,kromer2013systemic,feinstein2014measures,fouque2015systemic}, but these objects require Monte Carlo simulation for any computation.  In contrast,~\cite{ACM10,ACM12,amini_minca,detering2019inhomogeneous} are able to define a resiliency metric and determine capital requirements for banks to make the system acceptable to regulators.  This asymptotic framework thus provides for simple comparative statics.

\paragraph{Dynamic network models of default contagion.}
Considerations given thus far are solely in a static setting.  However, banking balance sheets are highly dynamic and subject to fluctuations due to, e.g., market movements.
Indeed the conclusion of \cite{EN01} gives a discussion of how to include multiple clearing dates and time dynamics, which has been studied in \cite{CC15,ferrara16}. Additionally, \cite{KV16} considers a similar approach to a financial model with multiple maturities.  

As \cite{CC15} presented in a discrete time setting: A firm is \emph{liquid and solvent} at some time $t$ if it has positive equity and has not previously been insolvent.  Due to the assumptions inherent in \eqref{eq:EN-fixedpt}, if a firm is liquid then it must pay in full.  A firm is \emph{illiquid and solvent} at some time $t$ if it has negative cash account, but is able to obtain a loan to cover its deficits, and has not previously been insolvent.  As will be described in this work, and as undertaken in~\cite{BBF18,sonin2017}, these loans will take the form of rolling forward unpaid debts from solvent firms to their obligees.  This is the key distinction that cannot exist in the pure Eisenberg-Noe framework as there is no future time point to repay the loan.  Finally a firm is \emph{insolvent} at time $t$ if either it has negative equity or it was previously deemed insolvent.

Most prior works on dynamic network models consider a discrete time setting~\cite{CC15,ferrara16,KV16}.  As far as the authors are aware, the only two extensions of the Eisenberg--Noe framework to continuous time are~\cite{BBF18,sonin2017}.  Neither of those works considers the mark-to-market equity of a bank in order to determine insolvency.  That is the key innovation being provided in Section~\ref{sec:finite} of this work.

\subsection{Mean Field Approaches}\label{sec:intro-meanfield}


In the dynamic mean field approaches to systemic risk, the starting point is to identify each bank in a large financial system with some notion of its financial robustness or distance-to-default at any given time.
Next, the financial system is then modelled as a system of interacting stochastic processes, whose values represent the current robustness. Typically, these models start from some form of Brownian dynamics, but colloquially one could say  that they are built on the following premise: in contrast to a risk-neutral Black--Scholes world, the modelling of systemic risk calls for room to play with the drift and other aspects of the coefficients, in a way that takes into account the system as a whole.

\paragraph{A concrete particle system.}
To fix ideas, let us consider a particular system of $n$ banks, described by their distances-to-default $X_i$ and corresponding default times
\[
\tau_i:=\inf\{t\geq0:X_i(t) < 0\} \qquad \text{for} \quad i=1,\ldots,n.
\]
Letting $X_i^{\tau}\!(t):=X_i(t\land\tau_i)$, a simple `structural' model of systemic risk  (inspired by \cite{fouque2015meanfield}) could then be based on dynamics of the form
\begin{equation}\label{dist_to_def_1}
dX_i^{\tau}\!(t )= \Bigl( \frac{\theta}{n} \sum_{j=1}^n \bigl(X_j^{\tau}\!(t)-X_i^{\tau}\!(t)\bigr) + \mu(t) \Bigr) dt + \sigma (t) dW_i(t),\quad t\leq\tau_i, \quad i=1,\ldots,n,
\end{equation}
where $W_i(t)= \sqrt{1-\rho^2}B_i(t)+\rho B_0(t)$, for a family of independent Brownian motions $B_0,\ldots,B_n$. Here $\rho>0$ models the presence of a `common' noise, namely $B_0$, that captures exposure to common risk factors, while $\theta>0$ incorporates an element of `herding' in the drift that could, for example, be the result of banks engaging in similar strategies and other interbank connections. To capture systemic risk, the natural quantities in this model are the average distance-to-default and the proportion of defaults given, respectively, by
\[
\textbf{M}^n\hspace{-1pt}(t):=\frac{1}{n}\sum_{i=1}^nX_i^{\tau}\!(t) \quad \text{and} \quad \lossn(t):=\frac{1}{n}\sum_{i=1}^n\mathbf{1}_{t\geq\tau_i}.
\]
Of course, an obvious weakness is the inherent symmetry in this model, and this is indeed a central point to be addressed later in this work when we return to the Eisenberg--Noe approach, as discussed in Section \ref{sec:intro-network} above. For the purposes of this overview, however, we remain in the symmetric setting.

\paragraph{Mean field analysis of \eqref{dist_to_def_1}.}
Sending $n\rightarrow\infty$ in \eqref{dist_to_def_1}, one can hope to simplify the analysis and simulation of the system, provided there is a law of large numbers effect. To see that there is indeed such an effect, note that $\textbf{M}^n_t = \langle \nu_t^n , \mathrm{Id} \rangle$ and $\textbf{L}^n_t= 1- \nu_t^n(0,\infty)$, where the empirical measures $\nu^n_t:=\frac{1}{n}\sum_{i=1}^n \mathbf{1}_{t<\tau_i}\delta_{X_i^n(t)}$ are tracking the surviving banks. From \cite{HS18} it is then known that  $(\nu^n,\textbf{M}^n,\textbf{L}^n)$ converges to a unique limit $(\nu,\textbf{M},\textbf{L})$, where $\textbf{M}_t = \langle \nu_t , \mathrm{Id} \rangle$, $\textbf{L}_t= 1- \nu_t(0,\infty)$, and $\nu_t$ has a density $V_t$ which solves the nonlinear SPDE
\begin{equation}\label{SPDE_intro}
dV_{t}(x)={\textstyle \frac{\sigma^{2}}{2}}\partial_{xx}V_{t}(x)dt - \partial_{x}\bigl( [\theta(\textbf{M}(t)-x)+\mu]V_t(x) \bigr) dt -\rho\sigma \partial_{x} V_{t}(x)dB_{0}(t),
\end{equation}
for $x\in(0,\infty)$, with an absorbing boundary condition at the origin, i.e.~$V_t(0)=0$. The distribution of the limiting processes $\textbf{M}$ and $\loss$ can be used to give measures of systemic risk. For example, one can track the probability of seeing changes in $\textbf{M}$ and $\loss$ above some threshold over a short period of time. A simple observation is that, for larger $\theta>0$ and $\rho>0$, the distribution of the change in $\loss$ over a given period becomes more concentrated at the extremes with banks more likely to either survive or default together.

Translating (\ref{SPDE_intro}) to the language of McKean--Vlasov SDEs, we have $\textbf{M}(t)=\mathbb{E}[X(t)\mathbf{1}_{t<\tau} \,|\, B_0]$ and $\loss(t)=\mathbb{P}(t\geq \tau  \, | \, B_0)$, where $\tau=\inf\{t\geq 0: X(t)\leq 0\}$ and
\begin{equation}\label{MV_intro}
dX(t )=  \bigl(\theta(\textbf{M}(t)-X(t))+\mu\bigr)dt + \sigma d(\sqrt{1-\rho^2}B+\rho B_0)(t).
\end{equation}
This follows by an application of It\^o's formula, which shows that (\ref{SPDE_intro}) is the (nonlinear) stochastic Fokker--Planck equation for (\ref{MV_intro})  absorbed at the origin, conditional on $B_0$. That is, we have
$\mathbb{E}[\phi(X_t)\mathbf{1}_{t<\tau}\mid B_0]=\int_0^\infty \phi(x)V_t(x)dx$ for all $\phi\in C^2_b(\mathbb{R})$.

Building on the above, one could consider strategic interactions in (\ref{dist_to_def_1}) with costs and controls depending on $\textbf{M}^n$ and $\loss^n$. This would then yield a mean field game involving the limit processes $\textbf{M}$ and $\loss$. Without the common noise, a framework for this type of mean field game has recently been developed in the two consecutive papers \cite{campi_mfg_1, campi_mfg_2}.


\paragraph{Mean field models of contagion.}

Recently, a new line of mean field modelling has been proposed in \cite{HLS18, HS18, NS17} aimed at studying \emph{default contagion} in large financial systems. Based on a `structural' approach, these models introduce an endogenous notion of contagion in systems such as (\ref{dist_to_def_1}), by imposing that bankruptcies should cause a drop in the distances-to-default of the other banks. Mathematically, this amounts in one way or another to incorporating the proportion of defaults $\loss^n$ into the dynamics, thus leading to positive feedback loops whereby defaults can shift other banks into default. Variations of this approach and further theoretical results can be found in \cite{LS18a, LS18b, NS18}. Moreover, we note that closely related approaches to contagion (in a dynamic but finite-dimensional setting) have also been considered in \cite{Battiston_2012, Lipton2016}.

In terms of numerical implementation, \cite{KLR18c, KR18} have proposed and analysed numerical schemes for the mean field model of \cite{HLS18}, and it is noted in \cite{KLR18c} that a modified version of \cite{Lipton2016} falls within this framework. These developments can be seen as following on from \cite{Lipton2015,KLR18a,KLR18b}, where similar models are studied for systems of two or three banks, and we note that passing to the mean field limit yields a way of alleviating the curse of dimensionality arising from the couplings due (in particular) to mutual obligations in large financial systems.

In this paper we will show how a variant of these `structural' approaches to contagion is intrinsically connected to a dynamic Eisenberg--Noe model with early defaults (as developed in Section~\ref{sec:finite}). Moreover, we will show (in Sections \ref{sec:mf} and \ref{sec:well-posedness}) that the associated mean field limit can be derived and analysed rigorously by extending the techniques from \cite{HLS18,LS18a,LS18b}.

\paragraph{The broader mean field literature on systemic risk.}
If, for simplicity, the constraints on the state space in particle system (\ref{dist_to_def_1}) are dropped, then the dynamics are precisely those of the early papers \cite{fouque2015meanfield, fouque2013illustrated}, where $X_i$ now denotes the (logarithmic) cash-reserves of bank $i$ and the mean-reversion models borrowing and lending in the interbank market. In \cite{fouque2015meanfield} these dynamics emerge as a Nash equilibrium of a stochastic game (where drifts are controlled and it is costly to diverge from the mean) with a variant of (\ref{MV_intro}) without absorption arising as the equilibrium dynamics for the limiting mean field game.

Starting from \cite{fouque2013stability}, several other papers on systemic risk have studied different versions of this mean-reverting setup (mostly without the common noise). These contributions can be loosely grouped into: systems with stabilisation by a central agent \cite{garnier2013b, garnier2017}, games with delay \cite{carmona2016delay, fouque_zhang}, games with model uncertainty \cite{huang_jaimungal}, utility optimisation by the  individual banks and a central bank \cite{sarantsev_maheshwari},  methods for introducing heterogeneity \cite{chong_kluppelberg, spilio_hetero}, jump-diffusion dynamics \cite{bo_capponi, pascucci_sysrisk, Benazzoli_campi_persio} and connections to the theory of risk measures \cite{biagini_fouque}. Still focusing on mean-reversion, constraints on the state space have been considered via Feller type square root diffusions in \cite{bo_capponi, fouque2013stability,misha_ichiba, sun_interbank} (with various additional features) and, in such a framework, \cite{capponi_clusters} has recently proposed a network structure with finitely many clusters of banks, where each cluster mean-reverts around different predetermined levels modelling the presence of target leverage ratios.

In addition to the `structural' approaches to contagion discussed earlier, there is a separate literature on contagion in large financial systems, wherein defaults are dictated by exponential clocks as in the `reduced-form' approach to credit risk. This leads to more implicit notions of contagion occurring at the level of the intensities. Firstly, \cite{giesecke_2013, giesecke_2015} propose a system of interacting intensities that are self-exciting via dependence on the proportion of defaults. Secondly, somewhat closer in spirit to (\ref{dist_to_def_1}), \cite{ichiba_2018} identifies the financial health of each bank in a large system with a geometric Brownian motion, but with default dictated by an exponential clock whose intensity can depend on the banks own health and the average healthiness of the system. Contagion amounts to each default causing a drop in the healthiness of the other banks by a random fraction, which in turn increases the default intensities.

\section{Dynamic Model of $n$ Banks}\label{sec:finite}

As mentioned above, the primary goal of this section is to introduce a dynamic network model to study default contagion in a finite system of banks.  To do so, we seek to extend the dynamic network model of \cite{BBF18} to incorporate early defaults due to negative (accounting) capital.  
This extension to include early defaults is novel and important in its own right.
This will be utilised in Section~\ref{sec:mf} to consider a comparison with a mean field limit.  More details on the reasons for undertaking that analysis are provided in Section~\ref{sec:mf}, and we also refer to the brief discussion in the introduction above.

This section is broken into two subsections.  First, in Section~\ref{sec:balance-sheet}, we describe the stylized balance sheet of all banks in the system.  This is used to define a general dynamic model for default contagion in the vein of \cite{EN01,RV13}.  Second, we simplify the parameters so as to study a specialized network setup that facilitates the later analyses of this work. As the primary goal of this work is to merge the network and mean field approaches in the literature, we find this specialized network setup is instructive. For the analysis of this section, none of these additional assumptions are required for the theoretical results.

Briefly, before undertaking the analysis, we wish to consider some notation utilised throughout this work.   
Consider a financial system with $n \in \bbn$ financial institutions. This system does not include the central bank or other financial entities not included within this system; we will consider such an entity, called the ``societal node'' and denote it by node $0$.  Notationally, let $\ncal = \{1,2,...,n\}$ be the set of banks and $\ncal_0 = \ncal \cup \{0\}$ include the societal node.
As we are considering a dynamic network model, consider a continuous set of clearing times $\bbt = [0,T]$ for some (finite) terminal time $T < \infty$.  For simplicity, assume throughout this work that the risk-free rate is 0 ($r = 0$).
Finally, we will use the notation  $Z(t)$ for the value at time $t\in \bbt$ of a process $Z: \bbt \to \bbr^n$.
We will now consider a model akin to the continuous-time setting of \cite{BBF18} in that we allow for liabilities to change over time and for firms to have stochastic cash flows.

\subsection{The Balance Sheet}\label{sec:balance-sheet}
\begin{figure}[t]
\centering
\begin{tikzpicture}[x=\linewidth/8]
\draw[draw=none] (0,9.5) rectangle (6,10) node[pos=.5]{\bf Balance Sheet @ $t$};
\draw[draw=none] (0,9) rectangle (3,9.5) node[pos=.5]{\bf Assets};
\draw[draw=none] (3,9) rectangle (6,9.5) node[pos=.5]{\bf Liabilities};

\filldraw[fill=blue!20!white,draw=black] (0,6.5) rectangle (3,9) node[pos=.5,style={align=center}]{External (Mark-to-Market) \\ $\bbe[x_i(T) \; | \; \fcal_t]$};
\filldraw[fill=yellow!20!white,draw=black] (0,4) rectangle (3,6.5) node[pos=.5,style={align=center}]{Interbank (Solvent) \\ $\sum_{j \in \acal_t} L_{ji}(T)$};
\filldraw[fill=orange!20!white,draw=black] (0,0) rectangle (3,4) node[pos=.5,style={align=center}]{Interbank (Insolvent) \\ $\sum_{j \in \ncal \backslash \acal_t} \left(\begin{array}{l}(1 - R_2) L_{ji}(\tau_j)\\ + R_2 L_{ji}(T)\end{array}\right)$};

\filldraw[fill=red!20!white,draw=black] (3,3) rectangle (6,9) node[pos=.5,style={align=center}]{Total \\ $\sum_{j \in \ncal_0} L_{ij}(T)$};
\filldraw[fill=green!20!white,draw=black] (3,0) rectangle (6,3) node[pos=.5,style={align=center}] (t) {Capital \\ $K_i(t)$};
\end{tikzpicture}
\caption{Stylized balance sheet for firm $i \in \ncal$ at time $t \in \bbt$.}
\label{fig:balance-sheet}
\end{figure}

In order to construct a continuous-time model we will begin by considering the stylized balance sheet for a generic bank $i \in \ncal$ in our system.  This balance sheet comes from a dynamic version of~\cite{EN01}.  Throughout time, all assets are of only two types: interbank assets and external assets.  All liabilities are either interbank (and thus assets for another bank $j \in \ncal$ in the system) or external and owed to the societal node $0$.

In order to construct a continuous-time model we will begin by considering our network parameters of cash flows and nominal liabilities.
We will now consider a banking system with stylized balance sheet as depicted in Figure~\ref{fig:balance-sheet}.

Let $x_i(T)$ be the value of the external assets for firm $i \in \ncal_0$ at the terminal time $T$.  This will often be denoted in vector notation as $x(T)$.  In mark-to-market accounting, at time $t \in \bbt$, these external assets should be valued in (risk-neutral) expectation, i.e., $\bbe[x(T) \; | \; \fcal_t] = x(0) + \int_0^t dx(s) + \bbe[\int_t^T dx(s) \; | \; \fcal_t]$.
The value of the external assets can, equivalently be described by the (marginal) cash flows external to the system, e.g., from depositors at the banks, as utilised in~\cite{BBF18} for a dynamic version of~\cite{EN01}.  In this context, we describe $dx(t)$ to be the marginal change in the external assets at time $t \in \bbt$, i.e., firm $i \in \ncal_0$ has incoming external cash flows $\int_{t_1}^{t_2} dx_i(t)$ between times $t_1 < t_2$.  Throughout this work we will take the external assets to follow a non-negative (It\^o) process.

In contrast, we will assume the total nominal liabilities matrix $L$ is a deterministic process of time as these obligations are contractually generated and have fixed repayment schedule.  In other words, by looking at all outstanding contracts at time $0$, the total amount that is owed between any two institutions (and externally) up to any time can be determined exactly.
Generally we will consider $dL(t)$ to be the marginal change in nominal liabilities matrix at time $t$, i.e., the liabilities owed from time $t_1$ to $t_2$ are defined by $\int_{t_1}^{t_2} dL(t) \in \bbr^{(n+1) \times (n+1)}_+$. By assumption $dL_{ij}(t) \geq 0$ for all firms $i,j \in \ncal_0$ as, without any payments made, total liabilities should accumulate over time. Additionally, $dL_{ii}(t) = 0$ for all firms $i \in \ncal_0$ to remove the possibility of self-dealing. 
This nominal liabilities matrix appears on both the asset and liabilities side firms.  The liabilities for firm $i$ is the total amount owed over $\bbt$, i.e., $\sum_{j \in \ncal_0} L_{ij}(T) = \sum_{j \in \ncal_0} \int_0^T dL_{ij}(s)$.
To simplify notation, we will define $d\bar p(t) := dL(t) \vec{1}$ for any time $t \in \bbt$ to denote the marginal change in the total liabilities vector where $\vec{1} = (1,\dots,1)^\tr \in \bbr^{n+1}$; correspondingly, the total liabilities owed by firm $i$ over $\bbt$ are given by $\bar p_i(T)$.
The interbank assets require consideration of historical price accounting since the interbank assets are generally nonmarketable.  As such, firm $i \in \ncal_0$ will give full value to all obligations (both past and future) $\sum_{j \in \acal_t} L_{ji}(T) = \sum_{j \in \acal_t} \int_0^T dL_{ji}(s)$ from solvent institutions $\acal_t$ at time $t$; for insolvent firms $j \in \ncal \backslash \acal_t$, firm $i$ will give full value up to the insolvency time $\tau_j \in \bbt$ (discussed further below), but only a fixed recovery rate $R_2 \in [0,1]$ on obligations from $j$ after insolvency, i.e., $\int_0^{\tau_j} dL_{ji}(s) + R_2 \int_{\tau_j}^T dL_{ji}(s) = (1 - R_2) L_{ji}(\tau_j) + R_2 L_{ji}(T)$.

\begin{assumption}\label{ass:balance-sheet}
The modeling assumptions expressed above can be summarized thusly: 
\begin{enumerate}
\item the external assets of each bank follow a stochastic process (which can be correlated to each other) and (being marketable) are marked-to-market with risk-neutral measure $\bbp$;
\item the interbank assets and liabilities are solely based on contracts written prior to time $0$ and have fixed repayment schedule; and
\item interbank assets (being nonmarketable) are valued using historical price accounting, i.e., priced at face value prior to a default event and reevaluated with the true recovery rate after default.
\end{enumerate}
These three key modeling assumptions lead to a contagion mechanism in which defaults come as a shock to the system and cause a jump in the capital of any connected institution.
\end{assumption}
The shocks due to default outlined above are realistic since the interbank assets are nonmarketable.  If, however, banks attempted a counterparty or network valuation adjustment (see, e.g., \cite{barucca2016valuation,BF18comonotonic}) default shocks would still be expected due to the assymetric and incomplete information available to the different banks.  We also wish to note that the historical price accounting rule undertaken herein provides the greatest possible value for interbank assets and thus provides a bound on any other valuation system.

The balance sheet capital for firm $i$ at time $t \in \bbt$ is exactly the difference on its balance sheet between assets and liabilities, i.e.,
\begin{equation}
\label{eq:Kapital-gen} K_i(t) = \bbe[x_i(T) \; | \; \fcal_t] + \sum_{j \in \acal_t} L_{ji}(T) + \sum_{j \in \ncal \backslash \acal_t} \left[(1-R_2) L_{ji}(\tau_j) + R_2 L_{ji}(T)\right] - \bar p_i(T).
\end{equation}
Insolvency for bank $i$ occurs at the first time that it has negative capital, i.e.,
\[\tau_i = \inf\{t \in \bbt \; | \; K_i(t) < 0\}\]
and the set of solvent firms at time $t$ is given by $\acal_t := \{i \in \ncal \; | \; \tau_i > t\}$.

\begin{remark}\label{rem:early-default}
The stochastic structure introduced herein is necessary for consideration of early defaults.  Without it, the capital of banks would be deterministic and all defaults would be known at the initial time $0$.  Though~\cite{BBF18} introduces a stochastic system for financial networks, it does not consider endogenous early defaults.  That is an innovation of this work.
\end{remark}

\begin{remark}\label{rem:EN}
In the balance sheet approach considered herein, due to the full recovery of interbank assets prior to default and a fixed recovery after default, the details of the Eisenberg--Noe~\cite{EN01} are only subtly utilised in the background.  That is, the constant recovery implies a pro-rata repayment scheme as in the Eisenberg--Noe framework; the difference between this repayment scheme and that of Eisenberg--Noe and Rogers-Veraart \cite{RV13} is that recovery is on the liability side rather than the asset side.  We take this as a simplification to ease the discussion and mathematics to focus primarily on the stylized contagion in this work.

Additionally, we can consider this balance sheet framework as akin to the dynamic network models of~\cite{BBF18,sonin2017}, but adding in notions from the discrete-time model of \cite{CC15} in which firms can default before the terminal time.  In that work there is a detailed discussion on an auction model for determining the recovered assets in case of default from which the remaining debts are paid; this is in contrast to the simplified exogenous recovery rates.  If we take the approach from~\cite{BBF18} in which firms pay off debts as they arrive and may have unpaid prior liabilities, the construction of the system dynamics requires further considerations.  
Briefly, let $V_i(t)$ denote the cash holdings of firm $i$ at time $t \in \bbt$.  Let $\pi_{ij}(t)$ be the relative liabilities at time $t$.  This is constructed in detail in a continuous-time setting in~\cite{BBF18}; we refer to that paper for a detailed discussion of the construction of the relative liabilities in this general setting.  It is possible that a firm has positive capital $K_i(t) > 0$ but insufficient funds to cover short term liabilities; in such a setting we assume that the debts roll-forward when they go unpaid by a \emph{solvent} firm as in~\cite{BBF18,sonin2017}.  When a firm defaults, we consider a recovery rate $R_1 \in [0,1]$ on the unpaid previous debts and $R_2 \in [0,1]$ on future obligations.  As such we have the cash holdings and (modified) capital equations at time $t \in \bbt$ as:
\begin{align*}
V(t) &= x_i(t) + \sum_{j \in \ncal} \left[\left(L_{ji}(t) - \pi_{ji}(t) V_j(t)^-\right) \mathbf{1}_{t < \tau_j}\right. \\
	&\qquad \left. + \left((1-R_2) L_{ji}(\tau_j) + R_2 L_{ji}(T) - (1-R_1) \pi_{ji}(\tau_j) V_j(\tau_j)^-\right) \mathbf{1}_{t \geq \tau_j}\right] - \bar p_i(t)\\
K_i(t) &= V_i(t) + \bbe[x_i(T) \; | \; \fcal_t] - x_i(t) + \sum_{j \in \ncal} \left(\pi_{ji}(t) V_j(t)^- + L_{ji}(T) - L_{ji}(t)\right) \mathbf{1}_{t < \tau_j} - [\bar p_i(T) - \bar p_i(t)].
\end{align*}
Much of the results of this work can be undertaken in this setting with a liability structure defined in Assumption~\ref{ass:finite}.  
For simplification and from financial interpretation: we will be interested in the setting where $R_1 = 1$.  We assume this from the idea that, prior to the default even though a firm may be illiquid, it is solvent and thus some \emph{lender of last resort} will guarantee these obligations that rolled forward. 
\end{remark}

\subsection{The Simplified Model}\label{sec:model}
\begin{figure}[t]
\centering
\includegraphics[width=0.6\textwidth]{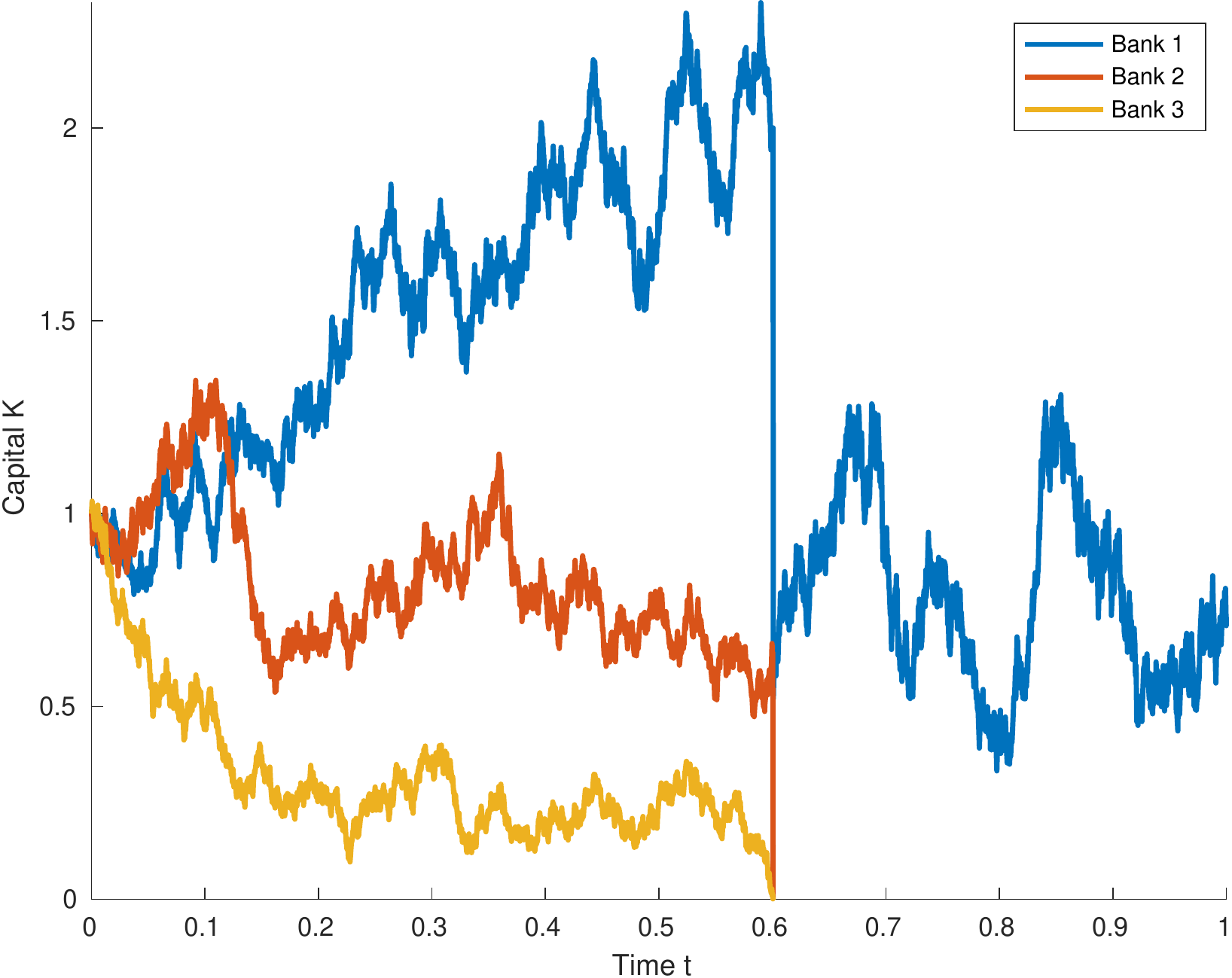}
\caption{A realisation of the 3 bank system described in Example~\ref{ex:3bank} with a contagious default at $t \approx 0.6$.}
\label{fig:3bank}
\end{figure}

We will make the following assumptions for the remainder of this paper.  These can be relaxed as in~\cite{BBF18} and discussed briefly in Remark~\ref{rem:EN}, but as the goal of this work is to demonstrate a simple and clear comparison between a dynamic Eisenberg--Noe model and default contagion in the mean field limit we will consider this simplification.
\begin{assumption}\label{ass:L-finite}
Throughout this work we consider a short time horizon $T$ model wherein the liability repayment schedule is constant over time, i.e., $dL_{ij}(t) = \lambda_{ij}dt$ for $i,j \in \ncal_0$ with $\lambda_{ii} = 0$ and $\lambda_{0j}(t) = 0$, and $L(0) = 0$.  
\end{assumption}
The constant nature of the network, as defined in Assumption~\ref{ass:L-finite}, is valid for a short time frame. Since financial crises occur over short time horizons, this fixed nature is therefore appropriate.  Further, this allows us to study how the initial network topology can cause default contagion and, ultimately, a systemic crisis.

For the remainder of this paper we will introduce the notation 
\[\bar{L}_i(T) := \sum_{j \in \ncal_0}[L_{ij}(T) - L_{ji}(T)] = T \lambda_{i0} + T \sum_{j \in \ncal} [\lambda_{ij} - \lambda_{ji}]\]
to denote the difference between obligations (external and interbank) and interbank assets.  Typically we will assume $\bar{L}_i(T) > 0$ for all firms $i$, i.e.,  bank $i$ has liabilities that cannot be offset solely by interbank assets.  Under such an assumption, every bank is a net borrower overall (in the sense that \emph{total} obligations are larger than \emph{interbank} assets); this is true even if a bank is not a net borrower in the interbank system.

\begin{assumption}\label{ass:finite}
For simplification and ease of use, we will assume for the remainder of this paper that the external cash flows follow (possibly time-dependent) correlated geometric Brownian motions, i.e., 
\[dx_i(t) = x_i(t)[\mu_i(t) dt + \sigma_i(t) dW_i(t)]\]
for vector of correlated Brownian motions $W$.
\end{assumption}

Under the setting of Assumptions~\ref{ass:L-finite} and~\ref{ass:finite}, we can compute the capital process $K(t)$ from~\eqref{eq:Kapital-gen} as:
\begin{align}
\label{eq:Kapital} K_i(t) 
&= x_i(t) e^{\int_t^T \mu_i(s) ds} - \bar{L}_i(T) - (1-R_2)\sum_{j \in \ncal\backslash\acal_t} (T-\tau_j)\lambda_{ji}.
\end{align}
In order to determine~\eqref{eq:Kapital}, we take advantage of the external assets following a geometric Brownian motion to find $E[x_i(T) \; | \; \fcal_t] = x_i(t) e^{\int_t^T \mu_i(s) ds}$.

\begin{assumption}\label{ass:initial-default}
We wish to assume that no banks are in default at time $t = 0$.  This is equivalent to bounding the initial external assets from below, i.e., $x_i(0) > T[\lambda_{i0} + \sum_{j \in \ncal}[\lambda_{ij} - \lambda_{ji}]]e^{-\int_0^T \mu_i(s) ds}$ almost surely for every bank $i \in \ncal$.
\end{assumption}

\begin{lemma}\label{lemma:exist}
If Assumptions~\ref{ass:L-finite}-\ref{ass:initial-default} are satisfied, there exists a greatest and least clearing capital $K^\uparrow \geq K^\downarrow$ (component-wise and for every time $t$).
\end{lemma}
\begin{proof}
First, recall that the default times are defined by $\tau_i = \inf\{t \in \bbt \; | \; K_i(t) < 0\}$ for bank $i$.  The fixed point problem for the capital process $K$ only depends on itself through the default times $\tau$.
Now, consider the fixed point in the capital process.  Note that as the capital process $K$ decreases the default times $\tau$ all decrease.  Further, as banks default, the entire system's wealth drops as well since $R_2 < 1$.  With this, we are able to complete this proof through a use of Tarski's fixed point theorem.
\end{proof}

Throughout the remainder of this work, we will focus on the greatest clearing solution $K^\uparrow$.  In the Eisenberg--Noe framework, this is computed using a \emph{fictitious default algorithm}.  Briefly, such an algorithm assumes that at time $t \in \bbt$, any bank that was solvent prior to $t$ ($\acal_{t-}$) is assumed to still be solvent; this is the best case scenario for all banks due to the downward stresses from a default.  Solvency ($K_i(t) \geq 0$) of all banks is then checked under this scenario; if no banks default we can move forward in time, otherwise any new defaults may cause a domino effect of further defaults.  In the case of defaults, we update the balance sheet of all solvent firms to determine if this shock causes a cascade of failures. This sequential testing for new defaults and updating the balance sheets continue until no new defaults occur.  In practice this algorithm is run using an event finding algorithm to determine the time of the initial default, at that time the cascading defaults are determined until the system re-stabilises at a new set of solvent institutions, and the stochastic processes evolve normally until the next default event.  This is demonstrated in Figure~\ref{fig:3bank} where the insolvency of one bank causes another bank to default as well.  If desired, the least clearing solution $K^\downarrow$ could be found analogously with a fictitious solvency algorithm instead.

We wish to preview a consideration of the \emph{cascade condition} of Section~\ref{subsubsec:stoch_evolution} and Section~\ref{sec:jumps_loss} which is a rule for determining the size of default cascades (tailored to our later reformulation of the Eisenberg--Noe banking system as a stochastic interacting particle system). We do this to highlight the similarity between this condition and the fictitious default algorithm widely used in the (finite) network setting. Further, we wish to emphasise that, in the mean field framework, a cascade needs to be defined more carefully, as liabilities between institutions are infinitesimals and it is no longer meaningful to talk about defaults of individual banks. The details underlying the definition of the cascade condition are provided in Section~\ref{Sec:determine_cascades}, and we emphasise already here that the iterations of this cascade condition correspond analogously with the fictitious default algorithm (see, in particular, Section~\ref{subsect:rule_cascade} for the precise description of the resolution of default cascades). As far as the authors are aware, a connection between the fictitious default algorithm and the cascade condition presented herein, or its predecessors in the mean field literature, has not previously been investigated.

We conclude this section with consideration of two numerical examples.  The first is the description of the small 3 bank system shown in Figure~\ref{fig:3bank}.  We then consider a larger system with a core-periphery structure.  As noted in~\cite{CP14,FL15}, many real world financial systems exhibit a core-periphery structure.

\begin{example}\label{ex:3bank}
This example will be constructed to demonstrate the primary features of the model, namely early defaults and the contagion thereof.
Consider a $n = 3$ bank system over a time horizon $\bbt = [0,1]$ for simplicity.  
These banks are connected to each other through interbank obligations
\[\lambda = \left(\begin{array}{cccc} 0 & 2 & 2 & 1 \\ 2 & 0 & 2 & 1 \\ 2 & 2 & 0 & 1 \end{array}\right)\]
where the last column denotes the obligations $\lambda_{\cdot 0}$ to the societal node.
These banks hold identical, correlated external assets following the geometric Brownian motion $dx_i(t) = x_i(t) [dt + \frac{1}{2} dW_i(t)]$ with correlations $dW_i(t) dW_j(t) = \frac{1}{2} dt$ for $i \neq j \in \{1,2,3\}$.  The initial value of these external assets are chosen so that the initial capital of all banks is 1, i.e., set $x_i(0) = 2 e^{-1}$ (which results in $K_i(0) = 1$) for all banks $i \in \{1,2,3\}$.  By setting the parameters in this way we satisfy Assumptions~\ref{ass:L-finite}-\ref{ass:initial-default}.  Finally, we fix the recovery rate for defaulted assets to be $R_2 = 0.1$.
One realisation of this system is shown in Figure~\ref{fig:3bank}.  In that realisation, bank $3$ defaults due to its own investments at time $\tau_3 \approx 0.6$.  This causes bank $2$ to default at that time $\tau_2 \approx 0.6$ as well due to the shock to its own capital from re-marking its interbank assets.  Bank $1$ remains solvent for the studied time frame $\bbt$ but a large negative shock is exhibited on its capital due to the default of both banks $2$ and $3$.
\end{example}

The following example is a 10 bank system exhibiting the core-periphery structure.  This network topology is discussed in more detail in~\cite{BE00,FL15}. Of particular note, empirical studies (see, e.g., \cite{CP14,FL15}) have demonstrated that real-world financial systems exhibit this structure.  In undertaking this study, two network structures will be considered: first, a highly connected networks with small obligations from peripheral to peripheral institutions; second, we consider the same network but with no obligations between peripheral institutions.  The second, approximate, network has a limit to its rank, i.e., the rank of the second network obligation matrix is bounded by twice the number of core institutions.  This approximating network, with a much sparser network, is more tractable computationally and, as such, will be used to motivate a rank decomposition structure in the mean field limit of the dynamic network model discussed within this section.  For more details, we refer to the next section.

\begin{example}\label{ex:core-periphery}
Consider a $n = 10$ bank core-periphery system with $2$ core banks and $8$ peripheral ones.  These banks are interconnected with the (randomized) interbank liabilities
\[\lambda = \left(\begin{array}{cccccccccc} 0 & 15.01 & 0 & 0 & 0 & 0 & 3.43 & 2.87 & 2.87 & 2.80 \\ 45.35 & 0 & 3.08 & 2.36 & 2.78 & 2.80 & 1.13 & 0.94 & 0.94 & 0.92 \\ 4.54 & 2.23 & 0 & 0.04 & 0 & 0.05 & 0.06 & 0.05 & 0.05 & 0 \\ 5.90 & 2.90 & 0 & 0 & 0 & 0.06 & 0.07 & 0 & 0 & 0 \\ 4.67 & 2.29 & 0.05 & 0.04 & 0 & 0 & 0 & 0.05 & 0 & 0.05 \\ 4.40 & 2.16 & 0.05 & 0.04 & 0 & 0 & 0.05 & 0.05 & 0.05 & 0 \\ 3.64 & 4.47 & 0 & 0.04 & 0 & 0.05 & 0 & 0 & 0 & 0.05 \\ 3.41 & 4.18 & 0 & 0.04 & 0.04 & 0.04 & 0.05 & 0 & 0 & 0.04 \\ 3.25 & 3.99 & 0 & 0 & 0 & 0 & 0.05 & 0.04 & 0 & 0 \\ 4.31 & 5.29 & 0 & 0.05 & 0 & 0 & 0 & 0 & 0 & 0 \end{array}\right)\]
and all banks owe \$1 to the societal node.  As in real financial systems, there are sparse, and small, obligations between peripheral firms.  For a comparison, consider a reduced system of obligations so that the obligations between peripheral firms are zeroed out, i.e.,
\[\hat\lambda = \left(\begin{array}{cccccccccc} 0 & 15.01 & 0 & 0 & 0 & 0 & 3.43 & 2.87 & 2.87 & 2.80 \\ 45.35 & 0 & 3.08 & 2.36 & 2.78 & 2.80 & 1.13 & 0.94 & 0.94 & 0.92 \\ 4.54 & 2.23 & 0 & 0 & 0 & 0 & 0 & 0 & 0 & 0 \\ 5.90 & 2.90 & 0 & 0 & 0 & 0 & 0 & 0 & 0 & 0 \\ 4.67 & 2.29 & 0 & 0 & 0 & 0 & 0 & 0 & 0 & 0 \\ 4.40 & 2.16 & 0 & 0 & 0 & 0 & 0 & 0 & 0 & 0 \\ 3.64 & 4.47 & 0 & 0 & 0 & 0 & 0 & 0 & 0 & 0 \\ 3.41 & 4.18 & 0 & 0 & 0 & 0 & 0 & 0 & 0 & 0 \\ 3.25 & 3.99 & 0 & 0 & 0 & 0 & 0 & 0 & 0 & 0 \\ 4.31 & 5.29 & 0 & 0 & 0 & 0 & 0 & 0 & 0 & 0 \end{array}\right).\]
This reduced system $\hat\lambda$ has rank $4$ instead of full rank for the original network.
In Table~\ref{table:core-periphery}, the default time for each bank is reported under the original (full) and reduced networks.  Notably, though these default times are not identical, they capture the general behavior quite accurately.  We will take advantage of this notion of the reduced system in the following sections.
\begin{table}
\centering
\begin{tabular}{r|*{10}{c}}
\textbf{Bank} & 1 & 2 & 3 & 4 & 5 & 6 & 7 & 8 & 9 & 10 \\ \hline
Full Network & -- & 0.173 & 0.173 & 0.042 & 0.179 & 0.173 & -- & -- & -- & 0.173 \\
Reduced Network & -- & 0.172 & 0.172 & 0.041 & 0.189 & 0.172 & -- & -- & -- & 0.172
\end{tabular}
\caption{Default times for institutions in Example~\ref{ex:core-periphery} under a single realisation of the external assets. Normed difference between these default times is $0.0097$.}
\label{table:core-periphery}
\end{table}
\end{example}
In the next section, we start from the above example and discuss the mean field limit that results from sending the number of banks to infinity in a suitable way.

\section{A core-periphery mean field model}\label{sec:mf}

As already mentioned in the introduction of this paper, the main motivation for the present section is theoretical in nature: the aim being to close a gap between the network literature on systemic risk and recent mean field contagion models. Specifically, we will relate the finite interbank system from Section \ref{sec:finite} to a mean field limit described by a conditional McKean--Vlasov problem akin to the problems studied in \cite{HLS18, HS18, LS18a, NS18, NS17}.

Aside from this theoretical perspective, there are several good practical reasons for studying the mean field limit of the model proposed in Section 3. First of all, the mean field limit rigorously facilitates a low parameter space, which can allow for a clearer identification of the main mechanisms at work, and which may serve as a vehicle for defining macroscopic events. Secondly, the mean field limit can allow for more efficient numerical simulations by replacing a large system of coupled SDEs with a single limiting object. Thirdly, one is unlikely to have precise data for the liabilities matrix, but the mean field limit makes a rigorous case for working with an approximate distribution. Finally, the lower parameter space can facilitate calibration to the average of a large sample of banks, as opposed to the unfeasible task of fitting fully heterogeneous parameters in the finite dynamic system.

In the present section, we focus on the financial motivation and thus restrict attention to the core-periphery structure discussed at the end of Section \ref{sec:finite}. This leads us to introduce a particular intuitive and tractable mean field point of view on the Eisenberg--Noe style interbank system from Section \ref{sec:finite}. While we give a careful presentation of the mathematical results for the mean field limit, along with a numerical example, the theoretical details are left to Section \ref{sec:well-posedness}, which treats a more general framework.

\subsection{A simple model of core-periphery interbank networks}\label{subsec:mean_fiel_model}

Before addressing the mean field setup, consider a finite financial system of size $n=m_0$ consisting of $m_\mathrm{c}$ core banks and $m_\mathrm{p}:=m_0-m_\mathrm{c}$ peripheral banks, where the peripheral banks are defined by not having any liabilities towards each other. In other words, the liabilities matrix for the system can be written in the block form
\begin{equation}\label{eq:lambda_block}
\lambda_{m_0\times m_0} = \left( \begin{array}{cc}
A & B\\{}
C & 0
\end{array}\right) = \left(\begin{array}{cc}
A_{ m_\mathrm{c} \times m_\mathrm{c}} & B_{m_\mathrm{c}  \times m_\mathrm{p} }\\{}
C_{m_\mathrm{p}\times m_\mathrm{c} } & 0
\end{array}\right).
\end{equation}
We could also work with sparse connections between the peripheral banks, but the idea here is to keep the model simple and focus on the core-periphery interactions, so we simply zero out the periphery-to-periphery interactions in line with the discussion in Section \ref{sec:model} above.

In general, the $\lambda_{ij}$'s can be completely different for each pair of banks $(i,j)$, but it is natural to suppose that they are nonetheless representative of some underlying structure in terms of how the core and peripheral banks interact. One tractable way of capturing this is to declare that
\begin{equation}\label{eq:lambda_noise}
\lambda_{ij}=(1+\epsilon_i)(1+\delta_j)\hat{\lambda}_{ij},
\end{equation}
where the $(\epsilon_i,\delta_i)$'s are random samples from $P\otimes P$, for some distribution $P$ with mean zero and support in $[-1,1]$ (or similar), and the $\hat{\lambda}_{ij}$'s are the fixed entries of a nicer matrix $\hat{\lambda}_{m_0\times m_0}$, which defines the underlying structure of the network. For concreteness, let us consider the specific example
\begin{equation}\label{eq:lambda_hat_concrete}
\hat{\lambda}_{m_0\times m_0} 
:= \left(\begin{array}{cc}
\hat{A} & \hat{B}\\{}
\hat{C} & 0
\end{array}\right) :=
\left(\begin{array}{ccc} \hspace{-20pt}

\left(\begin{array}{cc}
0 & 15\\{}
45 & 0
\end{array}\right)

& \left(\begin{array}{ccc}
0 & \cdots & 0\\{}
3 & \cdots & 3
\end{array}\right)_{2\times m_{\mathrm{p},1}} 

& \left(\begin{array}{ccc}
3 & \cdots & 3\\{}
1 & \cdots & 1
\end{array}\right)_{2\times m_{\mathrm{p},2}}\\[12pt]

\left(\begin{array}{cc}
5 & 2\\{}
\vdots  & \vdots \\{}
5 & 2
\end{array}\right)_{m_{\mathrm{p},1}\times 2} & 0 & 0 \\[20pt]

\left(\begin{array}{cc}
4 & 3\\{}
\vdots & \vdots \\{}
4& 3
\end{array}\right)_{m_{\mathrm{p},2} \times 2}

& 0 & 0
\end{array}\right).
\end{equation}
In this case, there are \emph{two core banks} (i.e., $m_\mathrm{c}=2$) and the \emph{peripheral banks} can be divided into \emph{two groups} (of size $m_{\mathrm{p},1}$ and $m_{\mathrm{p},2}$ with $m_{\mathrm{p}}=m_{\mathrm{p},1}+m_{\mathrm{p},2}$)  in terms of how they interact with the core. Nevertheless, the real connections are subject to noise---modelled by \eqref{eq:lambda_noise}---and hence $\lambda_{m_0\times m_0}$ can feature much more asymmetry in the core-periphery interactions.

\subsubsection{Growing the number of banks to infinity}\label{subsec:growing_number_banks}

Starting from the above system of size $m_0$, we now introduce a natural way of growing it to infinity. Based on the `initializing' system of size $m_0$, for each $m\geq1$, the idea is to construct a system of size $n=mm_0$ according to the following procedure:
\begin{itemize}
	\item multiply each of the $m_\mathrm{c}$ core banks into $m$ analogous entities (that can be seen as sub-entities comprising a core bank of $m$ times the size of the original), for a total of $mm_\mathrm{c}$ core entities
	\item multiply each of the $m_\mathrm{p}$ peripheral banks into $m$ analogous peripheral entities, for a total of $mm_\mathrm{p}$ peripheral entities
	\item let the external assets of each entity have an i.i.d.~copy of the same initial condition as well as the same drift and volatility as the original bank up to an i.i.d.~noise.
	\item impose that the $m$ sub-entities of a given core bank do not have liabilities towards each other (which is enforcing no self-dealing within the core bank)
	\item impose that, up to noise, the liability positions between a given core and peripheral entity are the same as those between the original core and peripheral bank only scaled by $m^{-1}$ (meaning that the underlying network structure is preserved and, the noise aside, each entity has the same total liabilities as the original bank of which it is a copy)
\end{itemize}

To be precise, starting from an underlying matrix $\hat{\lambda}_{m_0\times m_0}$ as in the example \eqref{eq:lambda_hat_concrete}, we construct the $n=mm_0$'th system by first fixing the \emph{underlying} network structure through the mapping
\begin{equation}\label{update_lambda_hat}
\hat{\lambda}_{m_0\times m_0} 
= \left(\begin{array}{cc}
\hat{A} & \hat{B}\\{}
\hat{C} & 0
\end{array}\right) \quad \longmapsto \quad \hat{\lambda}_{mm_0\times mm_0} := \frac{1}{m}\left(\begin{array}{cccccc}
\hat{{A}} & \cdots & \hat{{A}} & \hat{{B}} & \cdots & \hat{{B}}\\
\vdots & \ddots & \vdots & \vdots & \ddots & \vdots\\
\hat{{A}} & \cdots & \hat{{A}} & \hat{{B}} & \cdots & \hat{{B}}\\
\hat{{C}} & \cdots & \hat{{C}} & 0 & \cdots & 0\\
\vdots & \ddots & \vdots & \vdots & \ddots & \vdots\\
\hat{{C}} & \cdots & \hat{{C}} & 0 & \cdots & 0
\end{array}\right),
\end{equation}
and then the liabilities matrix $\lambda_{mm_0\times mm_0}$ is defined by setting
\begin{equation}\label{eq:lambda_random}
\lambda_{ij}:=(1+\epsilon_i)(1+\delta_j)\hat{\lambda}_{ij}
\end{equation}
as in \eqref{eq:lambda_noise}. Here the $\hat{\lambda}_{ij}$'s are now the entries of $\hat{\lambda}_{mm_0\times mm_0}$ given in \eqref{update_lambda_hat}, and the $(\epsilon_i,\delta_i)$'s are random samples drawn from the distribution $P\otimes P$, for a given probability measure $P$.

We stress that, due to the noise, the entries of $\lambda_{mm_0\times mm_0}$ can be entirely heterogeneous both across and within groups. In particular, a given sub-entity of the first core bank may interact differently with all entities representing the second core bank, and any given core entity may interact differently with all peripheral entities across the two groupings. Nevertheless, by passing to the mean field limit we may hope to discover the underlying structure as defined by the `initializing' matrix $\hat{\lambda}_{m_0\times m_0}$. The remaining subsections illustrate this for the specific example provided by \eqref{eq:lambda_hat_concrete}.

\subsection{The dynamics of the finite interbank system}\label{subsubsec:dynamics_finite_sys}
Returning to the Eisenberg--Noe model from Section \ref{sec:finite}, consider a finite system of size $n$, and note that the capital (\ref{eq:Kapital}) of each bank $i=1,\ldots,n$ can be written as a coupled system
\begin{align}\label{1st_K_i}
K_i(t) = x_i(t) e^{\int_t^T \! \mu_i (s)ds  } - \bar{L}_i(T) - T (1-R_2) \int_0^t (1-{\textstyle\frac{s}{T}}) d\loss^{\!n}_i(s),
\end{align}
where 
\begin{equation}\label{eq:ith_loss_process}
\loss^{\!n}_i(t):= \sum_{j=1}^{n} \lambda_{ji} \mathbf{1}_{t\geq \tau_j} \quad \text{with} \quad \tau_j=\inf\{ t\geq0 : K_i(t) < 0  \}.
\end{equation}
For simplicity, we will assume that $\bar{L}_i(T)=T\Lambda_i $ for some constants $\Lambda_1,\ldots,\Lambda_{n}>0$, meaning that, for each bank $i$, its total liabilities net of interbank assets over the period $[0,T]$ is given by the positive amount $T\Lambda_i$. In particular, if a bank is a net lender in the interbank market, then the surplus is more than offset by external liabilities, which is in line with what is observed in practice. Recalling that each $x_i(t)$ is a geometric Brownian motion, it is convenient to work with the following logarithmic `distances-to-default' defined by
\begin{equation}\label{eq:dist-to-def}
X_i(t) := \log \biggl\{    \frac{x_i(t) \exp\{ \int_t^T\!\mu_i(s)ds\}  }{\Lambda_i T + T(1-R_2)  \int_0^t (1-{\textstyle\frac{s}{T}}) d\loss^{\!n}_i(s) } \biggr\},
\end{equation}
for $i=1,\ldots,n$. This transforms the system \eqref{1st_K_i}-\eqref{eq:ith_loss_process} into the equivalent system
\begin{equation}\label{dist_to_default_system}
\left\{ \begin{aligned}
\loss^{\!n}_i(s) &= \sum_{j=1}^{n} \lambda_{ji} \mathbf{1}_{t\geq\tau_j}, \quad   \tau_j=\inf \{ t \geq 0 :  X_j(t) < 0  \} 
\\
dX_i(t) &= -\frac{\sigma_i(t)^2}{2}dt + \sigma_i(t) dW_i(t) - d\log\Bigl\{ 1 + \frac{(1-R_2)}{\Lambda_i} \int_0^t (1-{\textstyle\frac{s}{T}}) d\loss^{\!n}_i(s) \Bigr\} \\
X_i(0)&=\log \{ x_i(0) \} -\log \{\Lambda_i T\} +\int_0^T\!\mu_i(t)dt,
\end{aligned} \right.
\end{equation}
where we recall that $x_i(0)e^{\int_0^T\!\mu_i(t)dt}>\bar{L}_i(T)=\Lambda_i T$, by Assumption \ref{ass:initial-default}, which guarantees $X_i(0)>0$.
Moreover, we will assume that the Brownian motions are only correlated through a common noise, meaning that we can write $W_i(t)=\rho B_0(t) + \sqrt{1-\rho^2}B_i(t)$ for independent Brownian motions $B_0,\ldots B_n$.

Since the default times $\tau_i$ are part of the equations for the distances-to-default $X_i$, one has to be careful that there can be several solutions to \eqref{dist_to_default_system} depending on how one decides if a bank is in default at time $t$. For example, even if $X_i(t\shortminus)>0$ for all the banks, one may succeed in defaulting a few---or even all---of them at time $t$, provided the corresponding increase of the $\mathbf{L}^n_i(t)$'s make $X_i(t)$ drop below zero for precisely the banks we decided to default, where $X_i(t):=X_i(t\shortminus)-\{\text{jump from increase in } \mathbf{L}^n_i(t)\}$. Moreover, if it is indeed the case that $X_i(t\shortminus)\leq 0$ for some bank $i$, then we need to decide (in a way that is consistent with the equations) how this propagates as it may start a cascade of defaults at the same time $t$, for which there can again be multiple possible choices (much in line with the previous example).

The solution we choose to work with here amounts to picking the solution that gives the greatest clearing capital in the Eisenberg--Noe framework (see Lemma \ref{lemma:exist}) with any instantaneous default cascades resolved by an analogue of the Eisenberg--Noe fictitious default algorithm. In Section \ref{Sec:determine_cascades} we show how this corresponds to amending the particle system \eqref{dist_to_default_system} with what we call the \emph{cascade condition}---see \eqref{eq:finite_PJC_fragile} for its precise definition and derivation, albeit in a more general setting than the specific example considered here. This condition is intrinsic to the particle system formulation of our interbank model, and it uniquely determines the loss processes $\loss^{\!n}_i$ at `time $t$' given the state of the system immediately before, namely at `time $t\shortminus$' in the sense of taking a left limit. In particular, this ensures that \eqref{dist_to_default_system} has a unique strong c\`adl\`ag solution, as argued in the proof of Proposition \ref{prop:particle_sys_well-posed}. We will not discuss this condition any further here, but we briefly present its mean field analogue in Section \ref{subsubsec:concrete_cascades} below.

\subsection{The dynamics of the mean field limit}\label{subsec:dyn_mean_field}

For any given $m\geq1$, and a fixed initial size $m_0$, we will now consider the interbank system of size $n=mm_0$ modelled by \eqref{dist_to_default_system}, where the liabilities matrix and the other parameters are noisy realisations of the underlying core-periphery network structure defined by the concrete example \eqref{eq:lambda_hat_concrete}. Specifically, we impose that:
\begin{itemize}
	\item the liabilities matrix $\lambda_{mm_0\times mm_0}$ is constructed  from \eqref{eq:lambda_hat_concrete} via random samples $(\epsilon,\delta)$ from the distribution $P\otimes P$ as outlined in the previous subsection.
	\item the $i$'th set of parameters $(x_i(0),\sigma_i,\mu_i,\Lambda_i)$ is given as a function of the $i$'th random sample $\delta_i$, where the function is the same for all entities of the same type (out of the two core types and two peripheral types defined by  \eqref{eq:lambda_hat_concrete}).
\end{itemize}

Based on the analysis in Section \ref{sec:well-posedness}, it follows that the system we just described has a well-defined mean field limit as $n\rightarrow\infty$ (see, in particular, Section \ref{subsec:example_sect4}). This limit captures the coupled evolution of the four underlying types of banks (two core and two peripheral) after averaging over the infinitely many entities within each type. Let $I \subset[-1,1]$ denote the support of the distribution $P$, and let $\theta\mapsto (\sigma_{l,\theta},\mu_{l,\theta}, \Lambda_{l,\theta})$  denote the parameter function for each of the four types $l=1,\ldots,4$. Let us say that $l=1,2$ are the two core types and $l=3,4$ are the two peripheral types (in correspondence with $m_\mathrm{c}=1+1$ and $m_\mathrm{p}=m_\mathrm{p,1}+m_\mathrm{p,2}$ in \eqref{eq:lambda_hat_concrete}). To simplify the presentation of the mean field limit below, we write
\[
Y_{l,\theta}(t) :=  - \int_0^t \!\frac{\sigma_{l,\theta}(s)^2}{2}ds + \int_0^t\!\sigma_{l,\theta}(s) d(\rho B_0(s) + \sqrt{1-\rho^2}B_l(s)),
\]
and
\[
C_{l,\theta}:= (1+\theta)\frac{1-R_2}{\Lambda_{l,\theta}},
\]
for $l=1,\ldots,4$, and $\theta \in I$, where $B_0$ and $B_1,\ldots,B_4$ are independent Brownian motions. 

As the number of banks grows to infinity (in accordance with Section \ref{subsec:growing_number_banks}), the results of Section \ref{sec:well-posedness} below show that mean field limit of the finite interbank system is given by the coupled McKean--Vlasov problem
\begin{equation}\label{eq:concrete_mean_field_limit}
\begin{cases}
\displaystyle{\widetilde{\mathbf{L}}_{l}(t)  = \int_{I} \int_{0}^\infty   \mathbb{P}( t\geq \tau^x_{l,\theta} \mid B_0) V_{0}^l(x\,|\,\theta) dxdP(\theta)}, \;\;\; \tau^x_{l,\theta} =\inf \{ t\geq 0 : X^x_{l}(t) \leq 0 \},   \\[10pt]
\displaystyle X^{x}_{l,\theta}(t) = x + Y_{l,\theta}(t) - \log \Bigl( 1 + C_{l,\theta} \sum_{i=1}^4\tilde{\lambda}_{il}  \int_0^t \!(1-\tfrac{s}{T})d\widetilde{\mathbf{L}}_i(s)  \Bigr),  \end{cases}
\end{equation}
for $l=1,\ldots,4$, where the strength of the core-core and core-periphery interactions are fully captured by the simplified liabilities matrix
\begin{equation}\label{eq:lambda_tilde}
\tilde{\lambda}_{4 \times 4} = \left( \begin{array}{cccc}
0 & 15 &   0 m_{\mathrm{p},1} &  3m_{\mathrm{p},2}\\{}
45 & 0 &  3 m_{\mathrm{p},1} &  1m_{\mathrm{p},2}\\{}
5m_{\mathrm{p},1}  & 2m_{\mathrm{p},1} & 0 & 0\\{}
4m_{\mathrm{p},2}  &  3m_{\mathrm{p},2} & 0 & 0
\end{array}\right),
\end{equation}
and $V_{0}^l(\cdot\,|\,\theta)$ are the  initial densities for the distances-to-default of the four types $l=1,\ldots,4$ conditional on $\theta\in I$. See \eqref{eq:initial_condition} below for how these initial conditions relate to the parameters and the initial laws of the external assets.

Note that the contagion in \eqref{eq:concrete_mean_field_limit} is no longer felt as the result of a single default event. Instead, there are now four `infinite collections' of entities (corresponding to the four underlying types) who feel the contagion through the mutual exposures $\tilde{\lambda}_{ij}$ in relation to the \emph{proportion of defaults} within each infinite collection (given by the loss processes $\widetilde{\mathbf{L}}_{l}$, for $l=1,\ldots,4$). In the McKean--Vlasov formulation \eqref{eq:concrete_mean_field_limit}, these \emph{proportions} of default are really `average' \emph{probabilities} of default for the entities of each type, but see also the SPDE formulation \eqref{eq:loss_SPDE}-\eqref{eq:system_SPDE} below which makes the interpretation in terms of \emph{proportions} more explicit.

\begin{remark} The relative number of core and peripheral entities are specified by $m_0=m_\mathrm{c}+m_\mathrm{p}$, where $m_\mathrm{c}=2$ and $m_\mathrm{p}=m_\mathrm{p,1}+m_\mathrm{p,2}$. For a finite system of any size, the $m_\mathrm{c}=2$ collections of core sub-entities each comprise a fraction $\frac{1}{m_0}$ of the system, while a fraction $\frac{m_\mathrm{p,1}}{m_0}$ makes up the first collection of peripheral banks, and the final fraction $\frac{m_\mathrm{p,2}}{m_0}$ makes up the second collection of peripheral banks. As a result, in the mean field limit we have that: (i) the core feels the contagion from a given \emph{proportion} of defaults within the two peripheral groups at a strength multiplied by $m_\mathrm{p,1}$ and $m_\mathrm{p,2}$, respectively, and, similarly, (ii)  the peripheral groups feel contagion from the core at a strength multiplied by $m_\mathrm{p,1}$ and $m_\mathrm{p,2}$.
\end{remark}

\subsubsection{Stochastic evolution equations for the densities}\label{subsubsec:stoch_evolution}
Consider, for simplicity of presentation, the case where $P$ is a Dirac mass at zero (so $\theta$ drops from the equations), meaning that there is no additional heterogeneity within each of the four types (for practical purposes, one can think of having replaced the parameters by their mean values averaged over $P$). By applying It\^o's formula, and taking expectations conditional on $B^0$, we can reformulate \eqref{eq:concrete_mean_field_limit} as a system of four coupled (nonlinear and nonlocal) stochastic partial differential euqations (SPDEs). These SPDEs govern the (stochastic) densities of the distances-to-default for the four infinite collections of banks of a given type (conditional on the common noise $B^0$). This is arguably the more natural point of view for the dynamics of the mean field limit. Specifically, we have
\begin{equation}\label{eq:loss_SPDE}
\widetilde{\mathbf{L}}_{l}(t)= 1- \int_0^\infty V^l_t(x)dx, \quad \mathrm{for} \quad l=1,\ldots,4,
\end{equation}
where $V=(V^1,\ldots,V^4)$ solves a coupled system of SPDEs on the positive half-line of the form
\begin{align}\label{eq:system_SPDE}
dV_t^l(x) =  \frac{\sigma_l^2}{2}  \bigl( \partial^2_{xx} V^l_t(x) - \partial_x V^l_t(x) \bigr)dt - \sum_{i=1}^4 \tilde{\lambda}_{il} f_l(t)  \partial_x V^l_t(x) d \widetilde{\mathbf{L}}_{i}(t)+ \rho\sigma_l \partial_x V^l_t(x) dB_0(t),
\end{align}
with the Dirichlet boundary condition $V_t^l(0)=0$, for each $l=1,\ldots,4$. Note that this point of view makes clear the precise nature of the contagion: namely a nonlinear transportation of mass towards the origin, at a rate that is proportional to the current rates of default within each infinite collection of banks (as mediated by the mutual exposures $\tilde{\lambda}_{ij}$ between the four infinite collections). Indeed, in $dt$ amounts of time, the proportion of defaults within the $l$'th collection of banks is precisely $d\widetilde{\mathbf{L}}_{l}(t)$, since $\widetilde{\mathbf{L}}_{l}(t)$ gives the total loss of mass for the $l$'th collection of banks up to and including time $t$ (i.e., the accumulated proportion of defaults).

We note that, due to the irregularity in time of the common noise $B^0$, the time derivative of $\widetilde{\mathbf{L}}_{l}(t)$ does not exist if $\rho>0$, but the process is increasing, so the integrals against it are well-defined. Still, in order for the SPDE formulation \eqref{eq:system_SPDE} to make sense globally, as it is, we are implicitly relying on each $\widetilde{\mathbf{L}}_l$ being continuous. As we already discussed above, this may be violated, meaning that one (or more) of the loss processes $\widetilde{\mathbf{L}}_l$ can undergo a jump discontinuity, corresponding to an instantaneous macroscopic default cascade within the infinite collection of the $l$'th type (or types). Nevertheless, one can still attach a rigorous meaning to the SPDE, as long as it is understood to only hold on the random intervals between jump times in the following sense: at a jump time $t$, the densities are shifted according to the jump size, and thus the system of SPDE is restarted from the new set of initial conditions
\begin{equation}\label{eq:new_initial_SPDE}
V_t^l(x):=V_{t\shortminus}^l\bigl(x+\Theta_l(t,\Delta \mathbf{L}_l(t)) \bigr), \qquad x\in\mathbb{R}_{+},
\end{equation}
where $V_{t\shortminus}^l$ is the pointwise left-limit of $V_{s}^l$ as $s\uparrow t$,
\[
\Theta_l(t,z):=  \log \Bigl( 1 + C_{l}\! \int_0^{t\shortminus} \!(1-\tfrac{s}{T})d\mathbf{L}_l(s) + C_{l}(1-\tfrac{t}{T})z \Bigr) - \log \Bigl( 1 + C_{l}\! \int_0^{t\shortminus} \!(1-\tfrac{s}{T})d\mathbf{L}_l(s) \Bigr),
\]
and
\[
\mathbf{L}_l(t):= \sum_{i=1}^4 \tilde{\lambda}_{il}  \widetilde{\mathbf{L}}_{i}(t), \qquad \text{for} \quad l=1,\ldots,4.
\]
Note that we must allow the Dirichlet boundary condition to be violated when restarting at a jump time (and, as Remark \ref{eq:dirichlet_vs_jump} below points out, it is also a loss of the Dirichlet condition that leads to a jump). As concerns the timing and the sizes of the jumps, these are defined (in a c\`adl\`ag fashion) by what we call the \emph{mean field cascade condition}, namely
\begin{equation}\label{eq:PJC_SPDE}
\begin{cases}
\Delta \mathbf{L}_l(t) = \lim_{\varepsilon\downarrow 0} \lim_{m \uparrow \infty} \Delta^{\!(m,\varepsilon)}_{t,l}, \\[5pt] \Delta^{\!(m,\varepsilon)}_{t,l}= \Xi_l\bigl(t, \varepsilon +  \Delta^{\!(m-1,\varepsilon)}_{t,\hspace{0.5pt}\cdot}\bigr), \quad m\geq1, \\[5pt]
\Delta^{\!(0,\varepsilon)}_{t,l}= \Xi_l(t,\varepsilon),
\end{cases}
\end{equation}
where
\[
\Xi_l(t,z)=\Xi_l(t,z_1,\ldots,z_4):= \sum_{i=1}^4 \tilde{\lambda}_{il} \int_0^{\Theta_i(t,z_i)} V_{t\shortminus}^i(x)dx.
\]
This condition for the jumps is the mean field analogue of the cascade condition for the finite system discussed at the end of Section \ref{subsubsec:dynamics_finite_sys}. Intuitively, it amounts to subjecting the system to an arbitrarily small shock that ignites a fictitious default cascade and then keeping track of how it propagates in relation to the size of the initial shock: as we send the size of the initial shock to zero, either the size of the fictitious cascade goes to zero, and there is then no jump, or it converges to something positive, and this positive value is then the size of the jump corresponding to a bona fide instantaneous default cascade. The mean field cascade condition is carefully developed and motivated in Section \ref{sec:well-posedness}. It is a special case of the condition \eqref{eq:limit_PJC} in Section \ref{the_mean_field_limit}, which addresses a more general framework than the one considered in this section.

\begin{remark}\label{eq:dirichlet_vs_jump}
	As we note in Section \ref{subsubsec:concrete_cascades} below, the cascade condition gives $\Delta{\mathbf{L}}_l(t)=0$ for every $l=1,\ldots,4$, whenever each left-limit density $V^l_{t\shortminus}(x)$ vanishes as $x\downarrow 0$. More generally, there is no jump at time $t$ provided $\Xi_l(t,\epsilon)<\epsilon$ for small enough $\epsilon>0$, for each $l=1,\ldots,4$, as follows by the same arguments as in Section \ref{criterions_jump_or_not}. To see how this condition being violated can lead to a jump, consider the case where, at some time $t$, we have
	\begin{equation}\label{eq:condition_for_jump}
	\int_0^{\Theta_i(t,\epsilon)} V_{t\shortminus}^i(x)dx \geq \tilde{\lambda}_{ij}^{-1} \epsilon \quad \text{and} \quad 	 \int_0^{\Theta_j(t,\epsilon)} V_{t\shortminus}^j(x)dx \geq \tilde{\lambda}_{ji}^{-1} \epsilon,
	\end{equation}
	for small enough $\epsilon>0$, for some pair of banks $i,j\in\{1,\ldots,4\}$ with $\tilde{\lambda}_{ij}>0$ and $\tilde{\lambda}_{ji}>0$, meaning that banks in the $i$'th and $j$'th groups are exposed to each other (with $i=j$ being a possibility, provided banks within the same group are exposed to each other in a way that is significant in the mean field limit). Now suppose for a contradiction that $\Delta \mathbf{L}(t)=0$. Then the cascade condition implies that we can make $\lim_{m \uparrow \infty} \Delta^{\!(m,\varepsilon)}_{t,\cdot}$ as small as we like (since it vanishes as $\epsilon\downarrow0$). Thus, \eqref{eq:condition_for_jump} together with the dominated convergence theorem gives
	\[
	\lim_{m \uparrow \infty} \Delta^{\!(m,\varepsilon)}_{t,i}= \Xi_i\bigl(t, \varepsilon + \lim_{m \uparrow \infty} \Delta^{\!(m,\varepsilon)}_{t,\cdot} \bigr) \geq \epsilon + \lim_{m \uparrow \infty} \Delta^{\!(m,\varepsilon)}_{t,j},
	\]
	for small enough $\epsilon>0$, and the same conclusion holds with $i$ and $j$ interchanged. Together, these two inequalities yield a contradiction, and hence we conclude that there must indeed be a jump. With $(i,j)=(1,4)$ this corresponds to the situation at the jump time in Figure \ref{fig:heat_plots}. Of course, there is nothing sacred about the size $n=4$, and, unlike the particular interactions in \eqref{eq:lambda_tilde}, we could in general have a nonzero diagonal, so $i=j$ is perfectly valid if the core sub-entities within a given collection are exposed to contagion from each other.
\end{remark}

In order to better illustrate the dynamics of the mean field limit, we present a numerical simulation of the system of SPDEs \eqref{eq:loss_SPDE}-\eqref{eq:system_SPDE} with jumps governed by the mean field cascade condition \eqref{eq:PJC_SPDE} via \eqref{eq:new_initial_SPDE}. The outcome is plotted in Figure  \ref{fig:heat_plots}, which shows a heat plot for each of the four solutions $(t,x)\mapsto V^l_t(x)$ to the coupled system of SPDEs. The simulation is performed using an adaptation of the numerical scheme proposed in \cite[Sect.~4.2]{LS18a}.

\begin{figure}[H]
	\centering
	\begin{minipage}[b]{0.49\linewidth}
		\includegraphics[width=\textwidth]{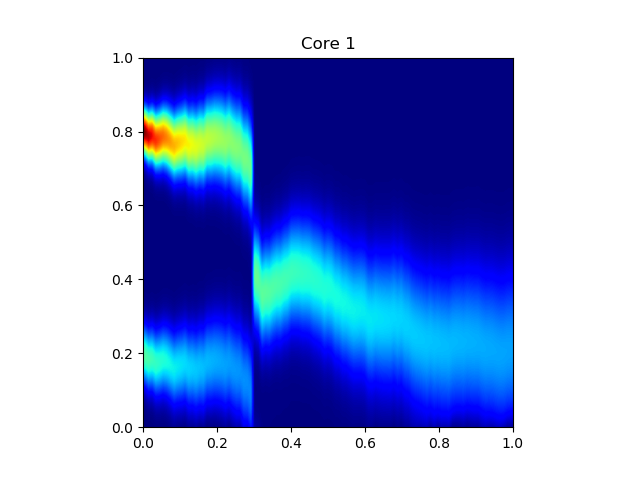}
	\end{minipage}\hfill
	\begin{minipage}[b]{0.49\linewidth}
		\includegraphics[width=\textwidth]{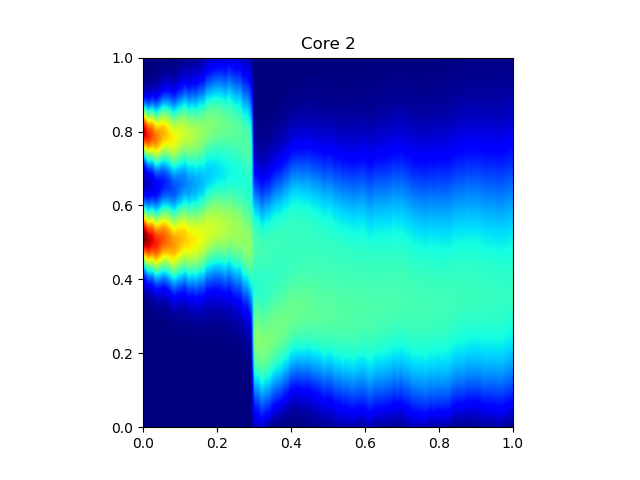}
	\end{minipage}
	\begin{minipage}[b]{0.49\linewidth}
		\includegraphics[width=\textwidth]{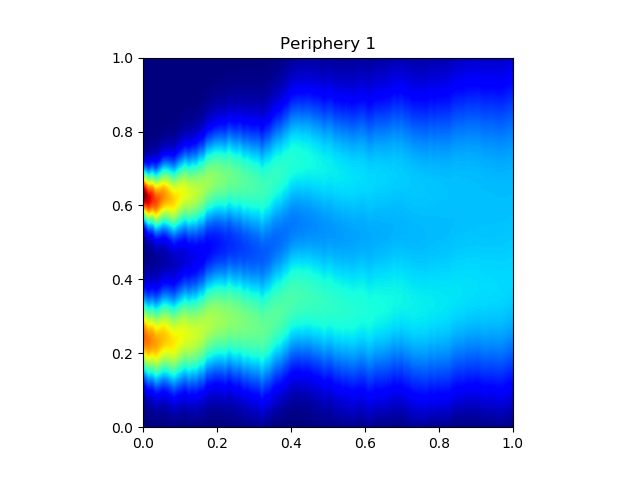}
	\end{minipage}\hfill
	\begin{minipage}[b]{0.49\linewidth}
		\includegraphics[width=\textwidth]{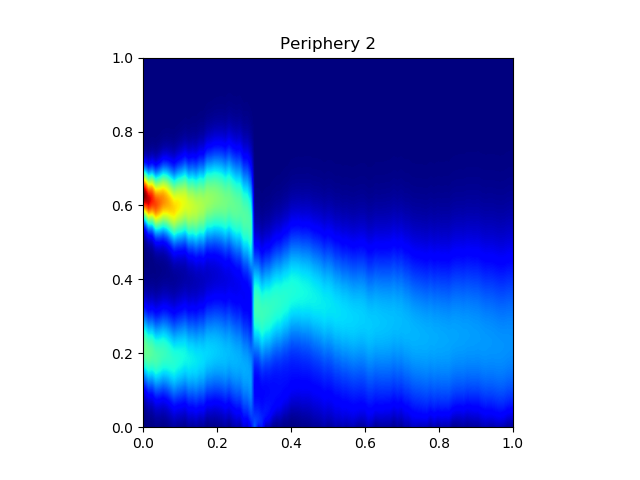}
	\end{minipage}
	\caption{The figure shows four heat plots for the distances-to-default of the four infinite collections (for a given realisation of the common noise $B^0$), where the horizontal axis is time, and the vertical axis is the distance-to-default. The interactions are given by $\tilde{\lambda}_{4\times4}$ in \eqref{eq:lambda_tilde} with $m_{\mathrm{p},1}=m_{\mathrm{p},2}=4$, as in Example \ref{ex:core-periphery}. The parameters are constant (with `Periphery 1' and `Core 2' having a more positive drift), and the initial conditions can be read off the heat plots at time $t=0$. The common noise starts out on a slight negative trend, which instigates a default cascade between the low performing fractions of `Core 1' and `Periphery 2', resulting in both fractions defaulting in their entirety. Moreover, these defaults spill into a severe downgrading of the financial health of `Core 2'. However, `Core 2' was otherwise performing well, so
		only a very small proportion of it defaults, and since `Periphery 1' is only exposed to defaults in `Core 2', this means that these events have no significant impact on `Periphery 1'.}
	\label{fig:heat_plots}
\end{figure}

In terms of the related mathematical literature, we stress that the papers \cite{DNS19, HLS18, LS18a, LS18b, NS17} are focused on `one-dimensional' variations of McKean--Vlasov problems akin to \eqref{eq:concrete_mean_field_limit}, whereas the recent paper \cite{NS18} studies a coupled system analogous to \eqref{eq:concrete_mean_field_limit} with only minor differences. In particular, \cite{NS18} provides an existence result based on a Schauder fixed point argument (but no results on uniqueness) and studies criteria under which any solution to the system must incur a blow-up. However, unlike the present paper, the results in \cite{NS18} neither address the relation to a finite particle system, nor do they consider a condition for uniquely specifying the jump sizes (in contrast to our cascade condition).

\subsubsection{On the ruling out of instantaneous default cascades}\label{subsubsec:concrete_cascades}
Due the averaging effect of passing to the mean field limit, one could reasonably expect the limiting loss processes $\widetilde{\mathbf{L}}_{l}$ to evolve continuously, and anything else would be somewhat surprising given that the McKean--Problem is driven by continuous Brownian dynamics. In many cases, it will indeed be true that the system evolves continuously. However, as we have just seen in Figure \ref{fig:heat_plots}, depending on the parameters, one or more of the loss processes may see their speed of increase diverge to infinity in a way that results in a jump discontinuity (see also \cite[Thm.~2.7]{LS18a} in a simplified setting). Naturally, such an event can be seen as defining an instantaneous `macroscopic' default cascade that survived the passage to the mean field limit.

In order to decide whether the solution is continuous or not, and in order to specify the size of a potential jump, we must amend \eqref{eq:concrete_mean_field_limit} with the \emph{mean field cascade condition} introduced above. As already mentioned, the details of this are reserved for Section \ref{sec:well-posedness}, however, it is worth taking a few moments to preview a simple result on when jumps can be ruled out, which illustrates the workings of the cascade condition.

Section \ref{criterions_jump_or_not}
presents a simple criterion for the initial densities that rules out a jump immediately after initializing the system. As above, we consider the case where there is no dependence on $\theta$, and note that the initial densities are then of the form
\begin{equation}\label{eq:initial_condition}
V_{0}^l(x)=v^l_0\bigl(\Lambda_{l} Te^{-\int_0^T\mu_{l}(s)ds}e^x\bigr)\Lambda_{l} Te^{-\int_0^T\mu_{l}(s)ds}e^x, \qquad \text{for } x>0,
\end{equation}
where $v_0^l$ is the initial density for the external asset process of banks of type $l$ (which is supported on $x>\Lambda_{l} Te^{-\mu_{l}T}$).
If, for every $l=1,\ldots,4$, there is a small $\epsilon_{l}>0$ such that
\begin{equation}\label{eq:concrete_no_intial_jump}
(1-R_2)T\sum_{j=1}^4 \tilde{\lambda}_{jl}     v^l_0 (x)e^{- \mu_{l}T} < 1, \qquad \text{for all } x \in(\Lambda_{l} Te^{-\mu_{l}T},\Lambda_{l} Te^{-\mu_{l}T} +\epsilon_{l}),
\end{equation}
then there is \emph{not} an instant jump at time $t=0$ and the solution remains continuous for a small amount of time after initialization.

\begin{remark}
	If each $x\mapsto v_0^l(x)$ is continuous near the boundary $x=\Lambda_{l} Te^{-\mu_{l}T}$ and $v_0^l(x)$ vanishes as $x\downarrow\Lambda_{l} Te^{-T\mu_{l}T}$, then  clearly \eqref{eq:concrete_no_intial_jump} is satisfied. However, if $v_0^l(x)$ converges to something strictly positive as $x\downarrow \Lambda_{l} Te^{-\mu_{l}T}$, then the  
	values of the parameters become decisive.
\end{remark}

At any given time $t\geq0$, the mean field cascade condition \eqref{eq:limit_PJC} gives the precise criterion for whether or not there is a jump, and what the size of the jump is, if there is one. However, here we only note that there is a simple (non-optimal) time-$t$ analogue of \eqref{eq:concrete_no_intial_jump} for ruling out jumps at any given time $t$ and in some short time interval thereafter. To see what this looks like, let $V^l$ be given by \eqref{eq:system_SPDE}; that is, $V_s^l(x)$ denotes the density of solvent banks of type $l$ with distance-to-default $x$ at time $s$, for a fixed realisation of the common noise $B_0$. If, for each $l=1,\ldots,4$, there is a small $\epsilon_l$ such that
\begin{equation*}
(1-R_2)(T-t)\sum_{j=1}^4 \tilde{\lambda}_{jl}     V_{t\shortminus}^l (x)< 1, \qquad \text{for all } x\in(0,\epsilon_v),
\end{equation*}
then there is no jump at time $t$ and the solution is guaranteed to remain continuous for a short time thereafter. Note that the criterion involves the left limit $V_{t\shortminus}^l(x)=\lim_{s\uparrow t}V_s^l(x)$, meaning that it is based on the state of the system strictly before time $t$ (where the state of the system is given by the distance-to-default densities for the solvent banks of the four types). The reader is referred to Section \ref{the_mean_field_limit} for further details.

\section{Convergence and well-posedness of the mean field}\label{sec:well-posedness}
Recall that we transformed the capital (\ref{eq:Kapital}) of each bank into an interacting particle system \eqref{dist_to_default_system} based on the notion \eqref{eq:dist-to-def} of their logarithmic distances-to-default. The remaining part of the paper is dedicated to a careful analysis of this particle system and its mean field limit. In relation to the previous section, we carry out the analysis under a more general assumption on the coefficients and the structure of the liabilities matrix (as $n\rightarrow \infty$). We then show in Section \ref{subsec:example_sect4} how to obtain the core-periphery model of Section 4 as a special case of this framework.

\subsection{The finite interbank system}
To streamline the presentation, we will work with a general version of the system of interacting distances-to-default \eqref{dist_to_default_system}. That is, we will focus on general particle systems of the form
\begin{equation}\label{general_dist_to_def_sys}
\left\{ \begin{aligned}
X_i(t) &= X_i(0) -\int_0^t b_i(s)ds + \int_0^t\sigma_i(s) dW_i(s) - F \Bigl( \int_0^t g(s) d\loss^{\!n}_i(s) \Bigr)  \\
\loss^{\!n}_i(s) &= \sum_{j=1}^{n}  \lambda_{ji} \mathbf{1}_{t\geq\tau_j}, \quad   \tau_j=\inf \{ t\geq 0 :  X_j(t)\leq0  \},
\end{aligned} \right.
\end{equation}
where each $X_i$ denotes the distance-to-default of `bank $i$' as derived from the expression for  bank $i$'s capital \eqref{eq:Kapital-gen} in the dynamic Eisenberg--Noe framework of Section \ref{sec:finite}. We recall that the transformation from \eqref{eq:Kapital-gen} to an interacting system of distances-to-default was carried out in Section \ref{subsubsec:dynamics_finite_sys}. The precise assumptions for the particle system are outlined in what follows.

First of all, we will assume that, for large $n$, the rank of the liabilities matrices $\lambda_{n\times n}$ is bounded by some value $k$ (uniformly in large $n>k$). Then, for large $n>k$, we have a factorization of the form
\begin{equation}\label{eq:rank_fac}
n\lambda_{n\times n} = U_{n\times k}V_{k \times n}.
\end{equation}
Here the natural choice of factorization comprises the matrices $U_{n\times k}:=(\tilde{u}_{ij})$ and $V_{k\times n}:=(\varsigma_i\tilde{v}_{ij})$ built from the singular value decomposition
$n\lambda_{n \times n}= \tilde{U} \textrm{diag}(\varsigma) \tilde{V}$, where $\textrm{diag}(\varsigma)$ is the $n\times n$ diagonal matrix with the singular values $\varsigma_1,\ldots,\varsigma_n$ on the diagonal (out of which no more than the first $k$ values are nonzero, since the rank is bounded by $k$). 

Spectral decompositions and low rank structures are omnipresent in statistical analysis and the applied sciences more generally. In relation to financial networks and systemic risk, simple aspects of this has, e.g.,~been utilised in contagion models \cite{amini_minca, cont_schaanning} and statistical methods for detecting core-periphery network structures \cite{porter_2016}. More recently, the preprint \cite{spiliopoulos_2019} studies a reduced form model for default clustering (based on interacting default intensities), using a singular value decomposition of the adjacency matrix in a way that is completely analogous to what we do here; namely to study the large population limit of the system under a bounded rank assumption which allows for a more tractable reformulation of the interactions.

\begin{example}
Suppose the liabilities matrix $\lambda_{n\times n}$ is constructed from an underlying matrix $\hat{\lambda}_{m_0\times m_0}$, as in \eqref{update_lambda_hat}--\eqref{eq:lambda_random}, where $\hat{\lambda}_{m_0\times m_0}$ is of the block form \eqref{update_lambda_hat}. Then the rank of $\hat{\lambda}_{m_0\times m_0}$ is at most $2m_0$, and one easily verifies that the rank of $\lambda_{n\times n}$ also stays bounded by $2m_0$ for any system of size $n=mm_0$, for all multiples $m\geq1$. This yields a particular example where the rank remains bouned as $n\rightarrow \infty$. We return to this in Section \ref{subsec:example_sect4}, where we detail how the model in Section 4 appears as a special case of the analysis presented here.
\end{example}

Note that \eqref{eq:rank_fac} amounts to 
\begin{equation*}
n\lambda_{ij}=\sum_{l=1}^k u_{il} v_{lj}, \qquad \text{for every} \quad i,j=1,\ldots,n,
\end{equation*}
where $u_{il}$ is $(i,l)$-entry of $U_{n \times k}$ and $v_{lj}$ is the $(l,j)$-entry of $V_{k\times n}$. Based on this, the utility of \eqref{eq:rank_fac} lies in the simple fact that
we can now decompose the processes $\loss^{\!n}_i$ from \eqref{eq:ith_loss_process} as
\begin{equation}\label{eq:k_loss_processes}
\loss^{\!n}_i(t) = \sum_{l=1}^k v_{li} \mathcal{L}^n_l(t), \quad \text{where} \quad \mathcal{L}^n_l(t):= \frac{1}{n} \sum_{j=1}^{n} u_{jl} \mathbf{1}_{t\geq\tau_j}, \quad \text{for} \quad i=1,\ldots,n.
\end{equation}
Crucially, these new \emph{loss processes} $\mathcal{L}^n_l$, for $l=1,\ldots,k$, do not depend on $i$ and, equally important, the number of them, namely $k$, is fixed as $n\rightarrow \infty$.

In order to make precise the financial meaning of \eqref{eq:k_loss_processes}, we interpret the  entries of $U_{n\times k}$ and $V_{k\times n}$ as latent factors identifying $k$ underlying \emph{channels of contagion} in the network structure (independently of the size $n$):
\begin{itemize}
	\item $u_{jl}$ captures how strongly bank $j$ contributes to the contagion of channel $l$, and
	\item $v_{li}$ captures how exposed bank $i$ is contagion from channel $l$
\end{itemize}

Let $\textbf{u}_1,\ldots,\textbf{u}_n \in \mathbb{R}^k$ denote the $n$ row vectors of $U_{n \times k}$ and $\textbf{v}_1,\ldots,\textbf{v}_n \in \mathbb{R}^k$ denote the $n$ column vectors of $V_{k\times n}$. That is,
\begin{equation}\label{eq:vectors_u_v}
\textbf{u}_i:=(u_{i1},\ldots,u_{ik}) \quad \text{and} \quad \textbf{v}_i:=(v_{1i},\ldots,v_{ki}), \qquad \text{for} \quad i=1,\ldots,n.
\end{equation}
Then bank $i$ is characterized by the pair of $k$-dimensional vectors $\mathbf{u}_i$ and $\mathbf{v}_i$, detailing, respectively, how it contributes to each of the $k$ (latent) channels of contagion and how it is impacted by them. Nevertheless, once we have identified the $k$ channels of contagion, the vector $\mathbf{v}_i$ alone can be seen as identifying bank $i$  in terms of how it is hit by contagion: if two banks have similar $\mathbf{v}_i$'s, they are similar in this crucial sense (although they may of course be dissimilar in terms of how strongly they contribute to contagion overall and to each of the various channels).

\begin{remark} In Section 4 we considered a specific core-periphery structure where the peripheral groups could be identified strictly by how they interact with the core (via the underlying matrix $\hat{\lambda}_{m_0\times m_0}$). In practice, the interbank liabilities may comprise a perturbation of this structure which is more heterogeneous (in addition to the noisiness) but nonetheless still of low rank (e.g.~due to asymmetric but sparse periphery-to-periphery connections). Thus, we may not have a small number of clear-cut groups as in Section 4, but the low rank (uniformly in $n$) would still allow the system to be decomposed into a small number of latent channels of contagion.
\end{remark}

Relying on the decomposition \eqref{eq:k_loss_processes}, the particle system \eqref{general_dist_to_def_sys} is transformed to take the form

\begin{equation}\label{general_system_rank-k}
\left\{ \begin{aligned}
X_i(t) &= X_i(0) +\int_0^t b_i(s)ds  + \int_0^t \sigma_i(s)dW_i(t) -F\Bigl( \sum_{l=1}^{k} v_{li} \int_0^t g(s) d\mathcal{L}^n_l(s) \Bigr) \\
\mathcal{L}^n_l(t) &= \frac{1}{n} \sum_{j=1}^{n} u_{jl} \mathbf{1}_{t\geq\tau_j}, \;\;   \tau_j=\inf \{ t\geq 0 :  X_j(t)\leq0  \}, \quad  l=1,\ldots,k.
\end{aligned} \right.
\end{equation}
Here, and it what follows, we assume that the Brownian motions $W^i$ are correlated through a single common Brownian motion. That is, for each $i=1,\ldots,n$, we have $W_i(t) =\rho B_0(t)+\sqrt{1-\rho^2}B_i(t)$ for a family of independent Brownian motions $B_0,\ldots,B_n$. In terms of the coefficients in \eqref{general_system_rank-k} we impose the following structural conditions which are motivated by the desire to include the original system \eqref{dist_to_default_system} and keep the analysis as simple as possible.
\begin{assumption}\label{MV_assump}
	$F$ is Lipschitz continuous and increasing with $F(0)=0$, while $g$ is continuous, non-negative, and decreasing. Furthermore, the asymmetry of the drifts and the volatilities is of the form $b_i(s)=b_{\mathbf{u}_i,\mathbf{v}_i}(s)$ and $\sigma_i(s)=\sigma_{\mathbf{u}_i,\mathbf{v}_i}(s)$. Finally, $ b_{u,v}$ and $ \sigma_{u,v}$ are deterministic functions of time, and we ask that $|\rho|<1$ and $\epsilon\leq \sigma_{u,v}\leq \epsilon^{-1}$, for a uniform constant $\epsilon>0$, as well as $\rho\sigma_{u,v} \in\mathcal{C}^{\kappa}(\mathbb{R})$, for some $\kappa>1/2$. 
\end{assumption}

Recall that the pair of $k$-dimensional vectors $\mathbf{u}_i$ and $\mathbf{v}_i$ from \eqref{eq:vectors_u_v} characterize bank $i$ in relation to interbank contagion. Together with the (random) initial conditions $X_i(0)$, for $i=1,\ldots,n$, this describes the asymmetry in the interbank market. In order to obtain something meaningful as the number of banks goes to infinity, we need to impose some structure through the convergence of their joint empirical measures defined by
\begin{equation}\label{empirical_measures_parameters}
\varpi^n:= \frac{1}{n} \sum_{i=1}^n \delta_{\textbf{u}_i}\otimes \delta_{\textbf{v}_i} \otimes \delta_{X_i(0)}, \qquad \text{for} \quad n\geq1.
\end{equation}
\begin{assumption}\label{assump:paramters_emperical}
	First of all, we assume $|\mathbf{u}_i|+|\mathbf{v}_i|\leq C$, for some $C>0$, uniformly in $i=1,\ldots,n$ and $n\geq 1$. Secondly, we ask that each $(\mathbf{u}_i,\mathbf{v}_i,X_i(0))$ is independent of the driving Brownian motions. Thirdly, we ask that $\varpi^n$ converges weakly to a probability measure $\varpi\in \mathcal{P}(\mathbb{R}^k\times \mathbb{R}^k \times \mathbb{R}_+)$, which we write as a joint law $\varpi=\mathrm{Law}(\mathbf{u},\mathbf{v},X(0))$ with $d\varpi(u,v,x)=d\nu_0(x|u,v)d\hat{\varpi}(u,v)$, where $\hat{\varpi}=\mathrm{Law}(\mathbf{u},\mathbf{v})$ and $\nu_0(\cdot|u,v)$ is the regular conditional law of $X(0)$ given $(\mathbf{u},\mathbf{v})=(u,v)$.
	Finally, letting $S(\mathbf{u}):=\mathrm{supp}(\mathrm{Law}(\mathbf{u}))$ and likewise for $\mathbf{v}$, we assume $S(\mathbf{u})$ and $S(\mathbf{v})$ are compact, and that for any $v\in S(\mathbf{v})$ we have $\sum_{l=1}^k u_lv_l \geq 0$ for all $u\in S(\mathbf{u})$, as in the finite system of size $n$ where $\sum_{l=1}^k{u_lv_l}=n\lambda_{ij}\geq0$ for every $u=\mathbf{u}_i$ and $v=\mathbf{v}_j$.
\end{assumption}

As we already pointed out in Section \ref{subsubsec:dynamics_finite_sys}, the particle system \eqref{general_system_rank-k} needs to be amended with a condition for how to resolve defaults (that is consistent with the equations and corresponds to the greatest clearing capital solution in Lemma \ref{lemma:exist}). This is achieved by insisting on the \emph{cascade condition} \eqref{eq:finite_PJC_fragile} which is the subject of the next subsection (Section \ref{Sec:determine_cascades}).

\begin{proposition}[Well-posedness of the particle system]\label{prop:particle_sys_well-posed}
	Let Assumption \ref{MV_assump} be in place. Equipped with the cascade condition \eqref{eq:finite_PJC_fragile}, as defined in Section \ref{Sec:determine_cascades} below, the system (\ref{general_system_rank-k}) has a unique strong c\`adl\`ag solution.   
\end{proposition}
\begin{proof}Up until the first default time, the system trivially has a unique strong solution that is continuous in time. Since we insist on the cascade condition \eqref{eq:finite_PJC_fragile}, the number of defaulting banks at the first default time is uniquely specified, and this then uniquely determines how to restart the system. Defining the solution recursively, for each of at most $n$ stopping times,  we obtain a unique strong solution with c\`adl\`ag  paths.
\end{proof}

\subsection{Characterizing the size of default cascades}\label{Sec:determine_cascades}
In this section we make precise when the loss processes $\mathcal{L}^n_l(t)$ should jump and what the size of each jump should be given the possibility of a default cascade---that is, when the default of one bank immediately forces more banks into default at the same instance of time. This takes some care in order to ensure the consistency with the system \eqref{general_system_rank-k}, but ultimately the situation is resolved by identifying the correct fixed point of an iterated mapping (as presented in Sections \ref{subsect:rule_cascade} and \ref{subsect:cascade_cond}). At first sight, the notation we introduce may appear a little abstract, but it leads to the convenient formulation \eqref{eq:finite_PJC_fragile} of what we call the \emph{cascade condition}, which characterizes the jump sizes of the particle system in an intrinsic way, and which guides the identification of the analogous condition for the mean field limit.  As discussed in Section~\ref{sec:model}, this cascade condition is conceptually related to the fictitious default algorithm of~\cite{EN01}.

Fix $n\geq1$ and consider the mapping from the type vector $\mathbf{v}_i$ to the total losses felt by bank $i$ given by
\[
\mathbf{v}_i \mapsto\mathbf{L}^{\!n}_{\textbf{v}_i}(t) := \sum_{l=1}^k v_{li} \mathcal{L}^n_l(t), \qquad \text{for} \quad i=1,\ldots,n.
\]
Recalling the decomposition \eqref{eq:k_loss_processes}, we simply have $\mathbf{L}^{\!n}_{\textbf{v}_i}(t)=\loss^{\!n}_i(t)$, but the point is to isolate how the asymmetric $i$-dependence arises strictly as a function of the vector $\mathbf{v}_i$ (whose $k$ components capture how significant banks of `type' $l=1,\ldots,k$ are to bank $i$). Notice that, while each $t\mapsto\mathcal{L}_l^n(t)$ in principle need not be increasing (depending on the rank factorization), the full process $t\mapsto \mathbf{L}^{\!n}_{\textbf{v}_i}(t) = \loss^{\!n}_i(t) $ is by definition increasing.

Given a c\`adl\`ag path $t\mapsto\eta(t)$, we write $\Delta\eta(t):=\eta(t)-\eta(t\shortminus)$. Then we can observe that, at any time $t\geq0$, the jump sizes of the loss processes must satisfy
\begin{align}\label{eq:jump_condition_step1}
\Delta \mathcal{L}^{n}_l(t) &=  \frac{1}{n} \sum_{i=1}^{n} u_{il} \mathbf{1}_{t = \tau_i} = \frac{1}{n} \sum_{i=1}^{n} u_{il} \mathbf{1}_{ \{ X_i(t\shortminus ) \in [0,  \Delta F_i^n (t)  ]  \} } \mathbf{1}_{t\leq \tau_i},
\end{align}
for each $l=1\ldots,k$, where $\Delta F_i^n(t)$ is the amount by which the $i$'th distance-to-default (or particle) is shifted down at time $t$ (due to losses from defaults at time $t$), namely
\[
\Delta F_i^n(t):= F\biggl( \sum_{j=1}^k v_{ji} \!\int_0^{t\shortminus}\!\!g(s)d\mathcal{L}^{n}_j(s) + g(t) \Delta \mathbf{L}^{\!n}_{\textbf{v}_i}(t) \biggr) - F\biggl(\sum_{j=1}^k v_{ji} \!\int_0^{t\shortminus}\!\!g(s)d\mathcal{L}^{n}_j(s) \biggr).
\]
Consequently, once we have identified the correct sizes of the jumps $\Delta \mathbf{L}^{\!n}_{\textbf{v}_i}(t)$, for $i=1,\ldots,n$, all the jump sizes $\Delta \mathcal{L}^{n}_l(t)$, for $l=1,\ldots,k$, are automatically uniquely specified by (\ref{eq:jump_condition_step1}). Indeed, the events $\{ t\leq \tau_i \}$ and the values $\int_0^{t\shortminus}\!\!g(s)d\mathcal{L}^{n}_i(s)$, for $i=1\ldots,n$, are all fixed at time $t$, since they are given as left-limits of the evolution of the system strictly before time $t$. On the other hand, the values $\Delta\textbf{L}^n_{\textbf{v}_i}(t)$ are to be determined at time $t$, and they will involve a choice, amounting to how we choose to resolve default cascades.

\subsubsection{Fixed point constraints and the cascade condition}\label{subsect:cascade_cond}

As for $\textbf{L}^n_{\textbf{v}_i}(t)$ above, it is important to realise that the $i$-dependence of $\Delta F_i^n (t)$ is again a function of $\textbf{v}_i$ alone. In order to make this clear (and to streamline the mathematical presentation), we introduce the random map $\Theta:\mathbb{R}_+ \times   (\mathbb{R}_+)^{\mathbb{R}^k} \times \mathbb{R}^k \rightarrow \mathbb{R}_+$ given by
\begin{equation}\label{eq:shift_down_map}
\Theta(t;f,y) :=  F\biggl( \sum_{j=1}^k y_j \!\int_0^{t\shortminus}\!\!g(s)d\mathcal{L}^{n}_j(s) + g(t) f(y) \biggr) - F\biggl(\sum_{j=1}^k y_j \!\int_0^{t\shortminus}\!\!g(s)d\mathcal{L}^{n}_j(s) \biggr).
\end{equation}
Clearly, we then have $\Delta F_i^n(t)= \Theta(t;\Delta \textbf{L}^n, \textbf{v}_i )$, so we can rewrite \eqref{eq:jump_condition_step1} as
\begin{align}\label{eq:jump_constraint1}
\Delta \mathcal{L}^{n}_l(t) &= \frac{1}{n} \sum_{i=1}^{n} u_{il} \mathbf{1}_{ \{ X_i(t\shortminus ) \in [0, \Theta^n(t;\Delta \textbf{L}^n, \hat{\textbf{v}}_i ) ]  , \, t\leq \tau_i \} }, \qquad l=1\ldots,k.
\end{align}
As already alluded to above, this shows that: once we pin down the mapping ${\textbf{v}} \mapsto \Delta \textbf{L}^n_{\textbf{v}}$, then the correct jumps of  $\mathcal{L}^n_1,\ldots,\mathcal{L}^n_k $ are automatically specified by the constraint \eqref{eq:jump_constraint1}.

Looking at \eqref{eq:jump_constraint1}, we immediately obtain a constraint for $\textbf{v} \mapsto \Delta \textbf{L}^n_{\textbf{v}}$ by simply summing over $l=1,\ldots,k$ weighted by the $v_{il}$'s, which yields the identity

\begin{equation}\label{eq:jump_constraint2}
\Delta \textbf{L}^n_{\textbf{v}_i}(t) = \sum_{l=1}^k v_{li} \biggl( \frac{1}{n} \sum_{j=1}^{n} u_{jl}  \mathbf{1}_{ \{ X_j(t\shortminus ) \in [0, \Theta^n(t,\Delta \textbf{L}^n, \textbf{v}_j )  ]  , \, t\leq \tau_j \} } \biggr), \qquad i=1,\ldots,n.
\end{equation}
Note that this is precisely saying that $\textbf{v} \mapsto \Delta \textbf{L}^n_{\textbf{v}}$ arises as a fixed point $\Xi(t;\Delta\textbf{L}^n,\cdot)=\Delta\textbf{L}^n_{(\cdot)}(t)$, where the random map $\Xi:\mathbb{R}_+ \times   (\mathbb{R}_+)^{\mathbb{R}^k} \times \mathbb{R}^k \rightarrow \mathbb{R}_+$ is defined by

\begin{equation}\label{eq:jump_map}
\Xi(t;f,x):= \sum_{l=1}^k x_l \biggl( \frac{1}{n} \sum_{j=1}^{n} u_{jl}  \mathbf{1}_{ \{ X_j(t\shortminus ) \in [0, \Theta(t; f, x) ]  , \, t\leq \tau_j \} } \biggr).
\end{equation}
However, the mapping $f\mapsto \Xi(t;f,\cdot)$ can have multiple fixed points, so the above fixed point constraint alone is not enough to determine the jump sizes $\Delta\textbf{L}^n$. That is, the system (\ref{general_system_rank-k}) is a priori ill-posed without a selection rule.

Based on the natural step-by-step resolution of default cascades (explained in detail in Section \ref{subsect:rule_cascade}  below), the correct selection rule simply amounts to a (suitably initialized) iterative application of the mapping $\Xi$. This can be formulated succinctly as
\begin{equation}\label{eq:finite_PJC_fragile}
\begin{cases}
\Delta \textbf{L}^n_{\textbf{v}}(t) := {\displaystyle \lim_{m\rightarrow n} \Delta^{m}_{t,\textbf{v}}} \quad \text{for} \quad \textbf{v}=\textbf{v}_1\ldots,\textbf{v}_n,    \\[7pt]
\Delta^{m}_{t,\textbf{v}}:= \Xi(t, \Delta^{m-1}_{t,\textbf{v}}, \textbf{v} ) \quad \text{and} \quad  \Delta^{0}_{t,\textbf{v}}:= \Xi(t;0,\textbf{v}),
\end{cases}
\end{equation}
which we will refer to as the \emph{cascade condition} for the jump sizes. As concerns the notion of a step-by-step resolution, the number of `steps' or `rounds' in the cascade is given by the smallest $\bar{m}$ such that $\lim_{m\rightarrow n} \Delta^{m}_{t,\cdot}= \Delta^{\bar{m}}_{t,\cdot}$. For clarity, further details on this are presented in the separate Section \ref{subsect:rule_cascade} below, where we give a precise mathematical definition of instantaneous default cascades leading to this condition.


\subsubsection{Detailed description of the resolution of cascades} \label{subsect:rule_cascade}

Let $\mathcal{A}_{t\shortminus} $ denote the (random) set of indices $i\leq n$ such that $t\leq\tau_i$ for each $i\in \mathcal{A}_{t\shortminus}$, meaning that bank $i$ was (and may still be) solvent---or more colloquially, `alive'---strictly before time $t$. Note that there is at least one default at time $t$ precisely when $X_{i_0}(t\shortminus)=0$ for some bank indexed by $i_0\in \mathcal{A}_{t\shortminus}$. Therefore, we define the (random) set of indices $\mathcal{D}^{0}_t$ as precisely those $i\in \mathcal{A}_{t\shortminus}$ for which $X_{i}(t\shortminus)=0$, corresponding to the initial set of defaults at time $t$ (which came about without any role played by contagion).

Supposing that $\mathcal{D}^{0}_t \neq \emptyset$, we now need to make it mathematically precise how to decide if a \emph{cascade} is initiated by the contagious effects from these initial defaults---and then we need to make precise how to resolve the total size of the cascade if it occurs.

Recalling the definition of $\Xi$ in \eqref{eq:jump_map}, the isolated effect of the initial defaults (which we recall are indexed by $\mathcal{D}^{0}_t$) is to increase each $\textbf{L}^n_{\textbf{v}_i}(t\shortminus)$ by an initial jump of size
\[
\Delta^{0}_{t,\textbf{v}_i }  :=   \Xi (t; 0, \textbf{v}_i )
= \frac{1}{n} \,\sum_{l=1}^k v_{li} \sum_{j\in\mathcal{D}_t^{(0)}}u_{j l }, \qquad \text{for} \quad i=1,\ldots,n.
\]
Next, we define the (random) set of indices $\mathcal{D}^1_t\subset {\mathcal{A}_{t\shortminus}}/ \mathcal{D}^{0}_t$ as precisely those $i\in {\mathcal{A}_{t\shortminus}}/ \mathcal{D}^{0}_t$ for which
\[
X_{i}(t\shortminus)-\Theta(t; \Delta^{0}_{t,(\cdot)}, \textbf{v}_{i} ) \leq 0.
\]
In other words, the members of $\mathcal{D}^{1}_t$ are precisely those  banks that enter into default at time $t$ on account of the contagion from the initiating set of defaults $\mathcal{D}^{0}_t$. Supposing that $\mathcal{D}_t^{1}\neq \emptyset$, we must update the losses to include this first round of contagion, and check if this induces another round of contagion. This amounts to considering a jump in $ \textbf{L}^n_{\textbf{v}_i}(t\shortminus)$ of size
\[
\Delta^{1}_{t,\textbf{v}_i } := \Xi (t;\Delta^{0}_{t,(\cdot)}, \textbf{v}_i  ) = \frac{1}{n}  \, \sum_{l=1}^k v_{li} \sum_{j\in \mathcal{D}_t^{0}\cup\mathcal{D}_t^{1}} u_{j l }, \qquad \text{for} \quad {i=1,\ldots,n},
\]
and thus defining $\mathcal{D}_t^{2}\subset \mathcal{A}_{t\shortminus}/\mathcal{D}_t^{1}$ as precisely those $i\in \mathcal{A}_{t\shortminus}/\mathcal{D}_t^{1}$ for which
\[
X_{i}(t\shortminus)-\Theta(t; \Delta^{1}_{t,(\cdot)}, \textbf{v}_{i} ) \leq 0.
\]
For any $m\leq n$, $\Delta^{m}_{t,(\cdot)} $ and $\mathcal{D}_t^{m}$ are defined analogously. Recalling that $n^{-1}\sum_{l=1}^k v_{li} u_{jl} = \lambda_{ji} \geq 0$, it is always the case that
\[
\Delta^{m}_{t,(\cdot)}  =  \Xi (t;\Delta^{m-1}_{t,(\cdot)}, \cdot ) \geq \Delta^{m-1}_{t,(\cdot)}
\]
in the pointwise sense. In particular, each $(\Delta^{m}_{t,\mathbf{v}})$ is an increasing sequence in $m=0,1,\ldots,n$, and it is immediate that $\mathcal{D}^{\bar{m}}_t = \emptyset$ implies 
\begin{equation*}
\Xi (t;\Delta^{\bar{m}}_{t,(\cdot)}, \textbf{v} ) = \Delta^{\bar{m}}_{t,\textbf{v}}, \qquad \text{for all} \quad \textbf{v}=\textbf{v}_1\ldots, \textbf{v}_n,
\end{equation*}
so the sequence eventually reaches a fixed point after the $\bar{m}$'th round of contagion-induced defaults (since there are only $n$ banks in total, note that we must have $\mathcal{D}_t^n=\emptyset$). Once a fixed point is reached at the $\bar{m}$'th iteration, there are no further defaults, and hence the default cascade is fully resolved with
\[
\Delta \textbf{L}^n_{\textbf{v}}(t):= \Delta^{\bar{m}}_{t,\textbf{v}} = \sum_{l=1}^k v_{li} \sum_{ j\in \mathcal{D}_t^{0}\cup \cdots \cup \mathcal{D}_t^{\bar{m}} } u_{j l }, \qquad \text{for} \quad \textbf{v}=\textbf{v}_1,\ldots, \textbf{v}_n.
\]
By construction, each sequence  $(\Delta^{m}_{t, \mathbf{v}}) $ stays fixed after this $\bar{m}$'th iteration, so the value $\Delta \textbf{L}^n_{\textbf{v}}(t)$ is indeed the limit of $\Delta^{m}_{t, \mathbf{v}}$ as $m\rightarrow n$, in agreement with the \emph{cascade condition} defined in \eqref{eq:finite_PJC_fragile}. This formulation of the jump size as a limit of iterated steps in the default cascade is instructive for the formulation of the mean field analogue, which follows in the next subsection.

\begin{remark} If all banks have interbank liabilities, then, for any $j$, there is an $i$ such that $\lambda_{ji}>0$. Hence, we get $\mathcal{D}^{m}_t \neq \emptyset$ if and only if $ \Xi (t;\Delta^{m}_{t,\cdot}, \textbf{v})\geq \Delta^{m}_{t,\textbf{v}} $ for all $\textbf{v} \in \{ \textbf{v}_1\ldots, \textbf{v}_n \}$ and $\Xi (t;\Delta^{m}_{t, \cdot} , \textbf{v}) > \Delta^{m}_{t,\textbf{v}}$ for at least one $\textbf{v} \in \{\textbf{v}_1\ldots, \textbf{v}_n \}$.
\end{remark}

\subsection{The mean field limit}\label{the_mean_field_limit}

Provided there is a suitable averaging effect, we can send $n\rightarrow\infty$ in the particle system (\ref{general_system_rank-k}) and thereby capture the `systemic' or macro-level properties of the finite system through its mean field limit. As we show in Appendix \ref{convergence_appendix}, the insistence on Assumptions \ref{MV_assump} and \ref{assump:paramters_emperical} is sufficient to ensure such a law of large numbers, and it then follows that the resulting mean field limit is given by a McKean--Vlasov problem of the form
\begin{equation}
\begin{cases}\label{CMV}
\mathcal{L}_l(t)  = \int_{\mathbb{R}_+\times \mathbb{R}^k \times \mathbb{R}^k}  u_l \mathbb{P}( t\geq \tau^x_{u,v} \mid B_0) d\varpi(x,u,v), \qquad l=1,\ldots,k, \\[6pt]
\tau^x_{u,v} =\inf \{ t\geq 0 : X^x_{u,v}(t) \leq 0 \},\\[6pt]
X^{x}_{u,v}(t) = x + \int_0^t \!b_{u,v}(s)ds + \int_0^t\!\sigma_{u,v}(s) dW_s - F\Bigl( \sum_{l=1}^k v_l \int_0^t \!g(s)d\mathcal{L}_l(s)  \Bigr),
\end{cases}
\end{equation}
for $t\in[0,T]$, where $W=\rho B + \sqrt{1-\rho^2}B_0$, for two independent Brownian motions $B \perp B_0$.

In general, one would expect the above problem to be continuous in time, due to the `averaging' effect of passing to the mean field limit, and, as long as this is the case, the system is fully specified by \eqref{CMV}. However, as we saw already in Section 4, the loss processes $t\mapsto \mathcal{L}_l(t)$ may in fact undergo jump discontinuities (in particular, \cite[Thm 1.1]{HLS18} can be adapted to show that such jumps must occur for some parameters). When accounting for this, one needs to be careful that the jump sizes are \emph{not} pinned down uniquely by the formulation \eqref{CMV}. Similarly to the cascade condition for the finite system, a concrete choice must be made that allows the system to be c\`adl\`ag and uniquely determines the jumps. This is the topic of the next subsection.

\subsubsection{Determining the jump sizes}\label{sec:jumps_loss}

Our first task is to show that the jump sizes must obey certain fixed point constraints. By analogy with the analysis of cascades in the finite system, we therefore define the mapping $\mathbf{L}:S(\mathbf{v})\rightarrow L^\infty(0,T)$ by
\[
\mathbf{L}_v(t):= \sum_{l=1}^k v_l\mathcal{L}_l(t),
\]
where $L^\infty(0,T)=L^\infty([0,T],\mathbb{R})$ is the space of bounded real-valued functions on $[0,T]$ under the equivalence relation of being equal almost everywhere. By Assumption \ref{assump:paramters_emperical} and the definition of each $\mathcal{L}_l$, the process $t\mapsto \mathbf{L}_v(t)$ is indeed bounded, for each $v\in S(\mathbf{v})$, and crucially the assumptions also imply that $t\mapsto \mathbf{L}_v(t)$ is increasing. Similarly to the constraint (\ref{eq:jump_constraint1}) for the finite system, we can infer directly from (\ref{CMV}) that the jumps of each $\mathcal{L}_l$ must satisfy
\begin{equation}\label{1st_fp_constraint}
\Delta \mathcal{L}_l(t) = \int_{\mathbb{R}_+\times \mathbb{R}^k \times \mathbb{R}^k}  u_l \mathbb{P}\bigl( X^x_{u,v}(t\shortminus ) \in [0, \Theta(t; \Delta \mathbf{L} , v)  ]  , \, t\leq \tau \mid B_0 \bigl) d\varpi(x,u,v),
\end{equation}
where
\[
\Theta(t;f,v):= F\biggl( \int_0^{t\shortminus}\!\!g(s)d\mathbf{L}_v(s) + g(t) f(v) \biggr) - F\biggl( \!\int_0^{t\shortminus}\!\!g(s)d\mathbf{L}_v(s) \biggr).
\]
Once the jumps of $\mathbf{L}$ are pinned down, the constraint \eqref{1st_fp_constraint} uniquely determines the jumps of each loss process $\mathcal{L}_l$.
Furthermore, from \eqref{1st_fp_constraint} and the definition of $\mathbf{L}$, it follows that the (random) value of $\Delta \mathbf{L}$ must be a fixed point of the (random) mapping $f\mapsto \Xi(t;f,\cdot)$, where
\begin{equation}\label{the_map_Xi}
\Xi(t;f,v):= \sum_{l=1}^k v_l \int_{\mathbb{R}_+\times\mathbb{R}^k\times\mathbb{R}^k} \hat{u}_{l}  \mathbb{P} \bigl( X^x_{\hat{u},\hat{v}}(t\shortminus ) \in [0, \Theta(t; f, \hat{v} ) ]  , \, t\leq \tau  \mid B_0\bigr) d\varpi(x,\hat{u},\hat{v}).
\end{equation}

In general, the map \eqref{the_map_Xi} can have  several, even infinitely many, fixed points (for example, $f\equiv0$ is always a fixed point, but this is not compatible with jumps), so $\Delta \mathbf{L}$ is not uniquely specified a priori. This situation is resolved by the selection rule \eqref{eq:limit_PJC} introduced below, mimicking the cascade condition \eqref{eq:finite_PJC_fragile} for the jumps in the finite system. However, unlike the finite system, the mean field limit always satisfies $\Xi(t;0,\cdot)=0$, so the occurrence and size of a potential default cascade must be identified by artificially exposing the system to an arbitrarily small shock. Specifically, we shift the system down by a small amount $\Theta(t;\epsilon,v)$ and then keep track of how the resulting losses propagate as $\varepsilon \downarrow 0$, meaning that the size of the initial shift $\Theta(t;\epsilon,v)$ vanishes. Mathematically, this means that, for $v\in S(\textbf{v})$, the jump size $\Delta \mathbf{L}_v(t)$ is given by the \emph{mean field cascade condition}
\begin{equation}\label{eq:limit_PJC}
\begin{cases}
\Delta \mathbf{L}_v(t) = \lim_{\varepsilon\downarrow 0} \lim_{m \uparrow \infty} \Delta^{\!(m,\varepsilon)}_{t,v}, \\[5pt] \Delta^{\!(m,\varepsilon)}_{t,v}= \Xi(t, \varepsilon +  \Delta^{\!(m-1,\varepsilon)}_{t,\hspace{0.5pt}\cdot}, v), \quad m\geq1, \\[5pt]
\Delta^{\!(0,\varepsilon)}_{t,v}= \Xi(t,\varepsilon , v),
\end{cases}
\end{equation}
where the equalities hold almost surely (recall that $\Xi$ is conditional on the common noise $B_0$).
We stress that the limit is well-defined, since $(\Delta^{(m,\varepsilon)}_{t,v})$ forms a bounded sequence that increases as $m\uparrow\infty$ and decreases as $\varepsilon\downarrow0$. Moreover, dominated convergence shows that the (random) map $v\mapsto \Delta \mathbf{L}_v(t)$ given by \eqref{eq:limit_PJC} is a fixed point of  $f\mapsto \Xi(t,f,\cdot)$, so this choice for the jump sizes is indeed consistent with the McKean--Vlasov problem \eqref{CMV}.
\begin{remark}
	We emphasise that the iterative structure of \eqref{eq:limit_PJC} lends itself easily to numerical implementation, and this is indeed the starting point for the algorithm behind the simulations in Figure \ref{fig:heat_plots}. Moreover, we note that, in the case of a symmetric network of obligations, \cite[Proposition 2.4]{HLS18} gives that the mean field cascade condition \eqref{eq:limit_PJC}  agrees with the corresponding notion of a `physical' jump condition considered in \cite{HLS18}.
\end{remark}

\subsubsection{A simple criterion for ruling out jumps}\label{criterions_jump_or_not}

We now present a simple criterion, namely \eqref{no_jump_condition1}, that rules out a jump discontinuity at a given time $t>0$ and guarantees the solution stays continuous in some small amount of time thereafter. Of course, the mean field cascade condition \eqref{eq:limit_PJC} already gives the precise criterion for whether or not the system undergoes a jump at time $t$, but  our aim here is to provide some intuition for the workings of this condition.

At any given time $t>0$, and for any pair $(u,v)$, we let $V_t(\cdot|u,v)$ denote the random density of the random sub-probability measure
\begin{equation}\label{eq:the_nu_t_measure}
A \mapsto \nu_t(A|u,v) := \int_{\mathbb{R}_+} \mathbb{P} \bigl( X^x_{u,v}(t) \in A , \, t < \tau^x_{u,v}  \,|\, B_0\bigr) V_0(x|u,v)dx.
\end{equation}
For the purposes of this subsection, we think of having fixed a realisation of $B_0$, so the below criteria \eqref{no_jump_condition1} should be understood as holding for this particular realisation: thus, the conclusion is that there is no jump for this particular realisation of $B_0$. Of course, if the criteria holds for all realisations of $B^0$, then it is an almost sure conclusion.

Recalling the definition of $\Theta$, we have $\Theta(t;f,v) \leq \Vert  F \Vert_{\text{Lip}} g(t) f(v)$,
so
\begin{equation*}
\Xi(t; f , v) \leq \sum_{l=1}^k v_l  \int_{\mathbb{R}^k\times\mathbb{R}^k}    \int_0^{\Vert  F \Vert_{\text{Lip}} g(t) f(\hat{v})} \hat{u}_{l}  V_{t\shortminus }(y|\hat{u},\hat{v}) dy d\hat{\varpi}(\hat{u},\hat{v}),
\end{equation*}
where $V_{t\shortminus}$ is the pointwise left limit of $V_{s}$ as $s\uparrow t$.
Now fix a time $t>0$, and suppose $V_{t\shortminus}$ and the joint distribution $\hat{\varpi}$ of $(\mathbf{u},\mathbf{v})$ satisfy the following criterion:~there is a small $\delta>0$ such that, for each $v\in S(\mathbf{v})$,
\begin{equation}\label{no_jump_condition1}
\Vert F \Vert_{\text{Lip}} g(t) \sum_{l=1}^k v_l  \int_{\mathbb{R}^k\times\mathbb{R}^k}    \hat{u}_{l}  V_{t\shortminus }(y|\hat{u},\hat{v}) d\hat{\varpi}(\hat{u},\hat{v}) \leq 1-\delta  \qquad \textrm{for all }  y\in(0,\epsilon_v),
\end{equation}
for some small enough $\epsilon_v>0$. Then we get
\begin{equation*}
\Delta^{\!(0,\varepsilon)}_{t,v}= \Xi(t,\varepsilon , v) \leq  (1-\delta) \varepsilon,
\end{equation*}
for all $\varepsilon>0$ sufficiently small such that $\Vert F \Vert_{\text{Lip}} g(t)\varepsilon < \epsilon_v$. In turn, for all $\varepsilon>0$ such that $\Vert F \Vert_{\text{Lip}} g(t)(1+(1-\delta))\varepsilon<\epsilon_v$, we have
\[
\Delta^{\!(1,\varepsilon)}_{t,v}= \Xi(t, \varepsilon +  \Delta^{\!(0,\varepsilon)}_{t,\hspace{0.5pt}\cdot}, v) \leq (1-\delta)\varepsilon + (1-\delta)^2\varepsilon,
\]
and, by recursion, for any given $m\geq1$ we thus have
\[
\Delta^{\!(m,\varepsilon)}_{t,v}= \Xi(t, \varepsilon +  \Delta^{\!(m-1,\varepsilon)}_{t,\hspace{0.5pt}\cdot}, v) \leq \varepsilon \sum_{l=0}^{m} (1-\delta)^{l+1},
\]
for all $\varepsilon>0$ such that $\Vert F \Vert_{\text{Lip}} g(t)(\sum_{l=0}^{m} (1-\delta)^{l})\varepsilon<\epsilon_v$.
Since $1-\delta<1$, the sum forms a geometric series that converges as $m\rightarrow\infty$,  so provided \eqref{no_jump_condition1} is satisfied we conclude that $\Delta \mathbf{L}_v(t)=0$ for all $v\in S(\mathbf{v})$, since
\[
\Delta \mathbf{L}_v(t) \leq  \lim_{\varepsilon\downarrow0}\lim_{m\uparrow\infty}\varepsilon \sum_{l=0}^{m} (1-\delta)^{l+1} = \lim_{\varepsilon\downarrow 0} \varepsilon \frac{1-\delta}{\delta} = 0.
\]
In particular, each $s\mapsto \mathcal{L}_l(s)$ must indeed continuous at time $t$. In other words, after the system takes an artificial hit of order $\varepsilon$, the induced rounds of contagion quickly become negligible with the
total effect being at most of order $\varepsilon$ (i.e., of the same order as the initial `artificial' shock), and so they disappear as the size of the initial shock is sent to zero. Furthermore, now that we know there is not a jump, a straightforward adaptation of \cite[Prop.~6.4.3]{sojmark_2019} shows that a bound of the form \eqref{no_jump_condition1} holds for some small amount time, so the solution remains continuous on this nonzero time interval.

On the other hand, Remark \ref{eq:dirichlet_vs_jump} from Section 4 provides a simple example where the condition \eqref{no_jump_condition1} is violated, for some $v\in S(\mathbf{v})$, and where it is proved that the cascade condition must therefore result in a jump discontinuity. In addition to the argument provided there, we remind the reader of Figure \ref{fig:heat_plots}, which illustrates the occurrence of the jump. In this respect, let us also stress that there can of course be cases where $t\mapsto \mathbf{L}_v(t)$ only jumps for some $v\in S(\mathbf{v})$ and not for others, provided there are certain types which are not exposed to the types experiencing default cascades. As a particular example of this, we could amend the example behind Figure \ref{fig:heat_plots} by imposing that `Core 2' is not exposed to losses in `Core 1' and `Periphery 2': then we obtain an example where $\mathbf{L}_l$ jumps for $l=1,4$ (`Core 1' and `Periphery 2') while it does not jump for $l=2,3$ (`Core 2' and `Periphery 1').
\begin{remark}
	Notice that if the Dirichlet boundary condition $V_{t-}(0|u,v)=0$ is satisfied (meaning that $\lim_{x\downarrow0}V_{t-}(x|u,v)=0$), then \eqref{no_jump_condition1} is definitely true. More generally, the criterion amounts to $y\mapsto V_{t\shortminus}(y|u,v)$ being sufficiently small relative to $\Vert F \Vert_{\textrm{Lip}}^{-1}g(t)^{-1}$ near $y=0$, depending on the joint distribution $\hat{\varpi}$. Starting from a nice initial condition, we will have $V_{t}(0|u,v)=0$ for some amount of time, but if the contagious feedback becomes too strong it may transport the density so fast towards the origin that there is a blow-up time $t_\star$: that is,  the derivative of $t\mapsto \mathbf{L}(t)$ diverges as $t\uparrow t_\star$, and the left limit density $V_{t_{\star}\shortminus}$ fails to vanish at zero, in a way such that the cascade condition enforces a jump discontinuity $\Delta \mathbf{L}(t_\star)>0$.
\end{remark}

\subsubsection{Idiosyncratic noise: well-posedness and regularity of the loss}

In this subsection we consider the McKean--Vlasov problem (\ref{CMV}) when $\rho=0$, meaning that there is no common noise. This makes it more tractable to get a handle on the regularity of the loss processes. Under a mild assumption on the initial profile of the system, we are able to generalise the arguments from \cite{HLS18} and thus show that the system is well-posed up until the $L^2$ norm of the gradient of the losses, namely  $(\partial_t\mathcal{L}_1,\ldots,\partial_t\mathcal{L}_k)$, explodes.

The assumption on the initial profile amounts to controlling the decay of the mass near the origin. Specifically, using the notation from Assumption \ref{assump:paramters_emperical}, we require that the initial condition $\nu_0$ satisfies $d\nu_0(x|u,v)=V_0(x|u,v)dx$ with 
\begin{equation}\label{eq:cond_V_0_decay}
V_0(x|u,v)\leq C_\star x^\beta \mathbf{1}_{x<x_\star} + D_\star \mathbf{1}_{x\geq x_\star}, \qquad \text{for all} \quad x>0,
\end{equation}
for constants $C_\star,D_\star,x_\star>0$ and a power $\beta\in(0,1]$, uniformly in $u\in S(\mathbf{u})$ and $v\in S(\mathbf{v})$.

\begin{theorem}[Well-posedness up to explosion]\label{MV_thm} 
	Suppose the initial condition $\nu_0$ satisfies \eqref{eq:cond_V_0_decay} and let Assumption \ref{MV_assump} be in place. Then there is a regular (i.e., differentiable) solution $(\mathcal{L}_1\ldots,\mathcal{L}_k)$ to the McKean--Vlasov problem \eqref{CMV} up to the explosion time
	\[
	t_{\star}:= \sup \Bigl\{ t>0 :  \sum_{l=1}^k \Vert \partial_t\mathcal{L}_l(\cdot) \Vert_{L^2(0,t)} < \infty   \Bigr\} \in (0,\infty],
	\]
	with the property that, for all $t<t_\star$, $\partial_s \mathcal{L}_l(s) \leq K s^{-(1-\beta)/2}$ on $[0,t]$ for some constant $K> 0$. Moreover, the solution is unique on $[0,t_\star]$ in the broadest possible sense: any generic c\`{a}dl\`{a}g solution to \eqref{CMV} satisfying the cascade condition \eqref{eq:limit_PJC} must agree with the regular solution on $[0,t_\star]$.
\end{theorem}
\begin{proof}
	See Section \ref{appendix_proof1} in the Appendix.
\end{proof}

We do not attempt to address general results on global uniqueness here, but it is natural to conjecture that there is indeed uniqueness under the cascade condition \eqref{eq:limit_PJC}. One would then expect to have a regularity result analogous to Theorem \ref{MV_thm}  holding on the intervals in between blow-ups. In the case of a symmetric network of liabilities and constant coefficients with $F(x)=x$ and $g(x)=1$, it follows from \cite[Prop.~2.4]{HLS18} that the McKean--Vlasov problem considered here (with only idiosyncratic noise) simplifies to that of \cite[Eqn.~(1.1)]{HLS18} with the physical jump condition \cite[Eqn.~(1.8)]{HLS18} in place of the cascade condition. For this problem, global uniqueness and regularity was recently established in the preprint \cite{DNS19}.

\subsubsection{Common noise: global well-posedness  with weak feedback}

For a given initial profile $\nu_0$ and feedback functions $F$ and $g$, we already argued above that solutions to the conditional McKean--Vlasov problem (\ref{CMV}) exist, by virtue of arising as limit points of the finite particle system \eqref{general_system_rank-k}, provided Assumptions \ref{assump:paramters_emperical} and \ref{MV_assump} are satisfied. 
For the details of this, we refer to Section \ref{convergence_appendix} in the appendix.

Existence aside, general results on global as well as local uniqueness remain a challenge in the presence of the common noise, even for simpler versions of the problem. Nevertheless, we have the following result when a `smallness condition', namely  (\ref{smallness_cond}), is imposed on the feedback functions $F$ and $g$ in relation to the initial profile of the system. This condition guarantees that the feedback from contagion is too weak to generate blow-ups in the mean field limit, independently of the different realisations of the common noise.
\begin{theorem}[Global well-posedness in the weak feedback regime]\label{thm_weak_feedback} Let Assumption \ref{MV_assump} be in place. If there is a $\delta>0$ such that
	\begin{equation}\label{smallness_cond}
	\Vert F \Vert_{\mathrm{Lip}} g(0) \sum_{l=1}^k v_l  \int_{\mathbb{R}^k\times\mathbb{R}^k}    \hat{u}_{l} \Vert V_{0}(\cdot|\hat{u},\hat{v}) \Vert_{\infty} d\hat{\varpi}(\hat{u},\hat{v}) \leq 1-\delta ,
	\end{equation}
	for all $v\in S(\mathbf{v})$, then there is a globally unique solution to (\ref{CMV}), and this solution is continuous in time.
\end{theorem}
\begin{proof}
	See Section \ref{Proof_thm_uniq_common} of the Appendix.
\end{proof}
The smallness condition \eqref{smallness_cond} should look familiar in light of Section \ref{criterions_jump_or_not}, and indeed the proof of continuity in time amounts to verifying that the smallness condition implies the bound \eqref{no_jump_condition1}. On the other hand, by adapting the arguments from \cite[Theorem 2.7]{LS18a}, one can show that: if the smallness condition \eqref{smallness_cond} does not hold, then there is a non-trivial probability of seeing jump discontinuities (for certain realisations of the common noise). We make no attempt at treating uniqueness in that regime here.

\subsubsection{The core-periphery model from Section 4}\label{subsec:example_sect4}

In this final section we show how the core-periphery mean field model \eqref{eq:concrete_mean_field_limit} from Section \ref{sec:mf} is a special case of the mean field limit \eqref{CMV}. Following the framework of Section \ref{subsec:example_sect4}, given $m_0=m_\mathrm{c}+m_\mathrm{p}$, we write $n=mm_0$ for arbitrary multiples $m\geq1$, and let the liabilities matrix $\lambda_{n\times n}$ be defined by \eqref{update_lambda_hat}--\eqref{eq:lambda_random} via the underlying matrix $\hat{\lambda}_{m_0\times m_0}$. Due to the special structure, we have the decomposition
\[
m_0 \lambda_{m_0\times m_0} = U_{m_0 \times k} V_{k \times m_0},
\]
with entries
\[
u_{ij} = (1+\epsilon_i)\hat{u}_{ij} \quad \text{and}\quad v_{ij} =(1+\delta_j)\hat{v}_{ij}
\]
given the entries $\hat{u}_{ij}$ and $\hat{v}_{ij}$ of the underlying rank decomposition
\[
m_0\hat{\lambda}_{m_0\times m_0} = \hat{U}_{m_0 \times k} \hat{V}_{k \times m_0}.
\]
Now consider the concrete choice \eqref{eq:lambda_hat_concrete} for $\hat{\lambda}_{m_0\times m_0}$, and notice that the four types (two core banks and two peripheral groups) means that its decomposition 
is fully described by just four row vectors of $\hat{U}_{m_0 \times k}$ and four column vectors of $\hat{V}_{k \times m_0}$: namely the two `core' pairs
\begin{equation}\label{eq:core_pairs}
(\tilde{\mathbf{u}}_1,\tilde{\mathbf{v}}_1):=(\hat{\mathbf{u}}_1,\hat{\mathbf{v}}_1) \quad \text{and} \quad (\tilde{\mathbf{u}}_2,\tilde{\mathbf{v}}_2):=(\hat{\mathbf{u}}_2,\hat{\mathbf{v}}_2) 
\end{equation}
as well as the two `peripheral' pairs
\begin{align}\label{eq:peripheral_pairs}
(\tilde{\mathbf{u}}_3,\tilde{\mathbf{v}}_3)&:=(\hat{\mathbf{u}}_3,\hat{\mathbf{v}}_3)=\cdots = (\hat{\mathbf{u}}_{2+m_{\mathrm{p},1}},\hat{\mathbf{v}}_{2+m_{\mathrm{p},1}}) \nonumber \\
(\tilde{\mathbf{u}}_4,\tilde{\mathbf{v}}_4)&:=(\hat{\mathbf{u}}_{3+m_{\mathrm{p},1}},\hat{\mathbf{v}}_{3+m_{\mathrm{p},1}})=\cdots = (\hat{\mathbf{u}}_{m_0},\hat{\mathbf{v}}_{m_0}).
\end{align}
As we grow the system to arbitrarily large sizes $n=mm_0$, according to \eqref{update_lambda_hat}, this underlying structure from the rank decomposition of $\hat{\lambda}$ is preserved. Thus, the empirical measures \eqref{assump:paramters_emperical} corresponding to the decompositions
\[
n \lambda_{n\times n} = mm_0 \lambda_{mm_0\times mm_0} = U_{mm_0 \times k} V_{k \times mm_0}
\]
take the form
\begin{align*}
&\varpi^{n} =  \frac{1}{m_0} \biggr( \frac{1}{m} \sum_{i=1}^{m}\delta_{(1+\epsilon^1_i)\hat{\textbf{u}}_1}\otimes \delta_{(1+\delta^1_i)\hat{\textbf{v}}_1} \otimes \delta_{X_i^1(0)}
+  \frac{1}{m} \sum_{i=1}^{m}\delta_{(1+\epsilon^2_i)\hat{\textbf{u}}_2}\otimes \delta_{(1+\delta^2_i)\hat{\textbf{v}}_2} \otimes \delta_{X_i^2(0)} \\
&+\sum_{j=1}^{m_{\mathrm{p},1}} \frac{1}{m} \sum_{i=1}^{m} \delta_{(1+\epsilon_{ij}^{3})\tilde{\textbf{u}}_3}\otimes \delta_{(1+\delta^{4}_{ij})\tilde{\textbf{v}}_3} \otimes \delta_{X_i^3(0)}
+ \sum_{j=1}^{m_{\mathrm{p},2}} \frac{1}{m} \sum_{i=1}^{m} \delta_{(1+\epsilon^{4}_{ij})\tilde{\textbf{u}}_4}\otimes \delta_{(1+\delta^{4}_{ij})\tilde{\textbf{v}}_4} \otimes \delta_{X_i^4(0)}
\biggr)
\end{align*}
for a general $n=mm_0$, for all $m\geq1$. Since all the $(\epsilon,\delta)$'s are drawn from $P\otimes P$ in an i.i.d.~way, it follows that $\varpi^n$ is weakly convergent with limiting law 
\begin{align}\label{eq:varpi_sec4}
\varpi = \frac{1}{m_0} & \mathrm{Law}(\mathbf{u}^1, \mathbf{v}^1, X^1(0)) + \ \frac{1}{m_0}\mathrm{Law}(\mathbf{u}^2, \mathbf{v}^2, X^2(0))\nonumber \\
&+ \frac{m_{\mathrm{p},1}}{m_0}\mathrm{Law}(\mathbf{u}^3, \mathbf{v}^3, X^3(0)) +\frac{m_{\mathrm{p},2}}{m_0}\mathrm{Law}(\mathbf{u}^4, \mathbf{v}^4, X^4(0))
\end{align}
where
\begin{equation}\label{eq:varpi_pushforward}
\mathrm{Law}(\mathbf{u}^i, \mathbf{v}^i) = (P \otimes P) \circ \phi^{-1}_i , \qquad \phi_i(\theta)=\bigl((1+\theta)\tilde{\mathbf{u}}_i, (1+\theta)\tilde{\mathbf{v}}_i\bigr),
\end{equation}
for $i=1,\ldots,4$, given the four (fixed and deterministic) principal pairs $(\tilde{\mathbf{u}}_i,\tilde{\mathbf{v}}_i)\in\mathbb{R}^k\times\mathbb{R}^k$ from \eqref{eq:core_pairs}--\eqref{eq:peripheral_pairs}. It remains to observe that, for each $i,j=1\ldots,4$, the mutual exposures
\begin{equation}\label{eq:tilde_lambda_sec4}
\tilde{\lambda}_{ij}:=\frac{1}{m_0} \tilde{\mathbf{v}}_{j} \cdot \tilde{\mathbf{u}}_i = \frac{1}{m_0} \sum_{l=1}^k (\tilde{\mathbf{v}}_{j})_l (\tilde{\mathbf{u}}_i)_l   = \frac{1}{m_0} \sum_{l=1}^k \tilde{u}_{il} \tilde{v}_{lj}
\end{equation}
are precisely those of $\tilde{\lambda}_{4\times4}$ in  \eqref{eq:lambda_tilde}.
Writing $\tau^x_{l,\theta}=\tau^x_{\phi_l(\theta)}$ and $X^x_{l,\theta}=X^x_{\phi_l(\theta)}$, it therefore follows from \eqref{eq:varpi_sec4}--\eqref{eq:varpi_pushforward} and \eqref{eq:tilde_lambda_sec4} that the general formulation of the mean field limit \eqref{CMV} simplifies to that of \eqref{eq:concrete_mean_field_limit} from Section 4, as desired.

\section{Conclusion}\label{sec:discussion}

In this work we introduced the first combined model that considers an Eisenberg--Noe style framework for interbank contagion which can be recast as an interacting particle system with a well-defined mean field limit. Therefore, we are able to draw a direct connection between these previously disparate frameworks for systemic risk, focusing either on the resolution of default cascades in finite bank networks or stochastic dynamics with simple mean field interactions.

The proposed contagion mechanism considers banks with stochastic external assets which, if they drop, can cause defaults before the maturity of all claims.  This is handled first for a finite number of institutions in a purely Eisenberg--Noe style framework, thus extending \cite{BBF18} to account for early defaults. Next, we demonstrate a limiting behaviour as the number of banks increase, which provides justification for performing the analysis of contagion at the level of the mean field limit. In this way, one can significantly lower the parameter space, and it becomes more tractable to pursue analytical results for the regularity of the system's evolution in time.  Moreover, one can circumvent the curse of dimensionality and avoid the slow convergence of Monte Carlo based methods, for example by implementing an analogue of the numerical scheme from \cite{LS18a} as we did in Figure \ref{fig:heat_plots} (alternatively, one could attempt to adapt the semi-analytical approach of \cite{KLR18c}).  As regards the antecedent mean field literature, we provide a more convincing financial underpinning for \cite{HLS18,HS18,LS18a, NS17,NS18}, while also extending the analytical results of \cite{HLS18, LS18a, LS18b} to allow for heterogeneous interactions and a more general form of contagion. Furthermore,  we add to \cite{NS18} by introducing the cascade condition for the resolution of instantaneous default cascades (i.e., jump sizes) as well as establishing results on convergence and uniqueness.

The model of default cascades presented herein can be utilized to answer many questions in systemic risk that are typically intractable analytically (as well as computationally inefficient) for finite bank networks. Additionally, the mean field limit allows for a cleaner analysis of key `systemic' quantities, as exemplified by the mean field cascade condition that gives a precise characterisation of default cascades with an instantaneous `systemic' impact. One important new strand of literature is that of network valuation adjustments~\cite{barucca2016valuation,BF18comonotonic}, in which prices of securities should account for the full network effects. In this regard, the stochastic dynamics underlying the framework herein makes it well-suited for, e.g., pricing credit default swaps on the financial system. By further utilizing the mean field limit, the lower parameter space can facilitate calibration of the stochastic dynamics, and this also opens up the possibility of relying on more analytical methods. These problems are intimately related with systemic risk measures.  For instance, the value-at-risk or CoVaR \cite{adrian2011covar} of the financial system are related to mappings such as
\[
a \mapsto \bbp(\mathbf{L}(t) > a) \quad \text{ and } \quad (a,\delta) \mapsto \bbp\bigl(\mathbf{L}({t+\delta}) - \mathbf{L}(t) > a\bigr),
\]
for a given time $t$.  More specifically, an interesting modification of CoVaR in the mean field limit for core-periphery systems, as discussed in Section~\ref{sec:mf}, is for consideration of the health of the aggregate system conditional on the stress of one of the ``groups'' of institutions. In fact, such structures may allow for the tractable consideration of general systemic risk measures of~\cite{chen2013axiomatic,feinstein2014measures,fouque2015systemic} as well.
Additionally, rather than applying these network valuation adjustments for measuring systemic risk in exogenously provided network structures, the pricing of risk in such a way may allow for considerations of endogenous network formation.  In such a setting, each financial institution would choose to invest in external projects or engage in interbank markets so as to solve some portfolio optimization problem.  Only with a consideration of credit pricing in a financial network would such endogenous network formation be tractable, and we believe this points towards an important avenue of future research.

\appendix
\section{Appendix}\label{sec:technical_appendix}

This appendix contains the proofs of the main results from Section \ref{sec:well-posedness} and is organised into three subsections. The first two subsections address the proofs of Theorems \ref{MV_thm} and \ref{thm_weak_feedback}, respectively, while the final subsection is focused on the identification of the mean field problem \eqref{CMV} as the large population limit of the finite particle system \eqref{general_dist_to_def_sys}. Throughout the appendix, we will be working under Assumptions \ref{MV_assump} and  \ref{assump:paramters_emperical}.

\subsection{Proof of Theorem \ref{MV_thm}}\label{appendix_proof1}

We introduce the notation $\Vert f \Vert_{t}:=\Vert f \Vert_{{L{}}^{\infty}(0,t)}$ and recall the notation $S(\mathbf{v})=\mathrm{supp}(\mathrm{Law}(\mathbf{v}))$ from Assumption \ref{assump:paramters_emperical}. Given this, we can consider  the space of continuous maps $v\mapsto\ell_{v}(\cdot)$ from $S(\mathbf{v})$ to $L^{\infty}(0,t)$, denoted by
\[
C^*_t:= C\bigl(S(\mathbf{v});L^{\infty}(0,t) \bigr),
\]
with respect to the supremum norm
\[
\Vert \ell \Vert_t^* := \sup_{v\in S(\mathbf{v})} \Vert \ell_v(\cdot)\Vert_t.
\]
Since the domain $S(\mathbf{v})$ is a compact subset of $\mathbb{R}^k$, by Assumption \ref{assump:paramters_emperical}, and the codomain $L^{\infty}(0,t)$ is a Banach space, this norm makes $C_t^*$ a Banach space. In order to construct a regular solution to the McKean--Vlasov problem (\ref{CMV}, $\rho=0$), until an explosion time, we will work with the map $\Gamma$, defined in (\ref{eq:Gamma_fp_map}) below. Our strategy is to identify a suitable closed subset of $C_t^*$, for small enough $t>0$, on which we can apply Banach's fixed point theorem.

\subsubsection{Comparison argument and existence of regular solutions}\label{exist_reg_soln}

Given $T>0$, we define the map $\Gamma : C_T^* \mapsto C_T^*$ by
\begin{equation}\label{eq:Gamma_fp_map}
\Gamma[\ell]_{v}(t):=\sum_{l=1}^k v_l \int_{\mathbb{R}_+ \times \mathbb{R}^k \times \mathbb{R}^k } \hat{u}_l \mathbb{P}(t\geq \tau^{x,\ell}_{\hat{u},\hat{v}}) d \varpi(x,\hat{u},\hat{v}),
\end{equation}
for all $t\in[0,T]$ and $v\in S(\mathbf{v})$, where
\begin{equation}\label{X^ell}
\begin{cases}
\tau^{x,\ell}_{u,v}=\inf\{t > 0: X^{x,\ell}_{u,v}(t)\leq0\} \\[6pt]
X^{x,\ell}_{u,v}(t)=x + \int_0^t \!b_{u,v}(s)ds + \int_0^t\!\sigma_{u,v}(s) dB_s - F\bigl( \int_0^t \!g(s)d\ell_v(s) \bigr).
\end{cases}
\end{equation}
Note that, as long as $s \mapsto \ell_v(s)$ is continuous or of finite variation, the integral of $g$ against $\ell_v$ in (\ref{X^ell}) is well-defined, since $g$ is both continuous and of finite variation by Assumption \ref{MV_assump} (see e.g.~\cite[Sect.~1.2]{stroock_integration}). Naturally, all of the results that follow are stated for inputs such that the mapping makes sense.

The cornerstone of our analysis is the next comparison argument. It leads us to the fixed point argument for existence of regular solutions, and it reappears in the generic uniqueness argument of Section \ref{generic uniqueness} which completes the full statement of Theorem \ref{MV_thm}.

\begin{lemma}[Comparison argument]\label{lemma1_HLS} Fix any two $\ell,\bar{\ell}\in C_T^*$ such that $s\mapsto\ell_v(s)$ and $s\mapsto \bar{\ell}_v(s)$ are increasing with $\ell_v(0)=\bar{\ell}_v(0)=0$ for all $v\in S(\mathbf{v})$. Fix also $t_0>0$ and suppose $s\mapsto\ell_v(s)$ is continuous on $[0,t_0)$ for all $v\in S(\mathbf{v})$. Then we have
	\[
	\bigl( \Gamma[\ell]_v(t) - \Gamma[ \bar{\ell} \, ]_v(t) \bigr )^+ \leq C \Vert (\ell - \bar{\ell}\,)^+\Vert_{t}^* \int_0^t (t-s)^{-\f{1}{2}} d \Gamma[\ell]_v(s),
	\]
	for all $t<t_0$ and all  $v\in S(\mathbf{v})$, where $C>0$ is a fixed numerical constant (i.e., it is independent of $t_0$ and $v$).
\end{lemma}
\begin{proof} Fix $t<t_0$. Recalling that $\ell_v(0)=0$, integration by parts (see e.g., \cite[Sect.~1.2]{stroock_integration}) gives
	\begin{equation*}
	\int_0^t g(s)d\ell_v(s) = g(t)\ell_v(t) + \int_0^t \ell_v(s) d(-g)(s),
	\end{equation*}
	and likewise for $\bar{\ell}_v$. Using this and the assumptions on $F$, $g$, $\ell$, and $\bar{\ell}$, we have
	\begin{align*}
	F\Bigl( \!\int_0^t g(s) & d\ell_v(s) \Bigr) - F\Bigl( \int_0^t g(s)d\bar{\ell}_v(s) \Bigr) \leq  \Vert F \Vert_{\text{Lip}} \Bigl( \int_0^t g(s)d\ell_v(s)  - \int_0^t g(s)d\bar{\ell}_v(s) \Bigr)^+\\
	&= \Vert F \Vert_{\text{Lip}} \Bigl( g(t)(\ell_v(t)-\bar{\ell}_v(t)) + \int_0^t (\ell(s)-\bar{\ell}(s)) d(-g)(s) \Bigr)^+\\
	&\leq g(t) \Vert F \Vert_{\text{Lip}} \bigl(\ell_v(t) - \bar{\ell}_v(t) \bigr)^+ +  \Vert F \Vert_{\text{Lip}} \int_0^t \bigl(\ell_v(s) - \bar{\ell}_v(s)\bigr)^+ d(-g)(s) \\
	&\leq   g(0) \Vert F \Vert_{\text{Lip}} \Vert (\ell_v - \bar{\ell}_v)^+\Vert_{t}.
	\end{align*}
	Thus, taking the difference between the two processes $X^{x,\ell}_{u,v}$ and $X^{x,\bar{\ell}}_{u,v}$, as defined in (\ref{X^ell}) coupled through the same Brownian motion, it follows that
	\[
	X^{x,\bar{\ell}}_{u,v}(t) - X^{x,\ell}_{u,v}(t) \leq   g(0) \Vert F \Vert_{\text{Lip}} \Vert (\ell_v - \bar{\ell}_v)^+\Vert_{t}.
	\]
	Therefore, using the continuity of $\ell_v$, for any $s\in[0,t]$, it holds on the event $\{ \tau^{x,\ell}_{u,v} = s  \}$ that
	\[
	X^{x,\bar{\ell}}_{u,v}(s) = X^{x,\bar{\ell}}_{u,v}(s) - X^{x,\ell}_{u,v}(s) \leq  g(0) \Vert F \Vert_{\text{Lip}} \Vert (\ell - \bar{\ell})^+\Vert_{s}^*.
	\]
	Based on this, we can replicate the arguments from \cite[Prop.~3.1]{HLS18}, by instead conditioning on the value of $\tau^{x,\ell}_{u,v}$ and using the previous inequality, to deduce that
	\begin{align*}
	\mathbb{P}(t\geq \tau^{x,\ell}_{u,v}) & - \mathbb{P}(t\geq \tau^{x,\bar{\ell}}_{u,v} )  \\ &\leq \int_0^t  \mathbb{P}\bigl( \inf_{r\in[s,t]} \textstyle\int_s^r\!\sigma_{u,v}(h) dB_h > - g(0)\Vert F \Vert_{\mathrm{Lip}} \Vert (\ell - \bar{\ell})^+\Vert_{s}^* \bigr ) d \mathbb{P}(s\geq \tau^{x,\ell}_{u,v}).
	\end{align*}
	Performing a time change in the Brownian integral, and using that there is a uniform $\epsilon>0$ such that $\epsilon \leq \sigma_{u,v} \leq \epsilon^{-1}$, by Assumption \ref{MV_assump}, it follows from the law of the infimum of a Brownian motion that
	\[
	\mathbb{P}(t\geq \tau^{x,\ell}_{u,v}) - \mathbb{P}(t\geq \tau^{x,\bar{\ell}}_{u,v} )
	\leq
	C \Vert (\ell - \bar{\ell}\,)^+\Vert_{t}^*\int_0^t (t-s)^{-\f{1}{2}} d \mathbb{P}(s\geq \tau^{x,\ell}_{u,v})
	\]
	where the constant $C>0$ is independent of $t$, $x$, $u$, and $v$. Now fix any $\tilde{v}\in S(\mathbf{v})$. Multiplying both sides of the above inequality by $\sum_{l=1}^k\tilde{v}_lu_l$ and recalling that this is non-negative for all $u$ in the support of $\varpi$, by Assumption \ref{assump:paramters_emperical}, we can then integrate both sides of the resulting inequality against $\varpi$, over $(x,u,v)\in \mathbb{R}_+ \times \mathbb{R}^k \times \mathbb{R}^k$, to arrive at
	\[
	\Gamma[\ell]_{\tilde{v}}(t) -  \Gamma[\bar{\ell}\,]_{\tilde{v}}(t)   \leq C \Vert (\ell - \bar{\ell}\,)^+\Vert_{t}^*\int_0^t (t-s)^{-\f{1}{2}} d \Gamma[\ell]_{\tilde{v}}(t),
	\]
	for all $t<t_0$, for some fixed numerical constant $C>0$ independent of $t_0$ and $\tilde{v}$. As the right-hand side is positive, this proves the lemma.
\end{proof}

For any $\gamma\in(0,1/2)$, $A>0$, and $t>0$, we define the space $\mathcal{S}(\gamma, A, t) \subset C_t^*$ by
\begin{equation}
\mathcal{S}(\gamma, A, t) := \bigl\{  \ell \in C\bigl(S(\mathbf{v});H^1(0,t) \bigr) : \ell_v^\prime(t) \leq A t^{-\gamma} \;\; \text{for a.e.} \; t\in[0,t], v\in S(\mathbf{v})     \bigr\},
\end{equation}
which is a complete metric space with the metric inherited from $C_t^*$. Moreover, we define the map
\[
\hat{\Gamma}[\ell;u,v](t):= \int_0^\infty \mathbb{P}(t\geq \tau^{x,\ell}_{u,v}) d \nu_0(x|u,v),
\]
so that $\Gamma[\ell]_{\tilde{v}}(t)=\sum_{l=1}^k \tilde{v}_l \int_{\mathbb{R}^k\times\mathbb{R}^k} u_l\hat{\Gamma}[\ell;u,v](t) d\hat{\varpi}(u,v)$. Then, for each $u$ and $v$, we can replicate the arguments from \cite[Sect.~4]{HLS18} for the function $t\mapsto \hat{\Gamma}[\ell;u,v](t)$ in place of the corresponding function considered there. Given $\hat{\varpi}$ and $V_0(\cdot|u,v)$ satisfying (\ref{eq:cond_V_0_decay}), we can thus conclude (by arguing precisely as in \cite[Prop.~4.9]{HLS18}), that there exists $A>0$ such that, for any $\varepsilon_0>0$, there is a small enough time $t_0>0$ for which
\begin{equation}\label{eq:Gamma_stable}
\Gamma : \mathcal{S}\bigl(\tfrac{1-\beta}{2}, A+\varepsilon_0, t_0 \bigr) \rightarrow  \mathcal{S}\bigl(\tfrac{1-\beta}{2}, A+\varepsilon_0, t_0 \bigr),
\end{equation}
where $A$ only depends on $C_\star$ and $x_\star$ from (\ref{eq:cond_V_0_decay}). Moreover, by analogy with \cite[Thm.~1.6]{HLS18}, we can deduce from Lemma \ref{lemma1_HLS}  that $\Gamma$ is a contraction on this space for small enough $t_0$. Therefore, the small time existence of a regular solution $\mathcal{L}_v^*(t)=\sum_{l=1}^k v_l \mathcal{L}_l(t)$ to (\ref{CMV}, $\rho=0$) now follows from an application of Banach's fixed point theorem as in the proof of \cite[Thm.~1.7]{HLS18}. Finally, by replicating the bootstrapping argument from the proof of \cite[Cor.~5.3]{HLS18}, we conclude that the regular solution extends until the first time $t_{\star}$ such that the $H^1$ norm of $(\mathcal{L}_1,\ldots,\mathcal{L}_k)$ diverges on $(0,t_\star)$. This proves the first part of Theorem \ref{MV_thm}.

\subsubsection{Generic uniqueness}\label{generic uniqueness}
It remains to verify that the general uniqueness result of \cite[Thm.~1.8]{HLS18} can be extended to the present setting, which will follow from the two lemmas below. The first lemma concerns a family of auxiliary McKean--Vlasov problems given by
\begin{equation}\label{eq:X_eps}
\begin{cases}
X^{x,\epsilon}_{u,v}(t)=x\mathbf{1}_{x\geq\varepsilon} - \frac{\varepsilon}{4} + \int_0^t \!b_{u,v}(s)ds + \int_0^t \!\sigma_{u,v}(s) dB_s - F\bigl( g(0) \lambda_v^\eps + \int_0^t \!g(s)d\mathbf{L}^\eps_v (s) \bigr)\\[6pt]
\mathbf{L}^{\eps}_{v}(t)= \sum_{l=1}^k v_l \int_{\mathbb{R}_+\times\mathbb{R}^k\times \mathbb{R}^k} \hat{u}_l \mathbb{P}(\tau^{x,\eps}_{\hat{u},
	\hat{v}}\leq t)d\varpi(x,\hat{u},\hat{v}) \\[6pt] \tau^{x,\eps}_{u,v}=\inf\{t\geq0:{X}^{x,\eps}_{u,v}(t)\leq0\},
\end{cases}
\end{equation}
for $\varepsilon>0$, where $\lambda_{v}^\eps:= \sum_{l=1}^k v_l \int_{(0,\eps)\times\mathbb{R}^k\times\mathbb{R}^k} \hat{u}_l d\varpi(x,\hat{u},\hat{v}) $.
This family of equations will serve as the equivalent of the `$\eps$-deleted solutions' introduced in \cite[Sect.~5.2]{HLS18}. By writing
\begin{equation}\label{eq:eps_loss}
\mathbf{L}_v^{\eps}(t)=\lambda_v^\eps+\tilde{\mathbf{L}}_v^{\varepsilon}(t) \quad \text{with} \quad \tilde{\mathbf{L}}_v^{\varepsilon}(t)= \sum_{l=1}^k v_l \int_{[\eps,\infty)\times\mathbb{R}^k\times \mathbb{R}^k}\hat{u}_l\mathbb{P}(\tau^{x,\eps}_{\hat{u},\hat{v}}\leq t)d\varpi(x,\hat{u},\hat{v}),
\end{equation}
we can show that these $\eps$-deleted problems are well-posed with regularity estimates that are uniform in $\eps>0$.

\begin{lemma}[Uniformly regular $\eps$-deleted solutions]\label{lemma2_HLS}
	There is an $\eps_0>0$ such that (\ref{eq:X_eps}) has a family of solutions $\{\mathbf{L}^\eps\}_{\eps\leq \eps_0}$ which are uniformly regular in the following sense: There exists $A>0$ and $t_0>0$ such that ${\mathbf{L}{}}^{\eps} \in \mathcal{S}(\tfrac{1-\beta}{2}, A, t_0)$ uniformly in $\eps\in(0,\eps_0]$. 
\end{lemma}
\begin{proof}First of all, we can note that $\lambda_v^\eps \leq C_\star \eps^{1+\beta}/(1+\beta)$ uniformly in $v$, for small enough $\eps>0$, by (\ref{eq:cond_V_0_decay}), and clearly $F(x)=o(x^{1/(1+\beta)})$ as $x\downarrow0$, since $F$ is Lipschitz with $F(0)=0$. Hence there exists $\eps_0>0$ such that $F(g(0)\lambda_v^\eps)\leq \eps/4$ for all $\eps\in(0,\eps_0)$.  Next, using (\ref{eq:eps_loss}) and making the change of variables $y=x-\eps/4 - F(g(0)\lambda_v^\eps)$ in (\ref{eq:X_eps}) we obtain the equivalent formulation
	\begin{equation*}
	\begin{cases}
	\tilde{X}^{y,\eps}_{u,v}(t)=y+\int_0^t b_{u,v}(s)ds + \int_0^t\sigma_{u,v}(s) dB(s) - F\bigl( g(0)\lambda_v^\eps + \int_0^t g(s)d \tilde{\mathbf{L}}_v^{\varepsilon}(s) \bigr) + F\bigl( g(0)\lambda_v^\eps \bigr)
	\\[4pt]
	\tilde{\mathbf{L}}_v^{\varepsilon}(t)= \sum_{l=1}^k v_l \int_{\mathbb{R}_+\times\mathbb{R}^k\times \mathbb{R}^k}\hat{u}_l\mathbb{P}(\tilde{\tau}^{y,\eps}_{\hat{u},\hat{v}}\leq t)V^\eps_0(y|\hat{u},\hat{v})dyd\hat{\varpi}(\hat{u},\hat{v})\\[4pt]
	V_0^\eps(y|u,v)=V_0\bigl(y+\f{\eps}{4} + F(g(0)\lambda_v^\eps) \bigr) \mathbf{1}_{\{ y + \f{\eps}{4} + F(g(0)\lambda_v^\eps)  \geq \eps  \}}
	\\[4pt]
	\tilde{\tau}^{y,\eps}_{u,v}=\inf \{ t\geq 0 : \tilde{X}^{y,\eps}_{u,v}(t) \leq 0 \} 
	\end{cases}
	\end{equation*}
	Now take $\eps_0 \leq x_\star$. Recalling that $F(\lambda_v^\eps)\leq \eps/4$ for all $\eps\in(0,\eps_0)$, we can then observe that
	\begin{align*}
	V_0^\eps(y|u,v) & \leq C_\star \bigl( y+\eps/4 + F(g(0)\lambda_v^\eps)  \bigr) \mathbf{1}_{y+\eps/4 + F(g(0)\lambda_v^\eps) \geq \eps} \\
	& \leq (y+\eps/2 )^\beta \mathbf{1}_{y \geq \eps/2} \leq 2^\beta C_\star y^\beta
	\qquad \text{for all} \quad y<x_\star/2,
	\end{align*}
	uniformly in $u$, $v$, and $\eps \in (0,\eps_0)$. Therefore, for each $\eps\in (0,\eps_0)$, we can indeed construct a regular solution $\tilde{\mathbf{L}}^{\varepsilon}$ to the above system by the first part of Theorem \ref{MV_thm} (as proved in Section \ref{exist_reg_soln}). Moreover, since the boundary control on $V_0^\eps(\cdot|u,v)$ is uniform in $\eps\in(0,\eps_0)$, uniformly in $u$ and $v$, it follows from the fixed point argument in Section \ref{exist_reg_soln} that the regularity of the solutions $\tilde{\mathbf{L}}^\eps$ is also uniform in $\eps\in(0,\eps_0)$.
	By (\ref{eq:eps_loss}), the uniform regularity of the original family $\mathbf{L}^\eps$ follows a fortiori from that of $\tilde{\mathbf{L}}^\eps$, and thus the proof is complete.
\end{proof}

Armed with Lemma \ref{lemma2_HLS}, we can now proceed to the final ingredient of the general uniqueness result, namely the `monotonicity and trapping' argument from \cite[Sect.~5.2]{HLS18}.

\begin{lemma}[Monotonicity and vanishing envelope]\label{lemma3_HLS} 
	Let $\mathbf{L{}}^*:S(\mathbf{v})\rightarrow L^\infty(0,T)$ be given by
	\[
	\mathbf{L{}}^*_v(t):=\sum_{l=1}^k v_l \mathcal{L}_l = \sum_{l=1}^k v_l  \int_{\mathbb{R}_+ \times \mathbb{R}^k \times \mathbb{R}^k} \hat{u}_l \mathbb{P}( t\geq \tau^x_{\hat{u},\hat{v}}) d\varpi(x,\hat{u},\hat{v})
	\]
	for a generic solution to \eqref{CMV} with $\rho=0$, and suppose there are no jumps of $(\mathcal{L}_1,\ldots,\mathcal{L}_k)$ on $[0,t_0)$ so $s\mapsto \mathbf{L{}}^*_v(s)$ is continuous on $[0,t_0)$ for all $v$. If $\mathbf{L}^\eps$ is a continuous `$\varepsilon$-deleted' solution on $[0,t_0)$, then $\mathbf{L}^\eps_v> \mathbf{L}_v$ on $[0,t_0)$, for all $v \neq 0$.  Moreover, if $\mathbf{L}$ is regular on $[0,t_0)$ and the family $\{ \mathbf{L}^\eps \}$ is uniformly regular on $[0,t_0)$, in the sense of Lemma \ref{lemma2_HLS}, then there is a $t_1\in(0,t_0)$ such that the envelope between the two is vanishing on $[0,t_1]$, that is, $\Vert \mathbf{L} - \mathbf{L}^\eps \Vert_{t_1}^* \rightarrow 0$ as $\varepsilon\rightarrow 0$.
\end{lemma}

\begin{proof}
	Noting that $\mathbf{L}^\eps_v(0)=\lambda_v^\epsilon > 0=\mathbf{L}_v(0)$ for $v\neq0$, towards a contradiction we let $t \in(0, t_0)$ be the first time $\mathbf{L}^\eps_v(t) = \mathbf{L}_v(t)$ for some $v\neq0$. Then it holds for any $s<t$ that
	\begin{align*}
	g(0)\lambda_v^\eps + \int_0^s \!g(r)d\mathbf{L}^\eps_v (r) & = \int_0^s \mathbf{L}^\eps_v(r) d(-g)(r) + g(s)\mathbf{L}^\eps_v(s) \nonumber\\
	& \geq  \int_0^s \mathbf{L}_v(r) d(-g)(r) + g(s)\mathbf{L}_v(s) = \int_0^s g(r)d\mathbf{L}_v(r).
	\end{align*}
	and, since $F$ is increasing, we thus have
	\begin{equation}\label{eq:eps/4_1}
	X^x_{u,v}(s)-X^{x,\eps}_{u,v}(s) = x\mathbf{1}_{x<\varepsilon}+\frac{\varepsilon}{4} +  F\Bigl( g(0) \lambda_v^\eps + \int_0^s \!gd\mathbf{L}^\eps_v \Bigr) - F\Bigl(\int_0^s \!gd\mathbf{L}_v \Bigr) \geq \frac{\varepsilon}{4},
	\end{equation}
	for all $s\in(0,t)$.
	Arguing as in the proof of \cite[Lemma 5.6]{HLS18}, it follows from (\ref{eq:eps/4_1})  that
	\begin{align*}
	\mathbf{L}_v^\eps(t) & \geq \mathbf{L}_v(t) + \sum_{l=1}^k v_l  \int_{\mathbb{R}_+ \times \mathbb{R}^k \times \mathbb{R}^k} \hat{u}_l \mathbb{P} \Bigl( \inf_{r\in[0,t] } X^x_{u,v}(s) \in (0, \eps/4 ] \Bigr) d\varpi(x,\hat{u},\hat{v}) > \mathbf{L}_v(t),
	\end{align*}
	which contradicts the definition of $t$, thus proving the first claim.
	
	For the second claim, can now rely on the fact that $\mathbf{L}^\eps_v>\mathbf{L}_v$ on $[0,t_0)$ for all $v\neq 0$. Consequently, since $X^{x,\varepsilon}_{u,v}(s)=0$ on the event $\{ \tau^{x,\varepsilon}_{u,v}=s \}$, we deduce that, on this event,
	\begin{equation}\label{eq:envelope}
	X^{x}_{u,v}(s) = X^{x}_{u,v}(s) - X^{x,\eps}_{u,v}(s) \leq \varepsilon + \frac{\varepsilon}{4} + g(0) \Vert F \Vert_{\mathrm{Lip}} \Vert \mathbf{L}^\eps_v -\mathbf{L}_v \Vert_s,
	\end{equation}
	where the inequality follows by the equality in (\ref{eq:eps/4_1}) and the same estimate as in the proof of Lemma \ref{lemma1_HLS}. From here, (\ref{eq:envelope}) allows us to replicate the proof of \cite[Lemma 5.7]{HLS18}, only with the term `$g(0) \Vert F \Vert_{\mathrm{Lip}} \Vert \mathbf{L}^\eps -\mathbf{L} \Vert_s^*$' in place of the term `$\alpha(L^{\eps}_s -L_s)$' appearing in that proof. This verifies the second claim.
\end{proof}

Based on Lemmas \ref{lemma1_HLS} and \ref{lemma3_HLS}, the uniqueness part of Theorem \ref{MV_thm} now follows immediately by retracing the proof of \cite[Thm.~1.8]{HLS18} (at the very end of \cite[Sect.~5]{HLS18}) with the cascade condition \eqref{eq:limit_PJC} taking the place of the physical jump condition \cite[(1.7)]{HLS18}.

\subsection{Proof of Theorem \ref{thm_weak_feedback} }\label{Proof_thm_uniq_common}

Let us begin by proving the continuity of a given solution satisfying the smallness condition \eqref{smallness_cond}. To this end, we fix a pair $(u,v)$ and write $X_{u,v}^x(t)=x+ Y(t) +Y_0(t)$ with $Y(t):=\int_0^t\rho\sigma_{u,v}(s)dB_s$ so $Y_0(t)$ is $B_0$-measurable. Letting $p(t,\cdot)$ denote the density of $Y(t)$, it follows from Tonelli's theorem that
\begin{align*}
\int_{\mathbb{R}_+} \mathbb{P} \bigl( X^x_{u,v}(t) &\in A , \, t < \tau^x_{u,v}  \mid B_0\bigr) V_0(x|u,v)dx \leq \int_{\mathbb{R}} \int_{A} p(t,y+x+Y_0(t)) V_0(x|u,v)dy dx \\
& =  \int_{A} \int_{\mathbb{R}} p(t,x+y+Y_0(t)) V_0(x|u,v) dx dy \leq \Vert V_0(\cdot | u,v) \Vert_\infty |A|,
\end{align*}
for all $A\in\mathcal{B}(\mathbb{R})$, since $p(t,\cdot)$ integrates to $1$. Recalling the definition of $V_t$ from \eqref{eq:the_nu_t_measure}, this shows that  $V_t(x|u,v)\leq \Vert V_0(\cdot | u,v) \Vert_\infty $ for all $x\in(0,\infty)$ and all times $t\geq0$. Therefore, the criterion \eqref{no_jump_condition1} holds for all times, by the smallness condition \eqref{smallness_cond}, and hence the given solution must be globally continuous in time.

To prove the uniqueness part of Theorem \ref{thm_weak_feedback}, we show how to extend the arguments behind \cite[Thm.~2.3]{LS18b}. Let $(X,\mathcal{L})$ and $(\bar{X},\mathcal{\bar{L}})$ be any two solutions to (\ref{CMV}) coupled through the same Brownian drivers $B$ and $B_0$. Then we define the increasing processes
\[
\mathbf{L}_v:=\sum_{l=1}^k v_l \mathcal{L}_l \quad \text{and} \quad \bar{\,\mathbf{L}}_v:=\sum_{l=1}^k v_l \bar{\mathcal{L}}_l
\]
for every $v\in S(\mathbf{v})$. Retracing the arguments of \cite[Lemma 2.1]{LS18b}, and applying Fubini's theorem, we can deduce that
\begin{align*}
\loss_v(s) - \bar{\,\loss}_v(s)  \leq \mathbb{E} & \biggl[ \sum_l^k v_l \!\int_{\mathbb{R}^k\times \mathbb{R}^k} \hat{u}_l \int_{\bar{I}_{\hat{v}}(s)}^{I_{\hat{v}}(s)} \!\!V_0(x|\hat{u},\hat{v})dx d\hat{\varpi}(\hat{u},\hat{v}) \,\Big|\, B^0\biggr],
\end{align*}
where
\[
I_v(t) := \sup_{s\leq t}\biggl\{ F\Bigl( \!\int_0^s g(r) d\mathbf{L}_v(r)  \Bigr) - Z_s \biggr\}, \quad \mathrm{and} \quad \bar{I}_v(t) := \sup_{s\leq t}\biggl\{ F\Bigl( \!\int_0^s g(r) d\bar{\mathbf{L}}_v(r)  \Bigr) - Z_s \biggr\},
\]
with
\[
Z_s:=\int_0^sb(r)dr+\int_0^s \sigma(r)d(\rho B^0 + \sqrt{1-\rho^2}B)(r).
\]
By symmetry, $\bar{\,\loss}_v(s)- \loss_v(s)$ satisfies the analogous bound with $I_v(s)$ and $\bar{I}_v(s)$ interchanged. Furthermore, by simply repeating the first estimate from the proof of Lemma \ref{lemma1_HLS}, only with $\mathbf{L}$ and $\!\bar{\,\mathbf{L}}$ in place of $\ell$ and $\bar{\ell}$, we have
\begin{align*}
F\Bigl( \!\int_0^t g(s) & d\mathbf{L}_v(s) \Bigr)  \geq F\Bigl( \int_0^t g(s)d\bar{\,\mathbf{L}}_v(s) \Bigr)   - g(0) \Vert F \Vert_{\text{Lip}} \Vert \mathbf{L} - \!\bar{\,\mathbf{L}}\Vert_{t}^*.
\end{align*}
Therefore, relying on this inequality together with the previous observation, we can retrace the arguments of \cite[Theorem 2.2]{LS18b} and conclude that
\begin{align*}
| \loss_v(s) - \bar{\,\loss}_v(s) | \leq g(0) \Vert F \Vert_{\text{Lip}}  \Vert \loss- \bar{\,\loss}\Vert_{s}^*  \sum_{l=1}^k v_l \!\int_{\mathbb{R}^k\times \mathbb{R}^k} \hat{u}_l  \Vert V_0(\cdot |\hat{u},\hat{v}) \Vert_{\infty} d\hat{\varpi}(\hat{u},\hat{v}) .
\end{align*}
At this point, the smallness condition \eqref{smallness_cond} gives
\[
| \loss_v(s) - \bar{\,\loss}_v(s) | \leq (1-\delta) \Vert \loss- \bar{\,\loss}\Vert_{s}^*,
\]
so, taking a supremum over $v\in S(\mathbf{v})$ on the left-hand side, we conclude that there is pathwise uniqueness.

\subsection{Convergence of the particle system}\label{convergence_appendix}

In this section we outline how the convergence to the conditional McKean--Vlasov problem \eqref{CMV} can be established by retracing the approach of \cite{LS18a} after some adjustments. The arguments rely heavily on specific properties of the M1-topology for the Skorokhod space of c\`adl\`ag paths. The reader is referred to \cite{whitt_2002} for an introduction to this topology. For concreteness, let us restrict to random start points $X_0^i$ satisfying \eqref{eq:cond_V_0_decay} near the absorbing boundary at zero, although it is possible to consider higher generality in these arguments. 

Let $D_\mathbb{R}$ denote the space of real-valued c\`adl\`ag paths on $[0,T]$, and let $(X_1,\ldots,X_n)$ be the unique strong solution to the particle system \eqref{general_system_rank-k} of size $n$ in $D_\mathbb{R} \times \cdots \times D_\mathbb{R}$ (recall Proposition \ref{prop:particle_sys_well-posed}). Moreover, as usual, we let $\mathbf{u}_i \in \mathbb{R}^k$ and $\mathbf{v}_i\in \mathbb{R}^k$ denote the type vectors from \eqref{eq:vectors_u_v}, for $i=1,\ldots,n$. For simplicity of notation, we are suppressing the dependence on $n\geq1$ in each triple $(\mathbf{u}_i,\mathbf{v}_i,X_i)\in \mathbb{R}^k\times\mathbb{R}^k\times D_\mathbb{R}$. Now consider the empirical measures
\begin{equation}
\mathbf{P}^n:= \frac{1}{n} \sum_{i=1}^n \delta_{\mathbf{u}_i} \otimes \delta_{\mathbf{v}_i} \otimes \delta_{X_i(\cdot)}, \qquad \text{for} \quad n \geq 1,
\end{equation}
which are random variables valued in the space of probability measures $\mathcal{P}(\mathbb{R}^k\times\mathbb{R}^k\times D_\mathbb{R})$. For $(u,v,\eta)\in\mathbb{R}^k\times\mathbb{R}^k\times D_\mathbb{R} $, we define the coordinate projections $\pi_{1,l}(u,v,\eta):=u_l$, $\pi_{2,l}(u,v,\eta):=v_l$, and $\pi_{3,t}(u,v,\eta):=\eta_t$ as well as $\pi_{t}(u,v,\eta):=(u,v,\eta_t)$ and $\pi_{(1,2)}(u,v,\eta):=(u,v)$. Writing $\mathbf{P}^n_t:=\mathbf{P}^n \circ \pi_{t}^{-1}$ and $\hat{\varpi}^n:=\mathbf{P}^n \circ \pi_{(1,2)}^{-1}$, we have $\mathbf{P}^n_0 = \varpi^n \Rightarrow \varpi$ and $\hat{\varpi}^n\rightarrow \hat{\varpi}$ by virtue of Assumption \ref{assump:paramters_emperical}. The first task is to ensure tightness of the pair $(\mathbf{P}^n,B_0)$ so that we can extract weakly convergent subsequences.

\begin{lemma}[Tightness]\label{prop:tightness}
	The sequence of random variables $(\mathbf{P}^n,B_0)$ is tight on the product space $\mathcal{P}(\mathbb{R}^k\times\mathbb{R}^k\times D_\mathbb{R}) \times C_{\mathbb{R}}$. Here $C_{\mathbb{R}}$ is the space of continuous real-valued paths on $[0,T]$ topologized by uniform convergence, and $\mathcal{P}(\mathbb{R}^k\times\mathbb{R}^k\times D_\mathbb{R})$ is topologized by weak convergence of measures as induced by the M1-topology on $D_\mathbb{R}$.
\end{lemma}
\begin{proof}
	Since $D_\mathbb{R}$ is a Polish space with the M1-topology, it is a classical result (see e.g.~\cite[Ch.I, Prop.~2.2]{sznitman_1991}) that the sequence of (random) empirical measures $\textbf{P}^n$ is tight if, for each $\varepsilon>0$, we can find  $K_\varepsilon$ compact in $\mathbb{R}^k\times\mathbb{R}^k\times D_\mathbb{R}$,  where $D_\mathbb{R}$ comes with the M1-topology, such that, for all $n\geq1$,
	\[
	E_n( K_\varepsilon)\geq 1-\varepsilon, \quad \text{where} \quad E_n( K_\varepsilon):=\frac{1}{n}\sum_{i=1}^n \mathbb{P}\bigl( (\mathbf{u}_i,\mathbf{v}_i,X_i )  \in K_\varepsilon \bigr).
	\]
	To fulfil this, it is sufficient that, for every $\varepsilon>0$, we can find a compact set $K_\varepsilon$ such that $\mathbb{P}( (\mathbf{u}_i,\mathbf{v}_i,X_i )  \in K_\varepsilon )\geq 1-\varepsilon$ for each $i=1,\ldots,n$ uniformly in $n\geq 1$. By Assumption \ref{assump:paramters_emperical}, we have $|\mathbf{u}_i|+|\mathbf{v}_i|\leq C$, for some $C>0$, uniformly in $i=1,\ldots,n$ and $n\geq 1$. Hence we can take $K_\varepsilon $ to be of the form  $K_\varepsilon  = \bar{B}_C \times S_\varepsilon $, where $\bar{B}_C$ is the closed ball of radius $C$ in $\mathbb{R}^{2k}$, and $S_\varepsilon$ is compact in $(D_\mathbb{R},\textrm{M1})$. Consequently, writing $X_i^n$ for the $i$'th particle in the size-$n$ particle system, it suffices to show that each sequence $(X_i^n)_{n\geq1}$ is tight with estimates that are uniform in $i=1,\ldots, n$ and $n\geq 1$. To this end, the first crucial observation is that
	\[
	t\mapsto F\Bigl( \sum_{l=1}^{k} v_{il} \int_0^t g(s) d\mathcal{L}^n_l(s) \Bigr)
	\]
	is increasing. Therefore, exploiting the special nature of the M1-topology, the uniform tightness estimates can be established by retracing the steps of \cite[Prop.~3.9]{LS18a}.
\end{proof}

We now turn to the problem of identifying the limit points of $(\mathbf{P}^n,B_0)$ as $n\rightarrow\infty$, where convergent subsequences are ensured by Prokhorov's theorem in light of the previous lemma. First of all, we define the mapping
\begin{equation}\label{loss_map}
(\mathcal{L}_l(\mu))(t):=\bigl\langle \mu , \pi_{1,l}(\cdot)  \mathbf{1}_{(\infty,0]}\bigl(\inf_{s\leq t}\pi_{3,t}(\cdot)\bigr)\bigr\rangle,
\end{equation}
for $\mu \in \mathcal{P}(\mathbb{R}^k\times\mathbb{R}^k\times D_\mathbb{R})$, where the rationale is of course that 
\begin{equation}\label{loss_empirical}
(\mathcal{L}_l(\mathbf{P}^n))(t) =\frac{1}{n} \sum_{j=1}^{n} u_{jl} \mathbf{1}_{t\geq\tau_j} = \mathcal{L}^n_l(t).
\end{equation}
Using the mappings $\mu\mapsto \mathcal{L}_l(\mu)$, we in turn define
\begin{equation}\label{mathcal_M_map}
(\mathcal{M}(u, v,\eta , \mu))(t) := \eta(t) - \eta(0) - \int_0^t b_{u,v}(s)ds - F\Bigl( \sum_{l=1}^k v_l \int_0^t g(s)d (\mathcal{L}_l(\mu))(s)  \Bigr),
\end{equation}
and we then intend to perform a martingale argument to identify the limit points of $(\mathbf{P}^n,B_0)$ based on mappings of the form
\begin{equation}\label{eq:cont_mapping}
(\mu , w) \mapsto \bigl \langle  \mu , \Psi\bigl(  \mathcal{M}(\cdot  , \mu) , w \bigr)        \bigr \rangle
\end{equation}
for $(\mu,w) \in \mathcal{P}(\mathbb{R}^k\times\mathbb{R}^k\times D_\mathbb{R}) \times C_\mathbb{R}$, for suitable functions $\Psi : D_\mathbb{R} \times C_\mathbb{R} \rightarrow \mathbb{R}$. Indeed, we can observe that
\[
\bigl \langle  \mathbf{P}^n , \Psi\bigl(  \mathcal{M}(\cdot  , \mathbf{P}^n) , B_0 \bigr)        \bigr \rangle = \frac{1}{n} \sum_{i=1}^n \Psi \Bigl( \int_0^t\sigma_{\mathbf{u}_i,\mathbf{v}_i}(s)d\bigl( \rho B_0(s) +\sqrt{1-\rho^2}B_i(s)   \bigr)    , B_0   \Bigr),
\]
where the right-hand side is a nice average of the function $\Psi$ applied to square integrable martingales on $[0,T]$. This will essentially allow us to transfer suitable martingale properties to the limit as $n \rightarrow \infty$, which is the machinery behind the next observations.

Proceeding as in \cite[Lemma 3.11]{LS18a} we can show that \eqref{eq:cont_mapping} and similar mappings are continuous for suitable functions $\Psi$ (the specific mappings are defined immediately before \cite[Lemma 3.11]{LS18a}). Fix a limit point $(\mathbf{P},B_0)$ of $(\mathbf{P}^n,B_0)$, realised along a convergent subsequence (due to Lemma \ref{prop:tightness} above).  Write $\varpi=\mathrm{Law}(\mathbf{u},\mathbf{v},X(0))$, where we recall that $\varpi=\mathbf{P}_0$ is the limit of $\varpi^n=\mathbf{P}^n_0$, as ensured by Assumption \eqref{assump:paramters_emperical}. Retracing the steps of \cite[Prop.~3.12]{LS18a} and \cite[Proof of Thm.~3.2, p.~26]{LS18a}, based on the aforementioned continuity results, we can show that there is a probability space $(\bar{\Omega},\bar{\mathcal{F}},\bar{\mathbb{P}})$ which supports our limiting random variables $(\mathbf{u},\mathbf{v}, \mathbf{P},B^0)$ and also carries a c\`adl\`ag process $X$ as well as a Brownian motion $B \perp B_0$, for which  $(B,B_0)$ is independent of $(\mathbf{u},\mathbf{v},X(0))$, such that
\begin{equation*}
(\mathcal{M}(\mathbf{u},\mathbf{v},X,\mathbf{P}))(t)= \int_0^t\sigma_{\mathbf{u},\mathbf{v}}(s)d\bigl( \rho B_0(s) +\sqrt{1-\rho^2}B(s)   \bigr)
\end{equation*}
holds for all $t\in[0,T]$, $\bar{\mathbb{P}}$-almost surely. In other words, on the background space $(\bar{\Omega},\bar{\mathcal{F}},\bar{\mathbb{P}})$, we have
\begin{align}\label{limit_martingale}
X(t) = X(0) &+ \int_0^t b_{\mathbf{u},\mathbf{v}}(s)ds+  \int_0^t\sigma_{\mathbf{u},\mathbf{v}}(s)d\bigl( \rho B_0(s) +\sqrt{1-\rho^2}B_i(s)   \bigr) \nonumber\\
&- F\Bigl( \sum_{l=1}^k \mathbf{v}_l \int_0^t g(s)d (\mathcal{L}_l(\mathbf{P}))(s)  \Bigr).
\end{align}
To avoid clouding the presentation, let us (for now) assume that the limiting random probability measure $\mathbf{P}$ is known to be $B_0$ measurable. Intuitively, this is what one expects, as the sequence $\mathbf{P}^n$ is subject to the common noise $B_0$, which is felt by all the particles, and hence should stay in the limit; whereas the effect of the idiosyncratic noise from the independent Brownian motions $B_1,\ldots,B_n$ will be averaged way in the limit. The situation where $\mathbf{P}$ is not known to be $B_0$-measurable is dealt with separately in Remark \ref{rem:relaxed_soln} below. Once we have that $\mathbf{P}$ is $B_0$-measurable, retracing the proof of  \cite[Thm.~3.2]{LS18a}, as we did above, not only gives \eqref{limit_martingale}, but also shows that $\mathbf{P}=\mathrm{Law}(\mathbf{u},\mathbf{v},X \, | \, B_0)$. Therefore, letting $\bar{\mathbb{E}}$ denote the expectation operator corresponding to $\bar{\mathbb{P}}$, we have
\begin{align}\label{limit_loss_computation}
(\mathcal{L}_l(\mathbf{P}))(t) &= \bigl\langle \mathbf{P} , \pi_{1,l}(\cdot)  \mathbf{1}_{(\infty,0]}\bigl(\inf_{s\leq t}\pi_{3,t}(\cdot)\bigr)\bigr\rangle \nonumber
\\ &= \bar{\mathbb{E}} \bigl[  \mathbf{u}_l \mathbf{1}_{(-\infty, 0]}\bigl(\inf_{s\leq t}X_s \bigr)  \, | \, B_0 \bigr]
\nonumber \\
&= \bar{\mathbb{E}} \bigl[  \mathbf{u}_l \bar{\mathbb{P}} \bigl[  \inf_{s\leq t}X_s \leq 0 \mid  B_0,\mathbf{u},\mathbf{v},X(0) \bigr] \, \bigr| \, B_0 \bigr].
\end{align}
Since $(\mathbf{u},\mathbf{v},X(0))$ is independent of $(B,B_0)$, using the equation for $X$ in \eqref{limit_martingale} and the definition of $\mathcal{L}$ in  \eqref{loss_map}, we can conclude from \eqref{limit_loss_computation} that
\begin{align}\label{mean_field_limit_loss}
(\mathcal{L}_l(\mathbf{P}))(t) & = \int_{\mathbb{R}_+\times \mathbb{R}^k \times \mathbb{R}^k}  u_l \mathbb{P}( t\geq \tau^x_{u,v} \mid B_0) d\varpi(x,u,v), \qquad l=1,\ldots,k,
\end{align}
where $\varpi=\mathbf{P}_0$ is the limit of $\varpi^n$ given by \eqref{empirical_measures_parameters}, and where we have defined
\begin{align}\label{limit_X_xuv}
\tau^x_{u,v} &:= \inf \{ t\geq 0 : X^x_{u,v}(t) \leq 0 \}, \;\; \text{and} \nonumber \\
X^{x}_{u,v}(t) &:= x + \int_0^t \!b_{u,v}(s)ds + \int_0^t\!\sigma_{u,v}(s) d( \rho B_0(s) +\sqrt{1-\rho^2}B_i(s)   \bigr) \nonumber \\
&\qquad\qquad - F\Bigl( \sum_{l=1}^k v_l \int_0^t \!g(s)d(\mathcal{L}_l(\mathbf{P}))(s)  \Bigr),
\end{align}
for all realisations $(u,v)$ of $(\mathbf{u},\mathbf{v})$. Consequently, we have recovered the desired mean field limit \eqref{CMV}, since the limit point $(\mathbf{P},B^0)$ of $(\mathbf{P}^n,B^0)$ satisfies the conditional McKean--Vlasov problem \eqref{mean_field_limit_loss}--\eqref{limit_X_xuv}. Furthermore, as in \cite[Prop.~3.6]{LS18a} and the proof of \cite[Prop.~3.9]{LS18a}, the above tightness and continuity results, along with the expression \eqref{loss_empirical}, give that (in the M1 topology on $D_\mathbb{R}$), the loss processes $\mathcal{L}^n_l=\mathcal{L}_l(\mathbf{P}^n)$ converge to the desired limiting loss processes $\mathcal{L}_l=\mathcal{L}_l(\mathbf{P})$ satisfying the conditional McKean--Vlasov problem \eqref{mean_field_limit_loss}--\eqref{limit_X_xuv}.

\begin{remark}\label{rem:relaxed_soln}
	Without assuming $B_0$-measurability, we need to work with what is defined as a `relaxed' solution to \eqref{CMV} in \cite[Sect.~3]{LS18a}. Specifically, the arguments from \cite{LS18a} only gives that $\mathbf{P}=\mathrm{Law}(\mathbf{u},\mathbf{v},X\,|\,B_0,\mathbf{P} )$ with $(B_0,\mathbf{P}) \perp B$ and $(B,(B_0,\mathbf{P}))\perp X(0)$, as opposed to $\mathbf{P}=\mathrm{Law}(\mathbf{u},\mathbf{v},X\,|\,B_0)$ which we relied on above.  That is, $\mathbf{P}$ fulfils the first criteria for being the conditional law of $(\mathbf{u},\mathbf{v},X)$ given $B_0$, but it is only known to be $(B_0,\mathbf{P} )$-measurable, and hence it may not be the true conditional law given $B_0$. Yet, it behaves in almost the same way, since it is also independent of $B$, which is precisely what we expected to happen in the limit, as the idiosyncratic noise is averaged away and the common noise $B_0$ is independent of $B$. Repeating \eqref{limit_loss_computation} with $(\mathbf{P},B_0)$ in place of $B_0$, and using that $(\mathbf{u},\mathbf{v},X(0))$ is independent of $(B,(B_0,\mathbf{P}))$, we instead get
	\[
	(\mathcal{L}_l(\mathbf{P}))(t) = \int_{\mathbb{R}_+\times \mathbb{R}^k \times \mathbb{R}^k}  u_l \mathbb{P}( t\geq \tau^x_{u,v} \mid B_0, \mathbf{P}) d\varpi(x,u,v),
	\]
	so there is potentially some extra randomness that has survived the limiting procedure. In other words, we have mildly relaxed the criterion that the loss processes should strictly be conditional on the common noise $B_0$. For this reason, the limit thus obtained is called a `relaxed' solution to \eqref{CMV}. Nevertheless, in cases where we have a pathwise uniqueness argument for \eqref{CMV} such as in Section \ref{Proof_thm_uniq_common} (the proof of Theorem \ref{thm_weak_feedback}), we can apply a Yamada-Watanabe argument as in \cite[Thm.~2.3]{LS18b} to ensure that $\mathbf{P}$ really is $B_0$ measurable and that we are hence only conditioning on the common noise $B_0$.
\end{remark}

\bibliographystyle{plainnat}
\bibliography{bibtex2}

\begin{thebibliography}{80}
\providecommand{\natexlab}[1]{#1}
\providecommand{\url}[1]{\texttt{#1}}
\expandafter\ifx\csname urlstyle\endcsname\relax
  \providecommand{\doi}[1]{doi: #1}\else
  \providecommand{\doi}{doi: \begingroup \urlstyle{rm}\Url}\fi

\bibitem[Adrian and Brunnermeier(2016)]{adrian2011covar}
Tobias Adrian and Markus~K. Brunnermeier.
\newblock {CoVaR}.
\newblock \emph{American Economic Review}, 106\penalty0 (7):\penalty0
  1705--1741, 2016.

\bibitem[Amini and Minca(2016)]{amini_minca}
Hamed Amini and Andreea Minca.
\newblock Inhomogeneous financial networks and contagious links.
\newblock \emph{Operations Research}, 64\penalty0 (5):\penalty0 1053--1176,
  2016.

\bibitem[Amini et~al.(2012)Amini, Cont, and Minca]{ACM12}
Hamed Amini, Rama Cont, and Andreea Minca.
\newblock Stress testing the resilience of financial networks.
\newblock \emph{International Journal of Theoretical and Applied Finance},
  15\penalty0 (1):\penalty0 1250006, 2012.

\bibitem[Amini et~al.(2016)Amini, Cont, and Minca]{ACM10}
Hamed Amini, Rama Cont, and Andreea Minca.
\newblock Resilience to contagion in financial networks.
\newblock \emph{Mathematical Finance}, 24\penalty0 (2):\penalty0 329--365,
  2016.

\bibitem[Anand et~al.(2015)Anand, Craig, and Von~Peter]{ACP14}
Kartik Anand, Ben Craig, and Goetz Von~Peter.
\newblock Filling in the blanks: Network structure and interbank contagion.
\newblock \emph{Quantitative Finance}, 15\penalty0 (4):\penalty0 625--636,
  2015.

\bibitem[Banerjee and Feinstein(2019)]{BF18comonotonic}
Tathagata Banerjee and Zachary Feinstein.
\newblock Pricing of debt and equity in a financial network with comonotonic
  endowments.
\newblock 2019.
\newblock Working paper.

\bibitem[Banerjee et~al.(2018)Banerjee, Bernstein, and Feinstein]{BBF18}
Tathagata Banerjee, Alex Bernstein, and Zachary Feinstein.
\newblock Dynamic clearing and contagion in financial networks.
\newblock 2018.
\newblock Working paper.

\bibitem[Barucca et~al.(2016)Barucca, Bardoscia, Caccioli, D'Errico, Visentin,
  Battiston, and Caldarelli]{barucca2016valuation}
Paolo Barucca, Marco Bardoscia, Fabio Caccioli, Marco D'Errico, Gabriele
  Visentin, Stefano Battiston, and Guido Caldarelli.
\newblock Network valuation in financial systems.
\newblock 2016.
\newblock Working paper.

\bibitem[Battiston et~al.()Battiston, Gatti, Gallegati, Greenwald, and
  Stiglitz]{Battiston_2012}
S.~Battiston, D.~Delli Gatti, M.~Gallegati, B.C. Greenwald, and J.E. Stiglitz.
\newblock Liaisons dangereuses: Increasing connectivity, risk sharing, and
  systemic risk.
\newblock \emph{Journal of Economic Dynamics and Control}, 36\penalty0 (8).

\bibitem[Benazzoli et~al.()Benazzoli, Campi, and
  Persio]{Benazzoli_campi_persio}
Chiara Benazzoli, Luciano Campi, and Luca~Di Persio.
\newblock Mean-field games with controlled jump- diffusion dynamics: existence
  results and an interbank illiquid market model.
\newblock Working paper.

\bibitem[Biagini et~al.(2019{\natexlab{a}})Biagini, Fouque, Frittelli, and
  Meyer-Brandis]{biagini_fouque}
Francesca Biagini, Jean-Pierre Fouque, Marco Frittelli, and Thilo
  Meyer-Brandis.
\newblock A unified approach to systemic risk measures via acceptance sets.
\newblock \emph{Mathematical Finance}, 29\penalty0 (1):\penalty0 329--367,
  2019{\natexlab{a}}.

\bibitem[Biagini et~al.(2019{\natexlab{b}})Biagini, Fouque, Frittelli, and
  Meyer-Brandis]{fouque2015systemic}
Francesca Biagini, Jean-Pierre Fouque, Marco Frittelli, and Thilo
  Meyer-Brandis.
\newblock A unified approach to systemic risk measures via acceptance sets.
\newblock \emph{Mathematical Finance}, 29\penalty0 (1):\penalty0 329--367,
  2019{\natexlab{b}}.

\bibitem[Bo and Capponi(2018)]{bo_capponi}
Lijun Bo and Agostino Capponi.
\newblock Systemic risk in interbanking networks.
\newblock \emph{SIAM Journal on Financial Mathematics}, 6\penalty0
  (1):\penalty0 386--424, 2018.

\bibitem[Borgatti and Everett(2000)]{BE00}
Stephen~P. Borgatti and Martin~G. Everett.
\newblock Models of core/periphery structures.
\newblock \emph{Social Networks}, 21\penalty0 (4):\penalty0 375--395, 2000.

\bibitem[Borovykh et~al.(2018)Borovykh, Pascucci, and Rovere]{pascucci_sysrisk}
Anastasia Borovykh, Andrea Pascucci, and Stefano~La Rovere.
\newblock Systemic risk in a mean-field model of interbank lending with
  self-exciting shocks.
\newblock \emph{IISE Transactions}, 50\penalty0 (9):\penalty0 806--819, 2018.

\bibitem[Braouezec and Wagalath(2019)]{BW19}
Yann Braouezec and Lakshithe Wagalath.
\newblock Strategic fire-sales and price-mediated contagion in the banking
  system.
\newblock \emph{European Journal of Operational Research}, 274\penalty0
  (3):\penalty0 1180--1197, 2019.

\bibitem[Campi and Fischer(2018)]{campi_mfg_1}
Luciano Campi and Markus Fischer.
\newblock N-player games and mean-field games with absorption.
\newblock \emph{The Annals of Applied Probability}, 28\penalty0 (4):\penalty0
  2188--2242, 2018.

\bibitem[Campi et~al.()Campi, Ghio, and Livieri]{campi_mfg_2}
Luciano Campi, Maddalena Ghio, and Giulia Livieri.
\newblock N-player games and mean-field games with smooth dependence on past
  absorptions.
\newblock Working paper.

\bibitem[Capponi and Chen(2015)]{CC15}
Agostino Capponi and Peng-Chu Chen.
\newblock Systemic risk mitigation in financial networks.
\newblock \emph{Journal of Economic Dynamics and Control}, 58:\penalty0
  152--166, 2015.

\bibitem[Capponi et~al.(2019)Capponi, Sun, and Yao]{capponi_clusters}
Agostino Capponi, Xu~Sun, and David Yao.
\newblock A dynamic network model of interbank lending {---} systemic risk and
  liquidity provisioning.
\newblock 2019.
\newblock To appear in Mathematics of Operations Research.

\bibitem[Carmona et~al.(2015)Carmona, Fouque, and Sun]{fouque2015meanfield}
Ren\'{e} Carmona, Jean-Pierre Fouque, and Li-Hsien Sun.
\newblock Mean field games and systemic risk.
\newblock \emph{Communications in Mathematical Sciences}, 13\penalty0
  (4):\penalty0 911--933, 2015.

\bibitem[Carmona et~al.(2018)Carmona, Fouque, Mousavi, and
  Sun]{carmona2016delay}
Ren{\'e} Carmona, Jean-Pierre Fouque, Seyyed~Mostafa Mousavi, and Li-Hsien Sun.
\newblock Systemic risk and stochastic games with delay.
\newblock \emph{Journal of Optimization Theory and Applications}, Mar 2018.
\newblock ISSN 1573-2878.
\newblock \doi{10.1007/s10957-018-1267-8}.
\newblock URL \url{https://doi.org/10.1007/s10957-018-1267-8}.

\bibitem[Chen et~al.(2013)Chen, Iyengar, and Moallemi]{chen2013axiomatic}
Chen Chen, Garud Iyengar, and Ciamac~C. Moallemi.
\newblock An axiomatic approach to systemic risk.
\newblock \emph{Management Science}, 59\penalty0 (6):\penalty0 1373--1388,
  2013.

\bibitem[Chong and Kluppelberg()]{chong_kluppelberg}
Carsten Chong and Claudia Kluppelberg.
\newblock Partial mean field limits in heterogeneous networks.
\newblock \emph{To appear in Stochastic Processes and their Applications},
  129\penalty0 (12):\penalty0 4998--5036.

\bibitem[Cifuentes et~al.(2005)Cifuentes, Shin, and Ferrucci]{CFS05}
Rodrigo Cifuentes, Hyun~Song Shin, and Gianluigi Ferrucci.
\newblock Liquidity risk and contagion.
\newblock \emph{Journal of the European Economic Association}, 3\penalty0
  (2-3):\penalty0 556--566, 2005.

\bibitem[Cont and Schaanning(2019)]{cont_schaanning}
Rama Cont and Eric Schaanning.
\newblock Monitoring indirect contagion.
\newblock \emph{Journal of Banking and Finance}, 104:\penalty0 85--102, 2019.

\bibitem[Craig and Von~Peter(2014)]{CP14}
Ben Craig and Goetz Von~Peter.
\newblock Interbank tiering and money center banks.
\newblock \emph{Journal of Financial Intermediation}, 23\penalty0 (3):\penalty0
  322--347, 2014.

\bibitem[Cucuringu et~al.(2016)Cucuringu, Rombach, Lee, and
  Porter]{porter_2016}
Mihar Cucuringu, Puck Rombach, Sang~Hoon Lee, and Maons~A. Porter.
\newblock Detection of core-periphery structure in networks using spectral
  methods and geodesic paths.
\newblock \emph{European Journal of Applied Mathematics}, 27\penalty0
  (6):\penalty0 846--887, 2016.

\bibitem[Delarue et~al.(2019)Delarue, Nadtochiy, and Shkolnikov]{DNS19}
Fran\c{c}ois Delarue, Sergey Nadtochiy, and Mykhaylo Shkolnikov.
\newblock Global solutions to the supercooled {Stefan} problem with blow-ups:
  regularity and uniqueness.
\newblock 2019.
\newblock Working paper.

\bibitem[Detering et~al.(2019)Detering, Meyer-Brandis, Panagiotou, and
  Ritter]{detering2019inhomogeneous}
Nils Detering, Thilo Meyer-Brandis, Konstantinos Panagiotou, and Daniel Ritter.
\newblock Managing default contagion in inhomogeneous financial networks.
\newblock \emph{SIAM Journal on Financial Mathematics}, 10\penalty0
  (2):\penalty0 578--614, 2019.

\bibitem[Eisenberg and Noe(2001)]{EN01}
Larry Eisenberg and Thomas~H. Noe.
\newblock Systemic risk in financial systems.
\newblock \emph{Management Science}, 47\penalty0 (2):\penalty0 236--249, 2001.

\bibitem[Elsinger(2009)]{E07}
Helmut Elsinger.
\newblock Financial networks, cross holdings, and limited liability.
\newblock \emph{{\"{O}sterrei}chische Nationalbank (Austrian Central Bank)},
  156, 2009.

\bibitem[Fang et~al.(2017)Fang, Spiliopoulos, and Sun]{spilio_hetero}
Fei Fang, Konstantinos Spiliopoulos, and Yiwei Sun.
\newblock The effect of heterogeneity on flocking behavior and systemic risk.
\newblock \emph{Statistics and Risk Modelling}, 34\penalty0 (3-4):\penalty0
  141--155, 2017.

\bibitem[Feinstein(2019)]{feinstein2017currency}
Zachary Feinstein.
\newblock Obligations with physical delivery in a multi-layered financial
  network.
\newblock \emph{SIAM Journal on Financial Mathematics}, 10\penalty0
  (4):\penalty0 877--906, 2019.

\bibitem[Feinstein(2020)]{feinstein2019leverage}
Zachary Feinstein.
\newblock Capital regulation under price impacts and dynamic financial
  contagion.
\newblock \emph{European Journal of Operational Research}, 281\penalty0
  (2):\penalty0 449--463, 2020.

\bibitem[Feinstein et~al.(2017)Feinstein, Rudloff, and
  Weber]{feinstein2014measures}
Zachary Feinstein, Birgit Rudloff, and Stefan Weber.
\newblock Measures of systemic risk.
\newblock \emph{SIAM Journal on Financial Mathematics}, 8\penalty0
  (1):\penalty0 672--708, 2017.

\bibitem[Feinstein et~al.(2018)Feinstein, Pang, Rudloff, Schaanning, Sturm, and
  Wildman]{feinstein2017sensitivity}
Zachary Feinstein, Weijie Pang, Birgit Rudloff, Eric Schaanning, Stephan Sturm,
  and Mackenzie Wildman.
\newblock Sensitivity of the {E}isenberg--{N}oe clearing vector to individual
  interbank liabilities.
\newblock \emph{SIAM Journal on Financial Mathematics}, 9\penalty0
  (4):\penalty0 1286--1325, 2018.

\bibitem[Ferrara et~al.(2016)Ferrara, Langfield, Liu, and Ota]{ferrara16}
Gerardo Ferrara, Sam Langfield, Zijun Liu, and Tomohiro Ota.
\newblock Systemic illiquidity in the interbank network.
\newblock Staff Working Paper 586, Bank of England, 2016.

\bibitem[Fouque and Ichiba(2013)]{fouque2013stability}
Jean-Pierre Fouque and Tomoyuki Ichiba.
\newblock Stability in a model of interbank lending.
\newblock \emph{SIAM Journal on Financial Mathematics}, 4:\penalty0 784--803,
  2013.

\bibitem[Fouque and Sun(2013)]{fouque2013illustrated}
Jean-Pierre Fouque and Li-Hsien Sun.
\newblock Systemic risk illustrated.
\newblock In \emph{Handbook on Systemic Risk}, pages 444--452. Cambridge
  University Press, 2013.

\bibitem[Fouque and Zhang(2018)]{fouque_zhang}
Jean-Pierre Fouque and Zhaoyu Zhang.
\newblock Mean field game with delay: a toy model.
\newblock \emph{Risks}, 6\penalty0 (90):\penalty0 1--17, 2018.

\bibitem[Fricke and Lux(2015)]{FL15}
Daniel Fricke and Thomas Lux.
\newblock Core-periphery structure in the overnight money market: evidence from
  the e-mid trading platform.
\newblock \emph{Computational Economics}, 45\penalty0 (3):\penalty0 359--395,
  2015.

\bibitem[Gai and Kapadia(2010)]{GK10}
Prasanna Gai and Sujit Kapadia.
\newblock Contagion in financial networks.
\newblock \emph{Proceedings of the Royal Society A: Mathematical, Physical and
  Engineering Sciences}, 466\penalty0 (2120):\penalty0 2401--2423, 2010.

\bibitem[Gandy and Veraart(2016)]{GV16}
Axel Gandy and Luitgard~A.M. Veraart.
\newblock A {B}ayesian methodology for systemic risk assessment in financial
  networks.
\newblock \emph{Management Science}, 63\penalty0 (12):\penalty0 4428--4446,
  2016.
\newblock \doi{{DOI: 10.1287/mnsc.2016.2546}}.

\bibitem[Garnier et~al.(2013)Garnier, Papanicolaou, and Yang]{garnier2013b}
Josselin Garnier, George Papanicolaou, and Tzu-Wei Yang.
\newblock Large deviations for a mean field model of systemic risk.
\newblock \emph{SIAM Journal on Financial Mathematics}, 4:\penalty0 151--184,
  2013.

\bibitem[Garnier et~al.(2017)Garnier, Papanicolaou, and Yang]{garnier2017}
Josselin Garnier, George Papanicolaou, and Tzu-Wei Yang.
\newblock A risk analysis for a system stabilized by a central agent.
\newblock \emph{Risk and Decision Analysis}, 6\penalty0 (2):\penalty0 97--120,
  2017.

\bibitem[Giesecke et~al.(2013)Giesecke, Spiliopoulos, and
  Sowers]{giesecke_2013}
Kay Giesecke, Konstantinos Spiliopoulos, and Richard Sowers.
\newblock Default clustering in large portfolios: Typical events.
\newblock \emph{The Annals of Applied Probability}, 23\penalty0 (1):\penalty0
  348--385, 2013.

\bibitem[Giesecke et~al.(2015)Giesecke, Spiliopoulos, Sowers, and
  Sirignano]{giesecke_2015}
Kay Giesecke, Konstantinos Spiliopoulos, Richard Sowers, and Justin Sirignano.
\newblock Large portfolio asymptotics for loss from default.
\newblock \emph{Mathematical Finance}, 25\penalty0 (1):\penalty0 77--114, 2015.

\bibitem[Glasserman and Young(2015)]{GY14}
Paul Glasserman and H.~Peyton Young.
\newblock How likely is contagion in financial networks?
\newblock \emph{Journal of Banking and Finance}, 50:\penalty0 383--399, 2015.

\bibitem[Hambly and S{\o}jmark(2019)]{HS18}
Ben Hambly and Andreas S{\o}jmark.
\newblock An {SPDE} model for systemic risk with endogenous contagion.
\newblock \emph{Finance and Stochastics}, 23\penalty0 (3):\penalty0 535--594,
  2019.

\bibitem[Hambly et~al.(2019)Hambly, Ledger, and S{\o}jmark]{HLS18}
Ben Hambly, Sean Ledger, and Andreas S{\o}jmark.
\newblock A {McKean--Vlasov} equation with positive feedback and blow-ups.
\newblock \emph{The Annals of Applied Probability}, 29\penalty0 (4):\penalty0
  2338--2373, 2019.

\bibitem[Huang and Jaimungal(2017)]{huang_jaimungal}
Xuancheng Huang and Sebastian Jaimungal.
\newblock Robust stochastic games and systemic risk.
\newblock 2017.
\newblock Working paper.

\bibitem[Hurd(2016)]{hurd2016}
T.R. Hurd.
\newblock \emph{Contagion! {S}ystemic Risk in Financial Networks}.
\newblock SpringerBriefs in Quantitative Finance. Springer International
  Publishing, 2016.
\newblock ISBN 9783319339306.

\bibitem[Ichiba et~al.(2019)Ichiba, Ludkovski, and Sarantsev]{ichiba_2018}
Tomoyuki Ichiba, Michael Ludkovski, and Andrey Sarantsev.
\newblock Dynamic contagion in a banking system with births and defaults.
\newblock \emph{Annals of Finance}, 15\penalty0 (4):\penalty0 489--538, 2019.

\bibitem[Itkin and Lipton(2017)]{Lipton2015}
Andrey Itkin and Alexander Lipton.
\newblock Structural default model with mutual obligations.
\newblock \emph{Review of Derivatives Research}, 20\penalty0 (1):\penalty0
  15--46, 2017.

\bibitem[Kaushansky and Reisinger(2019)]{KR18}
Vadim Kaushansky and Christoph Reisinger.
\newblock Simulation of particle systems interacting through hitting times.
\newblock \emph{Discrete and Continuous Dynamical Systems -- Series B},
  24\penalty0 (10):\penalty0 5481--5502, 2019.

\bibitem[Kaushansky et~al.(2018{\natexlab{a}})Kaushansky, Lipton, and
  Reisinger]{KLR18a}
Vadim Kaushansky, Alexander Lipton, and Christoph Reisinger.
\newblock Numerical analysis of an extended structural default model with
  mutual liabilities and jump risk.
\newblock \emph{Journal of Computational Science}, 24:\penalty0 218--231,
  2018{\natexlab{a}}.

\bibitem[Kaushansky et~al.(2018{\natexlab{b}})Kaushansky, Lipton, and
  Reisinger]{KLR18b}
Vadim Kaushansky, Alexander Lipton, and Christoph Reisinger.
\newblock Transition probability of brownian motion in the octant and its
  application to default modelling.
\newblock \emph{Applied Mathematical Finance}, 2018{\natexlab{b}}.

\bibitem[Kaushansky et~al.(2018{\natexlab{c}})Kaushansky, Lipton, and
  Reisinger]{KLR18c}
Vadim Kaushansky, Alexander Lipton, and Christoph Reisinger.
\newblock Semi-analytical solution of a {McKean-Vlasov} equation with feedback
  through hitting a boundary.
\newblock 2018{\natexlab{c}}.
\newblock Working paper.

\bibitem[Kromer et~al.(2016)Kromer, Overbeck, and Zilch]{kromer2013systemic}
Eduard Kromer, Ludger Overbeck, and Katrin Zilch.
\newblock Systemic risk measures on general probability spaces.
\newblock \emph{Mathematical Methods of Operations Research}, 84\penalty0
  (2):\penalty0 323--357, 2016.

\bibitem[Kusnetsov and Veraart(2018)]{KV16}
Michael Kusnetsov and Luitgard~A.M. Veraart.
\newblock Interbank clearing in financial networks with multiple maturities.
\newblock 2018.
\newblock Working paper.

\bibitem[Ledger and S{\o}jmark(2018{\natexlab{a}})]{LS18a}
Sean Ledger and Andreas S{\o}jmark.
\newblock At the mercy of the common noise: Blow-ups in a conditional
  {McKean--Vlasov} problem.
\newblock 2018{\natexlab{a}}.
\newblock Working paper.

\bibitem[Ledger and S{\o}jmark(2018{\natexlab{b}})]{LS18b}
Sean Ledger and Andreas S{\o}jmark.
\newblock Uniqueness for contagious {McKean--Vlasov} systems in the weak
  feedback regime.
\newblock 2018{\natexlab{b}}.
\newblock Working paper.

\bibitem[Lipton(2016)]{Lipton2016}
Alexander Lipton.
\newblock Modern monetary circuit theory, stability of interconnected banking
  network, and balance sheet optimization for individual banks.
\newblock \emph{International Journal of Theoretical and Applied Finance},
  19\penalty0 (6):\penalty0 1--57, 2016.

\bibitem[Liu and Staum(2010)]{LS10}
Ming Liu and Jeremy Staum.
\newblock Sensitivity analysis of the {E}isenberg-{N}oe model of contagion.
\newblock \emph{Operations Research Letters}, 35\penalty0 (5):\penalty0
  489--491, 2010.

\bibitem[Maheshwari and Sarantsev(2017)]{sarantsev_maheshwari}
Aditya Maheshwari and Andrey Sarantsev.
\newblock A model of interbank flows, borrowing, and investing.
\newblock 2017.
\newblock Working paper.

\bibitem[Mastromatteo et~al.(2012)Mastromatteo, Zarinelli, and Marsili]{MZM12}
Iacopo Mastromatteo, Elia Zarinelli, and Matteo Marsili.
\newblock Reconstruction of financial networks for robust estimation of
  systemic risk.
\newblock \emph{Journal of Statistical Mechanics: Theory and Experiment},
  2012\penalty0 (3):\penalty0 P03011, 2012.

\bibitem[Nadtochiy and Shkolnikov(2019{\natexlab{a}})]{NS17}
Sergey Nadtochiy and Mykhaylo Shkolnikov.
\newblock Particle systems with singular interaction through hitting times:
  Application in systemic risk modeling.
\newblock \emph{The Annals of Applied Probability}, 29\penalty0 (1):\penalty0
  89--129, 2019{\natexlab{a}}.

\bibitem[Nadtochiy and Shkolnikov(2019{\natexlab{b}})]{NS18}
Sergey Nadtochiy and Mykhaylo Shkolnikov.
\newblock Mean field systems on networks, with singular interaction through
  hitting times.
\newblock 2019{\natexlab{b}}.
\newblock To appear in The Annals of Probability.

\bibitem[Rogers and Veraart(2013)]{RV13}
Leonard~C.G. Rogers and Luitgard~A.M. Veraart.
\newblock Failure and rescue in an interbank network.
\newblock \emph{Management Science}, 59\penalty0 (4):\penalty0 882--898, 2013.

\bibitem[Shkolnikov and Ichiba(2013)]{misha_ichiba}
Mykhaylo Shkolnikov and Tomoyuki Ichiba.
\newblock Large deviations for interacting {Bessel-like} processes and
  applications to systemic risk.
\newblock 2013.
\newblock Working paper.

\bibitem[S{\o}jmark(2019)]{sojmark_2019}
Andreas S{\o}jmark.
\newblock Contagious {McKean--Vlasov} systems with positive feedback.
\newblock 2019.
\newblock DPhil Thesis, University of Oxford.

\bibitem[Sonin and Sonin(2017)]{sonin2017}
Isaac~M.\ Sonin and Konstantin Sonin.
\newblock Banks as tanks: A continuous-time model of financial clearing.
\newblock 2017.
\newblock Working paper.

\bibitem[Spiliopoulos and Yang(2019)]{spiliopoulos_2019}
Konstantinos Spiliopoulos and Jia Yang.
\newblock Network effects in default clustering for large systems.
\newblock 2019.
\newblock Working paper.

\bibitem[Stroock(2011)]{stroock_integration}
Daniel~W. Stroock.
\newblock \emph{Essentials of integration theory for analysis}, volume 262 of
  \emph{Graduate Texts in Mathematics}.
\newblock Springer, New York, 2011.

\bibitem[Sun(2018)]{sun_interbank}
Li-Hsien Sun.
\newblock Systemic risk and interbank lending.
\newblock \emph{Journal of Optimization Theory and Applications}, 179\penalty0
  (2):\penalty0 400--424, 2018.

\bibitem[Sznitman(1991)]{sznitman_1991}
Alain-Sol Sznitman.
\newblock Topics in propagation of chaos.
\newblock In \emph{\'{E}cole d'\'{E}t\'{e} de {P}robabilit\'{e}s de
  {S}aint-{F}lour {XIX}---1989}, volume 1464 of \emph{Lecture Notes in Math.},
  pages 165--251. Springer, Berlin, 1991.

\bibitem[Upper and Worms(2004)]{UW04}
Christian Upper and Andreas Worms.
\newblock Estimating bilateral exposures in the {G}erman interbank market: {I}s
  there a danger of contagion?
\newblock \emph{European Economic Review}, 48\penalty0 (4):\penalty0 827--849,
  2004.

\bibitem[Weber and Weske(2017)]{AW_15}
Stefan Weber and Kerstin Weske.
\newblock The joint impact of bankruptcy costs, fire sales and cross-holdings
  on systemic risk in financial networks.
\newblock \emph{Probability, Uncertainty and Quantitative Risk}, 2\penalty0
  (1):\penalty0 9, 2017.

\bibitem[Whitt(2002)]{whitt_2002}
Ward Whitt.
\newblock \emph{Stochastic-process limits}.
\newblock Springer Series in Operations Research. Springer-Verlag, New York,
  2002.
\newblock An introduction to stochastic-process limits and their application to
  queues.

\end{thebibliography}

\end{document}